\documentclass[10pt,journal,compsoc]{IEEEtran}

\usepackage{balance}  
\usepackage{graphicx} 
\usepackage{times}    
\usepackage{url}      
\usepackage{tabularx}
\usepackage[safe]{tipa}
\usepackage{amsmath}
\usepackage{amssymb}
\usepackage{amsthm}
\usepackage[boxed]{algorithm}
\usepackage{algcompatible}

\usepackage{array}
\usepackage{eqparbox}

\usepackage{subfigure}
\usepackage{multirow}

\newcommand{\eat}[1]{}
\usepackage[font=small,skip=0pt]{caption}
\makeatletter
\def\url@leostyle{%
  \@ifundefined{selectfont}{\def\UrlFont{\sf}}{\def\UrlFont{\small\bf\ttfamily}}}
\makeatother
\newtheorem{theorem}{Theorem}[section]

\makeatletter
\algrenewcommand\ALG@beginalgorithmic{\footnotesize}
\makeatother

\newtheorem{lemma}[theorem]{Lemma}
\newtheorem{proposition}[theorem]{Proposition}
\newtheorem{corollary}[theorem]{Corollary}

\newtheorem{definition}{Definition}[section]

\eat{
\usepackage{amsmath,amsfonts,amssymb}

\newtheorem{theorem}{Theorem}[section]

}

\def\bdd{\text{\texthtd} }
\def\pprw{8.5in}
\def\pprh{11in}

\usepackage[pdftex,bookmarks=false]{hyperref}
\hypersetup{
pdfauthor={LaTeX},
pdfstartview={FitH},
colorlinks,
citecolor=black,
filecolor=black,
linkcolor=black,
urlcolor=black,
breaklinks=true,
}

\ifCLASSOPTIONcompsoc
  \usepackage[nocompress]{cite}
\else
  \usepackage{cite}
\fi

\begin{document}

\title{A Novel Framework for Online Amnesic \\Trajectory Compression in Resource-constrained Environments}


\author{
 $^{\circ\dag}$Jiajun Liu \hspace{0.8cm} $^\dag$Kun Zhao \hspace{0.8cm} $^\dag$Philipp Sommer\\
 $^*$Shuo Shang \hspace{0.8cm} $^\dag$Brano Kusy \hspace{0.8cm} $^\diamond$Jae-Gil Lee\hspace{0.8cm} 
 $^\dag$Raja Jurdak\\
 {$ $}\\
 {$^\circ$Renmin University of China, Beijing, China \{jiajunliu@ruc.edu.cn\}}\\
{$^\dag$Data 61, CSIRO, Pullenvale, Australia}\\
{ \{jiajun.liu, kun.zhao, philipp.sommer, brano.kusy, raja.jurdak\}@csiro.au}\\
{$^*$China University of Petroleum, Beijing, China \{sshang@cup.edu.cn\}}\\
{$^\diamond$Department of Knowledge Service Engineering, KAIST, Korea \{jaegil@kaist.ac.kr\}}\\
\IEEEcompsocitemizethanks{\IEEEcompsocthanksitem This article is an extended version of \cite{DBLP:conf/icde/LiuZSSKJ15}}
}

\IEEEtitleabstractindextext{
\begin{abstract}
\eat{Long-term location tracking, where trajectory compression is commonly
used, has gained high interest for many
applications in transport, ecology, and wearable computing. However, s}
State-of-the-art
trajectory compression methods usually involve high space-time complexity or
yield unsatisfactory compression rates, leading to rapid
exhaustion of memory, computation, storage and energy resources. 
Their ability is commonly limited when operating in a resource-constrained environment especially
 when the data volume (even when compressed) far exceeds the storage limit.
Hence we propose a novel online framework for error-bounded trajectory
compression and ageing called the \emph{Amnesic Bounded Quadrant System} (ABQS), whose core
is the \emph{Bounded Quadrant System} (BQS) algorithm family that includes a normal version (BQS), Fast version (FBQS), and a Progressive version (PBQS). ABQS intelligently manages a given storage and compresses the trajectories
with different error tolerances subject to their ages.\eat{BQS compresses
trajectories with extremely small costs in space and time using
convex-hulls, and performs further compression on aged data automatically when needed. 
In the BQS algorithm family, we build a virtual coordinate system
centered at a start point, and establish a rectangular bounding box as
well as two bounding lines in each of its quadrants. In each quadrant, 
the points to be assessed are bounded by the convex-hull formed by the
box and lines. Various compression error-bounds are therefore derived to 
quickly draw compression decisions without expensive error computations. 
We also propose a The light version of BQS (FBQS)
achieves $\mathcal{O}(1)$ complexity in both time and space for
processing each point to suit the most constrained computation
environments. \eat{Then we describe in brief a 3D extension of the algorithm.}
Furthermore, we propose an online
framework that intelligently manages a given storage and compresses the trajectories
with different error tolerances subject to their ages.
}

In the experiments, we conduct comprehensive evaluations for the BQS algorithm family and the ABQS framework.
Using empirical GPS traces from flying foxes and cars, and synthetic data from simulation, 
we demonstrate the effectiveness of the standalone BQS algorithms in significantly
reducing the time and space complexity of trajectory compression, while
greatly improving the compression rates of the state-of-the-art
algorithms (up to 45\%). We also show that the operational time of the target 
resource-constrained hardware platform 
can be prolonged by up to 41\%. We then verify that with ABQS,
given data volumes that are far greater than storage space, ABQS is able to achieve 15 to 400 times 
smaller errors than the baselines. We also show that the algorithm is robust to extreme trajectory shapes.
\end{abstract}

\begin{IEEEkeywords}
Online trajectory compression; ageing; amnesic; resource-constrained; constant complexity
\end{IEEEkeywords}

}

\eat{
\keywords{
	Guides; instructions; author's kit; conference publications;
	keywords should be separated by a semi-colon.
}

\category{H.5.m.}{Information Interfaces and Presentation (e.g. HCI)}{Miscellaneous}
}

\maketitle

\section{Introduction}
Location tracking is increasingly important for transport, ecology, and wearable computing. In particular, long-term tracking of spatially spread assets, such as wildlife~\cite{Jurdak_tosn13}, augmented reality glasses~\cite{van2010survey}, or bicycles~\cite{DBLP:journals/tosn/EisenmanMLPAC09} provides high resolution trajectory information for better management and services.  For smaller moving entities, such as flying foxes~\cite{DBLP:conf/ipsn/JurdakSKKCMW13} or pigeons~\cite{nagy2010hierarchical}, the size and weight of tracking devices are constrained, which presents the challenge of obtaining detailed trajectory information subject to memory, processing, energy and storage constraints. However, the requirements on the tracking precision are barely relaxed, e.g. the goal of wildlife tracking sometimes is in the form of a guaranteed tracking error from 10 to a few hundred meters \cite{DBLP:conf/sensys/JurdakCDS10}.

Consider tracking of flying foxes as a motivating scenario. The computing platform is constrained in computational resources and is inaccessible in most occasions once deployed. The position data is acquired in a streaming fashion. The RAM available on the platform is 4 KBytes, while the storage space is only 1 MB to store trajectories over weeks and months before they can be offloaded. The combination of long-term operational requirements and constrained resources therefore 
requires an intelligent online framework that can process the trajectory stream instantaneously and efficiently, i.e. in constant space and time, and that can achieve high compression rate. 

Current trajectory compression algorithms often fail to operate under such requirements, as they either require substantial amount of buffer space or require the entire data stream to perform the compression \cite{douglas_peucker}\cite{Hershberger92speedingup}. Existing online algorithms, which process each point exactly once,  operate retrospectively on trajectory data or assume favourable trajectory characteristics, resulting in the worst-case complexity of their online performance ranging from $\mathcal{O}(nlogn)$ to $\mathcal{O}(n^2)$~\cite{squishe}. The high complexity of these methods limits their utility in resource-constrained environments. To address this challenge, we propose the \emph{Bounded Quadrant System} (BQS), an online algorithm that can run sustainably on resource-constrained devices. Its fast version achieves $\mathcal{O}(n)$ time and $\mathcal{O}(1)$ space complexity while providing guaranteed error bounds. By using a convex hull that is formed by a bounding box and two angular bounds around all points, the algorithm is able to make quick compression decisions without iterating through the buffer and calculating the maximum error in most of the cases. Using empirical GPS traces from the flying fox scenario and from cars, we evaluate the performance of our algorithms and quantify their benefits in improving the efficiency of trajectory compression. Its amnesic version, the  \emph{Amnesic Bounded Quadrant System} (ABQS), is able to manage a given volume of storage space and intelligently trade off space for precision, so that no hard data losses (data overwritten) will be present over a prolonged operation time.

Our contributions are \emph{threefold}:

\begin{enumerate}
\small
\vspace{-0.2cm}
\item We design an online compression algorithm family called BQS.  BQS uses convex hulls to compress streaming trajectories with error guarantees. The fast version in the algorithm family achieves constant time and space complexity for each step, or equivalently $\mathcal{O}(n)$ time complexity and $\mathcal{O}(1)$ space complexity for the whole data stream. \eat{We demonstrate the extensibility of BQS to support the 3-D case and a different error metric.}
\item We formulate the BQS algorithm with an uncertainty vector to support the progressive compression of trajectories, and subsequently propose \emph{Progressive BQS} (PBQS). Using PBQS as the corner stone, we propose a sophisticated online framework called  \emph{Amnesic Bounded Quadrant System} (ABQS) which not only compresses streaming trajectory data efficiently, but also manages historical data in an amnesic way. ABQS introduces graceful degradation to the entire trajectory history without hard data loss (overwriting), and hence achieves excellent compression performances when the operation duration is unknown and the data volume far exceeds the storage limit.
\item We comprehensively evaluate the BQS family and the ABQS framework using real-life data collected from wildlife tracking and vehicle tracking applications as well as synthetic data, and discuss the robustness to various types of trajectories with extreme shapes.
\end{enumerate}
Compared to the original BQS paper \cite{DBLP:conf/icde/LiuZSSKJ15}, this paper makes significant extensions and technical contributions as it: 
1) adds uncertainty to the BQS formulation and derives the progressive version of BQS, i.e. PBQS; 2) proposes the ABQS framework so that BQS is no longer only a family of standalone algorithms but an actual framework that can effectively and efficiently manage a storage space and maximize the trajectory information given a storage limit and an unknown operation duration; 3) significantly extends the experiments to study the robustness of BQS and the compression performances of ABQS. In addition, this version also extends the literature review significantly, and restructures and improves Sections \ref{sec:pre} and \ref{sec:bqs} for better clarity of presentation.

The remainder of this paper is organized as follows. We survey related work in the next section, and  discuss the background and motivate the need for a new online trajectory compression framework by describing our hardware platform and the data acquisition process in Section \ref{sec:mot}. The BQS family is presented subsequently in Section \ref{sec:bqs}, where a discussion on how to generalize the algorithm is also provided. We then propose the ABQS framework in Section \ref{sec:abqs}. Finally, we evaluate the proposed BQS family and the ABQS framework in Section \ref{sec:exp}, and conclude the paper.

\section{Related Work}
\label{sec:rl}
The rapid increase in the number of GPS-enabled devices in recent years has led to an expansion of location-based services and applications and our increased reliance on localization technology. One challenge that location-based applications need to overcome is the amount of data that continuous location tracking can yield. Efficient storage and indexing of these datasets is critical to their success, especially for embedded and mobile devices which have restricted energy, computational power and storage.

Several trajectory compression algorithms that offer significant improvements in terms of data storage have been proposed in the literature. We focus our review on lossy compression algorithms as they provide better trade-offs between high compression ratios and an acceptable error of the compressed trajectory.

Douglas and Peucker were amongst the first to propose an algorithm for reducing the number of points in a digital trajectory~\cite{douglas_peucker}. The \emph{Douglas-Peucker} algorithm starts with the first and last points of the trajectory and repeatedly adds intermediate points to the compressed trajectory until the maximum spatial error of any point falls bellow a predefined tolerance. The algorithm guarantees that the error of the compressed trajectory is within the bounds of the target application, but due to its greedy nature, it achieves relatively low compression ratios. The worst-case runtime of the algorithm is $\mathcal{O}(n^2)$ with $n$ being the number of points in the original trajectory which has been improved by Hershberger et al to $\mathcal{O}(n\log n)$~\cite{Hershberger92speedingup}.

The disadvantage of the \emph{Douglas-Peucker} algorithm is that it runs off-line and requires the whole trajectory. This is of limited use for modern location-aware applications that require online processing of location data. A generic sliding-window algorithm (conceptually similar to the one summarized in~\cite{opening_window}) is often used to overcome this limitation and works by compressing the original trajectory over a moving buffer. The sliding-window algorithm can run online, but its worst-case runtime is still $\mathcal{O}(nL)$ where $L$ is the maximum buffer size. On the other hand, multiple examples of fast algorithms exist in the literature~\cite{ChenXF12}\cite{Long_2013}\cite{DBLP:journals/pvldb/SongSZZ14}\cite{DBLP:journals/tkde/XieYCL14}. These algorithms, however, do not apply to our scenario as they only run off-line and cannot support location-aware applications in-situ.

SQUISH~\cite{squish} has achieved relatively low runtime, high compression ratios and small trajectory errors. However, its disadvantage was that it could not guarantee trajectory errors to be within an application-specific bound. A follow up work presented SQUISH-E~\cite{squishe} that provides options to both minimize trajectory error given a compression ratio and to maximize compression ratio given an error bound. The worst-case runtime of SQUISH-E algorithms is $\mathcal{O}(n\log\frac{n}{\lambda})$, where $\lambda$ is the desired compression ratio.  While the compression ratio-bound flavor of SQUISH-E can run online, the error-bound version runs offline only.

There are different disadvantages for existing algorithms that repeatedly iterate through all points in the original trajectory. SQUISH-E is approaching linear computational complexity for large compression ratios, however, the compressed trajectory has unbounded error. The STTrace algorithm \cite{sttrace} and the MBR algorithm~\cite{liu_MBR} represent more complex algorithms. STTrace~\cite{sttrace} uses estimation of speed and heading to predict the position of the next location. MBR maintains, divides, and merges bounding rectangles that represent the original trajectories. However both algorithsm fall outside of capabilities of our target hardware platform and do not suit our application scenario well. 

In contrast to existing methods, our approach achieves constant time and space complexity for each point by only considering the most recent minimal bounding convex-hulls. We show in the evaluation section that the compression ratios that our approach achieves are superior to those of the related trajectory compression algorithms. Although simplistic approaches such as Dead Reckoning~\cite{Trajcevski06on-linedata}\cite{ Kjaergaard:2009:EER:1555816.1555839} achieve comparable runtime performance, we show that our algorithm significantly outperforms these protocols in compression ratio while guaranteeing an error bound.

A group of methods have been developed based on the idea of \emph{Piece-wise Linear Regression} (PLR) on time-series data, and have achieved competitive precision and efficiency. For example, using a bottom-up approach, Keogh et. proposed to use a ``segment"-based representation of time-series for similarity search and mining \cite{DBLP:conf/ictai/Keogh97, DBLP:conf/kdd/KeoghP98}. In \cite{DBLP:conf/icdm/KeoghCHP01}, an algorithm called SWAB successfully reduced the time complexity of such approximation to $\mathcal{O}(n)$ by combining the power of a sliding-window approach and a bottom-up approach. Note that these algorithms are designed to deal with single dimensional time-series by optimizing the normalized error residue on the time-series values over a segment (on the $y$-axis). Though such settings can indeed be modified and extended to handle 2D trajectories, which concern perpendicular errors instead, they cannot solve our problem becasue they are unable to guarantee an error bound for the individual points in the original trajectory, though guaranteed error residue for approximated segments is possible. For moving object tracking, it is important to keep the ``outliers'' in the location history, whereas optimizing a segment-wise error residue does not fulfil such requirement.

Meanwhile, data ageing has become a popular technique to approximate streaming data in an amnesic fashion. The process is often referred to as an analogy to amnesia: the older the memory is, the less its value is and the more likely it will be forgotten. The idea is popularized for time-series data in \cite{DBLP:conf/icde/PalpanasVKGT04,DBLP:journals/tkde/PalpanasVKG08} as the authors proposed a generic framework that supports user-specific amnesic functions. Such technique also received much attention in later studies such as \cite{DBLP:conf/icde/GandhiFS10,DBLP:conf/cikm/PotamiasPS06,DBLP:conf/ipsn/Nath09}. Again, we note that such literature focuses on optimizing the error residue for the segments in a time-series \cite{DBLP:conf/icde/PalpanasVKGT04,DBLP:journals/tkde/PalpanasVKG08,DBLP:conf/icde/GandhiFS10,DBLP:conf/cikm/PotamiasPS06}, while in the case of location tracking it is vital to ensure an error bound for every original location point.

\vspace{-0.2cm}
\section{Background}
\label{sec:mot}
In this section, we present the background of the study. The hardware system architecture used in the real-life bat tracking application is described. We also briefly introduce two existing solutions \cite{Heckbert95surveyof,opening_window}. By analyzing these algorithms we provide insights into a new algorithm is neeeded for our application. Both  algorithms will be evaluated too in the comparative evaluation.

\vspace{-0.3cm}
\subsection{Motivating Scenario}

We employ the Camazotz mobile sensing platform~\cite{DBLP:conf/ipsn/JurdakSKKCMW13}, which has been specifically designed to meet the stringent constraints for weight and size when being employed for wildlife tracking. A detailed specification of the nodes is provided in \cite{DBLP:conf/icde/LiuZSSKJ15}. Sensor readings and tracking data can be stored locally in external flash storage (1 MByte) until the data can be uploaded to a base station deployed at animal congregation areas using the short range radio transceiver. We also place the same hardware on the dashboard of a car to capture traces of vehicles in urban road networks. 

\eat{, in particular for different species of flying foxes, also known as megabats (Pteropus).  Animal ethics requires that the total payload weight is smaller than a certain percentage (usually below 5\%) of the animal's body weight, which corresponds to a weight limit of roughly 20-30\,g for flying foxes. Camazotz is a light-weight but feature-rich platform built around the Texas Instrument CC430F5137 system on chip, which integrates a 16-bit microcontroller (32 KBytes ROM, 4 KBytes RAM) and a short-range radio in the 900\,MHz frequency band. We use a rechargeable Lithium-ion battery connected to a solar panel to provide power to the device. Several on-board sensors such as temperature/pressure sensor, accelerometer/magnetometer and a microphone allow for multi-modal activity detection~\cite{DBLP:conf/ipsn/JurdakSKKCMW13} and sensing of the animal's environment. A ublox MAX6 receiver allows to determine the current position using the Global Positioning System (GPS). }

The GPS traces collected by such platforms are often used to analyze the mobility and the behavior of the moving object \cite{Zhao:2015qf}, or to perform spatial queries \cite{DBLP:journals/geoinformatica/ShangYDXZZ12}. Hence, it is most important to gather information for the object's major movements. The key features from the traces are the areas where it often visits, the route it takes to travel between its places of interests, and the time windows in which it usually makes the travels. However, the hardware limitations of the platform constrain its capability to capture such information in the long term. Motivated by this application, we propose an online trajectory compression and management framework on such resource-constrained platforms to reduce the data size and extend the operational time of the tracking platform in the wild. The framework will introduce a bounded error and discard some information for small movements, but will capture the interesting moments when major traveling occurs. Moreover, the framework will revisit the compressed trajectories from time to time when additional storage space is needed. In the process, older trajectories will be further compressed, subject to an extended error tolerance to reflect the significance decay of the data over time.

\subsection{Existing Solutions}
\subsubsection{Buffered Douglas-Peucker}
\emph{Douglas-Peucker} (DP) \cite{Heckbert95surveyof} is a well-known algorithm for trajectory
compression. In our scenario, in which the buffer size is largely constrained due to
memory space limit on the platform, we are still able to apply DP on a
smaller buffer of data. We call it \emph{Buffered Douglas-Peucker} (BDP) . The incoming points are accumulated in the buffer until it is full, then the partial trajectory defined by
the buffered points is processed by the DP algorithm.  However,
such solution has inferior compression rates mainly due the extra points taken 
when the buffer is repeatedly full, 
preventing a desirable high compression rate from being achieved.

An implication of BDP
is that both the start and end points in the buffer will be kept in the
compressed output every time the buffer is full, even when they can actually be discarded safely. 
In the worst case scenario where the object is moving in a straight line from
the start to the end, this solution will use $floor(\frac{N}{M})+1$
points, where $N$ is the number
of total points and 
$M$ is the buffer size. In contrast, the optimal solution needs to keep only two points.
 Although the overhead
depends on the shape of the trajectory and the buffer size, generally BDP takes considerably more points than necessary, particularly for small buffer sizes.

\vspace{-0.3cm}
\subsubsection{Buffered Greedy Deviation}
\emph{Buffered Greedy Deviation} (BGD, a variation of the generic sliding-window algorithm) represents another
simplistic approach. In this strategy whenever a point arrives, we
append the point to the end of the buffer, and do a complete calculation of the
maximum error for the trajectory segment defined by the points in the
buffer to the line defined by the start point and the end point. 
If this error already exceeds the tolerance, then we keep
the last point in the compressed trajectory, clear the buffer and
start a new segment at the last point. Otherwise the process continues for the next incoming point.

The algorithm is easy to implement and guarantees the error tolerance,
however it too has a major weakness. The compression rate is heavily
dependent on the buffer size, as it faces the same problem as 
BDP. If we increase the buffer size,  because the
time complexity is $\mathcal{O}(n^2)$, the computational complexity would increase
drastically, which is undesirable in our scenario 
because of the energy limitations. Therefore, BGD represents a significant
compromise on the performance as it has to make a direct trade-off between time
complexity and compression rate.

Clearly, a more sophisticated algorithm that can guarantee the bounded error, 
process the point with low time-space complexity, and achieve
high compression rate is
desired. We propose the \emph{Bounded Quadrant System} (BQS) algorithm family to address this problem. Before delving into the details, we present some
notations and definitions to help the reader understand
BQS's working mechanism.

\vspace{-0.3cm}
\subsection{Preliminaries}
\label{sec:pre}
We provide a series of definitions as the
necessary foundations for further discussion. 

\eat{
\begin{table*}[htp]
\centering
\caption{Symbols and Notations}
\label{tbl:sym}
\begin{tabular}{ | c || p{12cm} | }
\hline Notation & Description \\ 
\hline $v = <lat, lon, t>$ & location at timestamp t \\
\hline $wp = <lat_{min}, lat_{max}, lon_{min}, lon_{max} , w>$ & waypoint area\\
\hline $\tau=\{v_1, ...,v_k\}$ &  trajectory segment $\tau$, composed by a
series of points $v_i$\\ 
\hline $\tau'=\{v_1, v_k\}$ &  compressed trajectory segment $\tau'$\\
\hline  $\mathbb{T}=\{\tau_1,...,\tau_n\}=\{v_1, ...,v_n\}$ & trajectory, each
represents a trip\\
\hline  $\mathbb{T}'=\{v_1, v_j,...v_k\}$ & compressed trajectory, each
adjacent pair $\{v_i,v_j\}$ represent a compressed segment, \newline while together
represents a trip\\
\hline $\varrho = <\beta,\iota>$ & quadrant bounding system\\
\hline $\beta = <min_x, max_x, min_y, max_y>$ & bounding box \\  
\hline $\iota = <l^{left}, l^{right}>$ & left and right angle-bounding
lines \\ 
\hline $\theta_{lb}, \theta{ub}$ &  the smallest and largest angle
from any point in the buffered segment to the start point\\ 
\hline $d^{lb}, d^{ub}$ & the lower bound and
upper bound of the maximum deviation for a quadrant bounding system\\ 
\hline
\end{tabular}
\end{table*}
}
{\small
\begin{definition}[Location Point]
A location point $v=<latitude,longitude,timestamp>$ is a tuple that records the spatio-temporal
information of a location sample.
\end{definition}

\begin{definition}[Segment  and Trajectory]
A trajectory segment is a set of location points that are taken
consecutively in the temporal domain, denoted as $\tau=\{v_1, ...,v_n\}$. A trajectory
is a set of consecutive segments,
denoted as $\mathbb{T}=\{\tau_1,\tau_2,...\}$.
\end{definition}
}

Given the definitions of segment and trajectory, we introduce
the concept of compressed trajectory with bounded error: 

{\small \begin{definition}[Deviation]
Given a trajectory segment $\tau=\{v_1, ...,v_n\}$, the 
deviation $\bdd(\tau)$  is defined as the largest point-to-line-segment distance from any
location $v_i \in \{v_2,...,v_{n-1}\}$ to the line segment defined by $v_1$ and $v_n$. The
trajectory deviation is defined as the maximal segment deviation from any
of its segments, as $max(\bdd (\tau_i)), \tau_i \in \mathbb{T}$. 
\label{def:dev}
\end{definition}}
Deviation is a distance metric to measure the maximum error from the
compressed trajectory segment to the original trajectory. Without loss of generality, 
 we use point-to-line-segment distance in this definition. Note that
 point-to-line distance can be easily used within BQS too, after a few 
 minor modifications to the lower bound and upper bound calculations.

{\small \begin{definition}[Key Point]
Given a buffer $\tau=\{v_1, ...,v_k\}$, a
deviation tolerance $\epsilon^d$, a new location
point $v_{k+1}$, $v_k$ is a key point if $\bdd(\tau) \leq \epsilon^d$ and
$\bdd(\tau \bigcup \{v_{k+1}\}) > \epsilon^d$, where $\epsilon^d$ is the
error tolerance.
\end{definition}}
In other words, when a new point is sampled, the
immediate previous point is classified as key point if the new point
results in the maximum error of any point in the buffer exceeding $\epsilon^d$.

{\small \begin{definition}[Compressed Trajectory]
Given a trajectory segment $\tau=\{v_1, ...,v_n\}$, its compressed
trajectory segment
is defined by the start and end location $v_1$ and $v_n$, and is denoted
as $\tau'=\{v_1,v_n\}$. The compressed trajectory $\mathbb{T}' =\{v_i, v_j, ...,v_k\}$ of
$\mathbb{T}$ is the set of starting and ending locations of all the
segments in $\mathbb{T}$, ordered by the position of its source
segment in the original trajectory.
\end{definition}

{\small
\begin{definition}[Error-bounded Trajectory]
An error-bounded trajectory is a compressed trajectory with the 
deviation for any of its compressed segments smaller than or equal to
a given error tolerance
$\epsilon^d$. Formally: 
given a trajectory $\mathbb{T}=\{\tau_1,...,\tau_k\}$, and
its compressed trajectory $\mathbb{T}' =\{v_i, v_j, ...,v_k\}$,
$\mathbb{T}'$ is error-bounded by $\epsilon^d$ if $\forall \tau_i'\in \mathbb{T}'$,
$\bdd(\tau_i') \leq \epsilon^d$.
\end{definition}}
}
\begin{figure}[htp]
\vspace{-0.2cm}
\hspace{1cm}
\includegraphics[height=3.5cm]{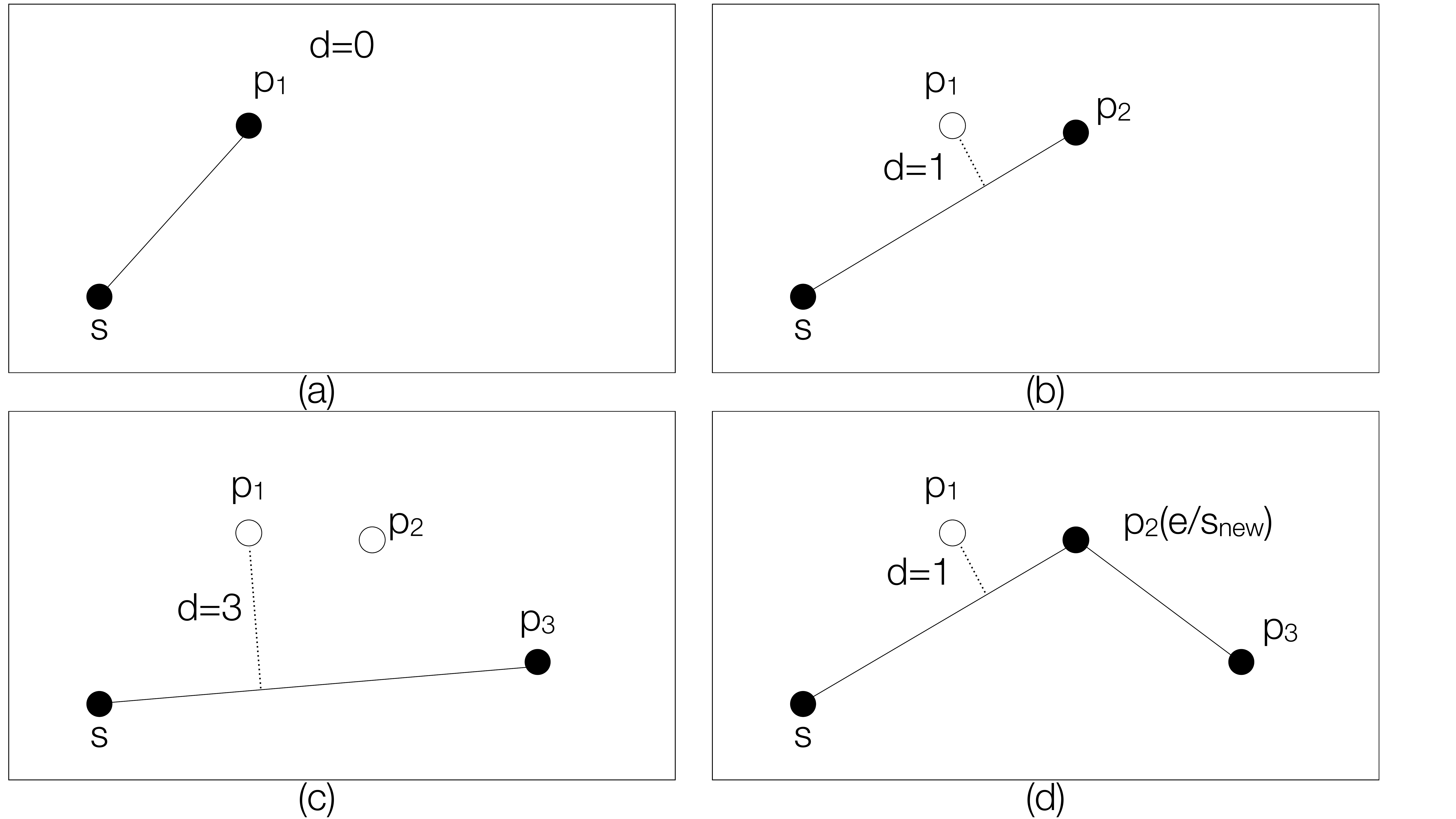}
\caption{Error-bounded Compression ($\epsilon=2$)}
\vspace{-0.3cm}
\label{fig:cexp}
\end{figure}

Figure \ref{fig:cexp} demonstrates the process of error-bounded
trajectory compression. Assuming that the current trajectory segment
starts from $s$, when adding the first point $p_1$ (Figure
\ref{fig:cexp}(a)), the deviation is $0$. Hence we proceed to the next
incoming point $p_2$ (Figure \ref{fig:cexp}(b)). Here the deviation
is $1$, which lays within the error tolerance, so the current segment can be
safely represented by $\overline{sp_2}$. However after $p_3$ arrives,
the deviation of the segment reaches $3>\epsilon$ because of
$p_1$ (Figure \ref{fig:cexp}(c)). Clearly $p_3$ should not be included in the current
trajectory segment. Instead, the current segment ends at $p_2$ and a
new segment is then started at $p_2$ (Figure \ref{fig:cexp}(d)). The
new segment includes $p_3$ and the above process is repeated until the $\epsilon$ is exceeded again.
Such process guarantees that any trajectory segment has a smaller deviation than $\epsilon$.

When a trajectory is turned into a compressed trajectory, the temporal
information it carries changes representation too. Instead of
having a timestamp at every location point in the original trajectory,
the compressed trajectory uses the timestamps of the key points as the
anchors to reconstruct the time sequences. Given a trajectory segment
defined by two key points $v_s, v_e$, the reconstructed location at
timestamp $\overline{t}$ ($v_s.t \le \overline{t}\le v_e.t$) is:

\begin{equation}
\small
v_{\overline{t}} = < h_{lat}(\mathrm{P}, v_s, v_e, \overline{t}), h_{lon}(\mathrm{P},
v_s, v_e, \overline{t}), \overline{t} >.
\end{equation}

where function $h$ is an interpolation function that uses a
distribution function $\mathrm{P}$, a start value, an end
value, and a timestamp at which the value should be
interpolated. $\mathrm{P}$ interpolates the location at a timestamp
according to a distribution. As
an example, the
$h$ and $\mathrm{P}$ functions for interpolating the latitude can be defined as:

{\small
\begin{eqnarray}
\mathrm{P}(\overline{t}) &=& \frac{\overline{t}-t^{v_s}}{t^{v_e} - t^{v_s}}\\
h_{lat} (\mathrm{P}, v_s, v_e, \overline{t}) &=& v_s.lat +
\mathrm{P}(\overline{t}) \times (v_e.lat - v_s.lat)
\end{eqnarray}
} 
where $\mathrm{P}$ is set to reconstruct the uniform
distribution. However, in practice this function can be derived online
to fit the distribution of the actual data. For instance, an online algorithm for
fitting Gaussian distribution by dynamically updating the variance and mean can be
implemented with semi-numeric algorithms described in
\cite{Knuth:1997:ACP:270146}, which can be used to derive
$\mathrm{P}$. 
\eat{Again, we want to state that
reconstructing the trajectory flawlessly is not the goal of this
algorithm. The trajectories processed by our algorithm will capture
the dominant movements. The algorithm overall is
designed to capture larger movements while maintaining a reasonable,
guaranteed error tolerance, which is useful for many types of mobility analysis.
}

As we favor an online algorithm where each point is processed
exactly once, the problem is turned into answering the following question: does 
the incoming point result in a new compressed trajectory segment or can it be
represented in the previous compressed trajectory segment? To address
this question, we first provide an overview and then the details for the BQS algorithm.

\begin{figure}[htp]
\vspace{-0.3cm}
\hspace{.8cm}
\includegraphics[height=4.5cm]{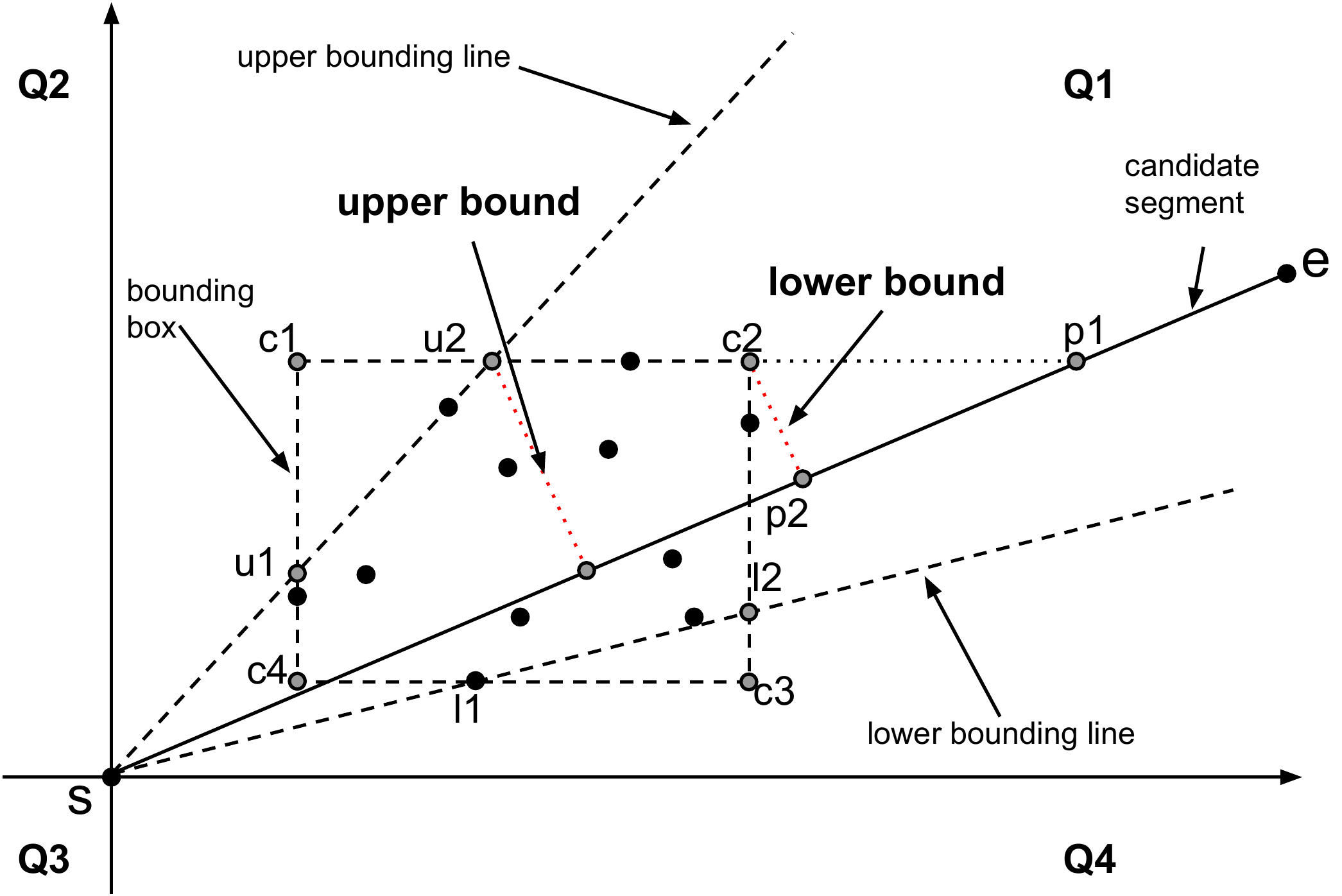}
\caption{An Example of the BQS}
\label{fig:bqs}
\vspace{-0.8cm}
\end{figure}

\section{The BQS Algorithm}
\label{sec:bqs}

The motivation of the framework is that we need a trajectory
management infrastructure to instantaneously and continuously manage historical trajectory data with
minimal storage space while maintaining maximum compression error bound and 
capturing the major movements of the mobile object. First we need an 
efficient online trajectory compression algorithm, that yields
error-bounded results in low time and space complexity and that minimizes
the number of points taken, i.e. the Bounded Quadrant System (BQS).

\eat{
Th framework aims to address the following:
\begin{enumerate}
\item 
\item We need a waypoint area discovery functionality to facilitate
  advanced tasks like trip time estimation etc. A waypoint is an area
  at where multiple trajectories join paths but part ways afterwards.
\end{enumerate} 
Subsequently, we formulate the sub-problems as follows:
\begin{enumerate}
\item Trajectory Compression: \textit{
Given a limited buffer size and computation power, and a stream
that continues to feed new location points, determine
whether the incoming point should be kept as the part of the current
segment, or the trigger of a new segment, so that the current segment's 
deviation will be bounded by the tolerance $\epsilon$.
}

\item Trajectory Management: \textit{
While a trajectory is being processed, and compressed segments are fed
into the framework, the framework should be able to index the segment,
merge similar segments, and further compress aged trajectories to optimize space utilization.
}
\item Waypoint Area Discovery:\textit{
The framework should be able to periodically 
search and retrieve similar historic segments and trajectories so that
small areas which have heavy and diversified historical traffic are discovered as
waypoint areas. 
}
\end{enumerate} 
}

A BQS is a convex hull that bounds a
set of points (e.g. Figure \ref{fig:bqs}). A BQS is constructed by the following steps:

{\small 
\begin{enumerate}
\item For a trajectory segment, we split the space into four
  quadrants ($Q1,Q2,Q3$ and $Q4$), with the origin set at the start point $s$ of the current
  segment, and the axes set to the UTM (Universal Transverse Mercator) projected $x$ and $y$ axes.
\item For each quadrant, a bounding box is set for the buffered points
  in that quadrant, if there are any. There can be at most four
  BQS for a trajectory segment. In Figure \ref{fig:bqs} there is only one $\overline{c_1c_2c_3c_4c_1}$ in $Q1$
  as the trajectory has not  reached other quadrants.
\item We keep two bounding lines that record
  the smallest and greatest angles between the $x$ axis and the line from the origin to
  any buffered point for each quadrant ($\overline{su_{2}}$ and $\overline{sl_{2}}$).
\item We have at most eight significant points in every quadrant systems - four
  vertices on the bounding box, four intersection points from the
  bounding lines intersecting with the bounding box. Some of the points may overlap.
  In this case, we have $c_1,c_2,c_3,c4,l_1,l_2,u_1,u_2$.
\item Based on the deviations from the significant points to the
  current path line, we have a group of lower bound candidates and
  upper bound candidates for the maximum deviation. From these candidates we derive a pair of lower
  bound and upper bound for the maximum deviation $<d^{lb},~d^{ub}>$, to make compression decisions without the full computation of segment
  deviation in most of the cases. 
\end{enumerate}
}
Here the lower bound $d^{lb}$
  represents the smallest deviation possible for all the points in the segment
  buffer to the start point, while the upper bound $d^{ub}$ represents the
  largest deviation possible for all the points in the segment
  buffer to the start point. With $d^{lb}$ and $d^{ub}$, we have three
  cases:
  {\small 
\begin{enumerate}
\item If $d^{lb} > \epsilon ^d$, it is unnecessary to perform
  deviation calculation because the deviation is guaranteed to break
  the tolerance, so a new segment needs to be started.
\item If $d^{ub} \leq \epsilon ^d$ , it is unnecessary to perform
  deviation calculation because the deviation is guaranteed to be
  smaller or equal to the tolerance, so the current point will be included in the
  current segment, i.e. no need to start a new segment.
\item If $d^{lb} \leq \epsilon ^d < d^{ub}$, a deviation
  calculation is required to determine whether the actual deviation is breaking
  the tolerance.
\end{enumerate}
}
Hence a pair of bounds is considered ``tight'' if the
  difference between them is small enough that it leaves minimum
  room for the error tolerance $\epsilon^d$ to be in between them.



\eat{
An illustration of a BQS is provided in Figure \ref{fig:bqs}. Here we have
a BQS in the first quadrant. Note that it shows only one BQS but in
reality there could be at most four BQS for each segment, one for each
quadrant. In the
figure, $s$ is the start point of the current segment, which is also
used as the origin of the BQS. The solid
black dots are the buffered points for the current segment. The bounding box
$\overline{c_1c_2c_3c_4c_1}$ is the minimum bounding rectangle for all
the points in the first quadrant. The bounding lines
$\overline{su_{2}}$ and $\overline{sl_{2}}$ record the greatest and
    smallest angles that any line from the origin to the points can
    have (w.r.t. the $x$ axis), respectively.
}

The intuition of the BQS structure is to exploit the favorable properties
of the convex hull formed by the
significant points from the bounding box
and bounding lines for the buffered points (excluding the start point). That is, with the
polygons formed by the bounding box and the bounding lines, we
can derive a pair of tight lower bound and upper
bound on the deviation from the points in the buffer to the line segment
$\overline{se}$. 
With such bounds, most of the
deviation calculations on the buffered points are avoided. Instead we can
determine the compression decisions by only assessing a few significant
vertices on the bounding polygon. The splitting of the space into four
quadrants is necessary as it guarantees a few useful properties to
form the error bounds.

To understand how the BQS works, we use the example
with a start point $s$, a few points in the buffer, and the last incoming
point $e$ as in Figure \ref{fig:bqs}. The goal is to determine whether
the deviation will be greater than the
tolerance if we include 
the last point in the current segment. First we present two fundamental theorems:

{\small 
\begin{theorem}
\label{thm:-1}
Assume that a point $p$ satisfies $d(p, s)\leq \epsilon$, where
s is the start point, then
\begin{eqnarray}
d^{max}(p, \overline{se}) \leq \epsilon
\end{eqnarray}
regardless of the location of the end point $e$.
\end{theorem}
}

Note that the proofs for all the theorems provided in this section are provided in \cite{DBLP:conf/icde/LiuZSSKJ15}.
This theorem supports the quick decision on an incoming point even
without assessing the bounding boxes or lines. Such a point is directly ``included''
in the current segment, and the BQS structure will not be affected. It is also important to separate these
points so that BQS need only to consider points further from the origin, because these points could 
potentially result in huge bounding angles and close-to-origin
bounding boxes, rendering the BQS structure ineffective.

{\small 
\begin{theorem}
\label{thm:0}
Assume that the buffered points $\{p_i\}$ are
bounded by a rectangle in the spatial domain, with the vertices $c_1, c_2,
c_3, c_4$, if we denote the line defined by the start
point $s$ and the end point $e$ as $\overline{se}$, then we always have:
\begin{eqnarray}
d^{max}(p_i, \overline{se}) \geq min\{ d(c_i, \overline{se}) \} = d^{lb}\\
d^{max}(p_i, \overline{se}) \leq max\{ d(c_i, \overline{se}) \} =d^{ub}
\end{eqnarray}
\end{theorem}
}
\eat{The bounding box
dictates that on each of its edges there must be at least one point.
If for any edge $\overline{c_jc_k}$ of the bounding box we denote a
buffered point on it as $p^e$, then we have $ min\{ d(c_j, \overline{se}), d(c_k, \overline{se})\} \leq d(
p^e, \overline{se} ) \leq max\{ d(c_j, \overline{se}), d(c_k, \overline{se})\}
$. Consolidating the bounds on all of the edges, we have the proof of
the theorem.
}

Theorems \ref{thm:-1} and \ref{thm:0} show how some of the points can
be safely discarded and how the basic lower bound and upper bound properties are
derived. However, Theorem \ref{thm:0}
only provides a pair of loose bounds that can hardly avoid any deviation
computation. To obtain tighter and useful bounds, we need to introduce
a few
advanced theorems. Throughout the theorem definitions we use the following notations:
{\small 
\begin{itemize}
\item Corner Distances: We use $d^{corner} = \{ d(c_i, \overline{se})\},
  ~ i\in\{1,2,3,4\}$ to denote the distances from each vertex of
  the bounding box to the current path line.
\item Near-far Corner Distances: We use $d^{corner-near} = \{ d(c_n, \overline{se})\}$ and $d^{corner-far} = \{ d(c_f, \overline{se})\}$,
  to denote the distances from the nearest vertex $c_n$ and the
  farthest vertex $c_f$ of the bounding box (near and far in terms of
  the distance to the origin)  to the current line. 
  The nearest and farthest corner points are determined by the
  quadrant the BQS is in. For example, in Figure \ref{fig:bqs} $c_n=c_4$ and $c_f=c_2$. 
\item Intersection Distances: We use $d^{intersection} = \{ d(p,
  \overline{se})\},  ~ p \in\{l_1, l_2, u_1, u_2\}$ to
  denote the distances from each 
intersection to the current path line, where $l_i$ are the
  intersection points of the lower angle bounding line and the bounding box,
  and $u_i$ are the
  intersection points of the upper angle bounding line and the bounding box.
\end{itemize} 
}

Some advanced bounds are defined as follows:

{\small 
\begin{theorem}
\label{thm:1}
Given a BQS, if the line $\overline{se}$ is in the quadrant, and
$\overline{se}$ is in between the two bounding lines ($\theta_{lb} \leq
\theta_{s,e} \leq \theta_{ub}$), then we have the following bounds on
the segment's deviation: 
\begin{eqnarray}
&d^{max}(p, \overline{se})& \geq d^{lb} = \\ \nonumber &max& \begin{cases} ~ & min\{ d(l_1, \overline{se}),
  d(l_2,\overline{se}) \}  \\ ~ & min\{ d(u_1, \overline{se}),
  d(u_2,\overline{se}) \} \\ 
  ~& \begin{cases} ~ & d^{corner-near}, \forall d(s,e)< d(s, c_f)
  \\ ~ & d^{corner-far}, \forall d(s,e)\geq d(s, c_f) \end{cases}\end{cases}\\
\label{eqn:ub} &d^{max}(p, \overline{se})& \leq  d^{ub} = \\
\nonumber &max&\{  d^{intersection}, d^{corner-far}\}
\end{eqnarray}
\end{theorem}
}
A line $l$ is ``in'' the quadrant $Q$ if the angle $\theta^l$ between $l$ and the
$x$ axis satisfies $ \theta^Q_{start}
\leq \theta^l < \theta^Q_{end}$, where $\theta^Q_{start}$ and
$\theta^Q_{end}$ are the angle range of the quadrant where the BQS
resides. Note that this definition is distance metric-specific. \eat{Since
we use point-to-line distance, a line is automatically ``in'' exactly
two quadrants.}
In future references we assume
$\theta^l$ satisfies $ \theta^Q_{start}
\leq \theta^l < \theta^Q_{end}$ if it is ``in''
the quadrant.

{\small 
\begin{theorem}
\label{thm:2}
Given a BQS, if  $\overline{se}$ is in the quadrant, and
is outside the two bounding lines ($\theta_{ub} <
\theta_{s,e} ~or~ \theta_{lb} > \theta_{s,e}$), we have the same bounds on
the segment deviation as in Theorem \ref{thm:1}.
\end{theorem}
}
If the path line is not in the same quadrant with the BQS because the new $e$ moves to another quadrant, 
we use Theorem \ref{thm:4} to derive the bounds:
{\small \begin{theorem}
\label{thm:4}
Given a BQS, if the line $\overline{se}$ is not in the quadrant, 
the bounds of the segment deviation are defined as: 
\begin{eqnarray}
d^{max}(p, \overline{se}) &\geq& d^{lb} = \\ \nonumber &max& \begin{cases} ~ & min\{ d(l_1, \overline{se}),
  d(l_2,\overline{se}) \}  \\ ~ & min\{ d(u_1, \overline{se}),
  d(l_2,\overline{se}) \} \\ ~ & 3^{rd}largest(\{d^{corner}\})  \end{cases}\\
d(p, \overline{se}) &\leq&  max\{  d^{corner} \} = d^{ub}
\end{eqnarray}
\end{theorem}
}

\eat{A summary of the bounds and criteria is provided in the
following table:
\begin{table}[htp]
\caption{Bounds and conditions}
\label{tbl:bnd}
\begin{tabular}{ | c |c|c|c | }
\hline \shortstack{In\\quadrant} & \shortstack{Intersects\\box} &
\shortstack{In bounding\\lines} & \shortstack{Applicable\\ Theorem} \\ 
\hline $\surd$ & $\surd$ & $\surd$ & \ref{thm:1} \\
\hline $\surd$ & $\surd$ & $\times$ &\ref{thm:2} \\
\hline $\surd$ & $\times$ & $\times$ & \ref{thm:3} \\
\hline $\times$ & $\times$ & $\times$ & \ref{thm:4} \\
\hline
\end{tabular}
\end{table}
Note that the conditions here are progressive. A line must be in the
quadrant to be able to intersect with in the bounding box, then it has
be intersecting with the bounding box to be in the bounding
lines. Given such relations, the table here covers every possibility
that can occur given a line and a QBS, which guarantees in every case
there will be applicable lower bound and upper bound to help avoid an
actual full calculation of segment deviation.}

\vspace{-0.7cm}
\subsection{The BQS Algorithm}

The BQS algorithm is formally described in Algorithm \ref{alg:1}:

\begin{algorithm}[thp]
\caption{The BQS Algorithm}
\textbf{Input:} $s$, $e$, $\mathcal{B}$, $\epsilon^d$ \COMMENT{\textcolor{blue}{start\&end, buffer, tolerance}}\\
\textbf{Algorithm:}

\begin{algorithmic}[1]

  \IF{$d(s,e) \leq \epsilon^d $} \COMMENT{\textcolor{blue}{check for trivial decision}}
  \STATE $Decision: \mathcal{B} \leftarrow e \cup \mathcal{B}$ 
  \ELSE
  
  \STATE Compute $d^{lb}_i$, $d^{ub}_i$ using Theorems \ref{thm:1},  \ref{thm:2}, and \ref{thm:4}
  \STATE $d^{lb} \leftarrow max\{ d^{lb}_i \}$ and   $d^{ub} \leftarrow max\{ d^{ub}_i \}$
  
  \IF{$d^{ub} \leq \epsilon^d $} \COMMENT{\textcolor{blue}{low upper bound, safe to continue}}
  \STATE $Decision: \mathcal{B} \leftarrow e \cup \mathcal{B}$
  
  \ELSIF{$d^{lb} > \epsilon^d $} \COMMENT{\textcolor{blue}{high lower bound, stop immediately}}
  \STATE $Decision$: Current segment stops and a new segment starts

  \ELSIF{$d^{lb} \leq \epsilon^d < d^{ub}$} \COMMENT{\textcolor{blue}{decision uncertain from bounds}}
  \STATE $d \leftarrow ComputeDeviation(\mathcal{B}, \overline{se})$
  \STATE $Decision$: made according to $d$
  \ENDIF
  \ENDIF
  \STATE Update the corresponding BQS for $e$

\end{algorithmic}

\label{alg:1}
\end{algorithm}


The algorithm starts by checking if there is a trivial decision: by
Theorem \ref{thm:-1}, if the incoming point $e$ lays within the range of
$\epsilon$ of the start point, no matter where the final end point is,
the point will not result in a greater deviation than
$\epsilon$, so $e$ is added to buffer  $\mathcal{B}$ (Lines 1-2). After $e$ passes this test, it means $e$ may result in a
greater deviation from the buffered points. So we assume that $e$ is
the new end point, and assume the current segment is presented by
$\overline{se}$. Now for each quadrant we have a BQS, maintaining their
respective bounding boxes and bounding lines. For each BQS, we have a
few (8 at most) significant points, identified as the four corner
points of the bounding box and the four intersection points between
the bounding lines and the bounding box. According to the theorems we
defined, we can aggregate four sets of lower bounds and upper bounds
for each quadrant, and then a global lower bound and a global upper
bound for all the quadrants (Lines 4-5). According to the global lower
bound and upper bound, we can quickly make a compression decision. If
the upper bound is smaller than $\epsilon$, it means no buffered point
will have a deviation greater than or equal to $\epsilon$, so the
current point $e$ is immediately included to the current trajectory
segment and the current segment goes on (Lines 6-7). On the contrary, 
if the lower bound is greater than
$\epsilon$, we are guaranteed that at least one buffered point will
break the error tolerance, so we save the current segment (without
$e$) and start a new segment at $e$ (Lines 8-9). Otherwise if the
tolerance happens to be in between the lower bound and the upper
bound, an actual deviation computation is required, and decision will
be made according to the result (Lines 10-12). Finally, if the
current segment goes on, we put $e$ into its corresponding quadrant
and update the BQS structure in that quadrant (Line 15).

\begin{figure}[htp]
\vspace{-0.3cm}
\hspace{1.8cm}
\includegraphics[height=4cm]{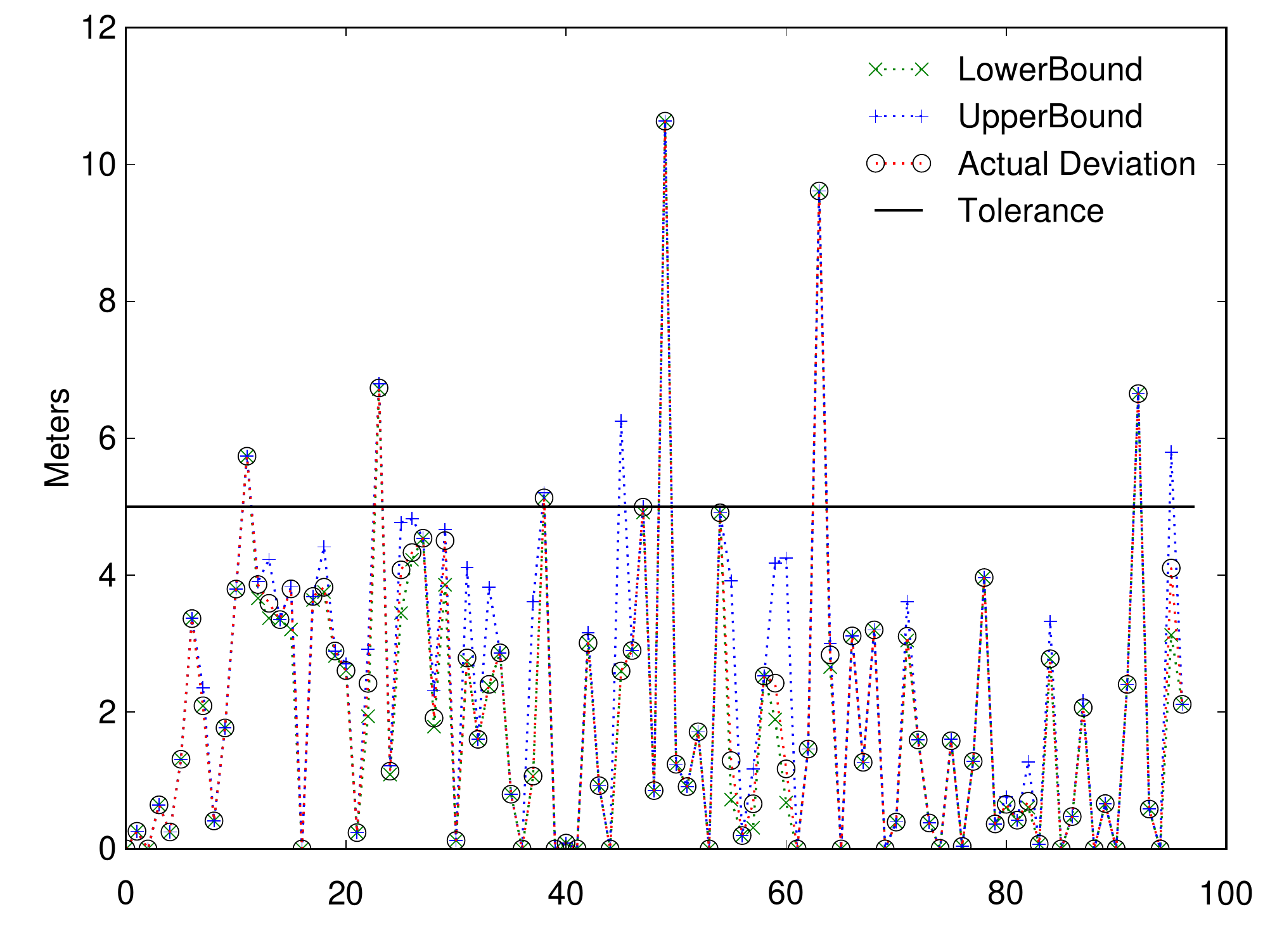}
\caption{Bounds v.s. Actual Deviation}
\label{fig:bds}
\vspace{-0.3cm}
\end{figure}

Figure \ref{fig:bds} demonstrates the lower and upper
bounds as well as the actual deviations of some 
randomly-chosen location points from the real-life
flying fox dataset, with
$\epsilon^d$ set to $5~m$. The $x$ axis shows the indices of the points,
while the solid horizontal line indicates the error tolerance. It is evident that
in most cases the bounds are very tight and that in more than $90\%$ of the
occasions we can determine if a point is a key point by using only the
bounds and avoid actual deviation calculations.

A technique called data-centric rotation is used to further tighten 
the bounds \cite{DBLP:conf/icde/LiuZSSKJ15}, which also shows how to generalize
the algorithm to the 3-D case and to a different error metric. 

\eat{
\vspace{-0.3cm}
\subsection{Data-centric Rotation}
A technique called data-centric rotation is used to further tighten 
the bounds. When a new segment is started, instead of constructing and 
updating the BQS immediately after the arrival of new points, we allow 
a tiny buffer to store the first few 
points (e.g. 5) that are not in the range of $\epsilon$ within the start point
(meaning that these points will actually affect the bounding box). 
With the buffer, we compute the centroid of the buffered
points, and rotate the current $x$ axis to the line from the start
point to the centroid. By applying this rotation, we enforce that the 
points are split into two BQS. This improves the
tightness of the bounds because the bounding convex-hulls are
generally tighter with less spread. Once the rotation angle is
identified, each new point for the same segment is temporally 
rotated by the same angle when estimating their distances to the line
$\overline{se}$.

\begin{figure}[htp]
\vspace{-0.4cm}
\hspace{1.5cm}
\includegraphics[width=4.7cm]{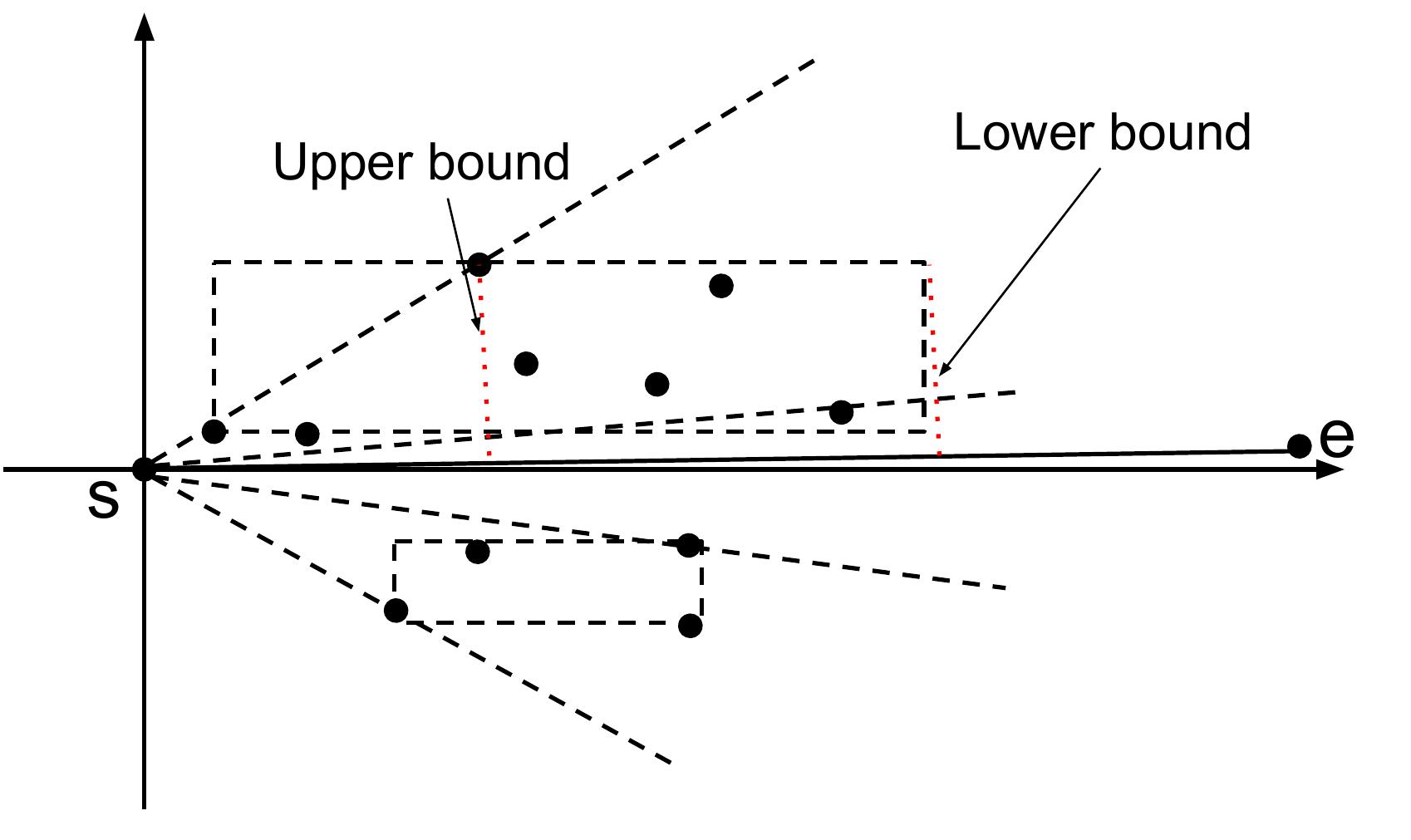}
\caption{Data-centric Rotation}
\label{fig:rtt}
\end{figure}
\vspace{-0.3cm}
Figure \ref{fig:rtt} shows the effect of the rotation. The
points are tidentical as in Figure \ref{fig:bqs}, except for the rotation that ``centers'' them to the $x$ axis. After
the rotation they are split into two BQS and it's visually
evident that the gap between the lower and the upper
bound becomes smaller. As in reality the likelihood is 
substantial for a moving object to travel in a major direction despite
slight heading changes, this step improves the BQS's pruning power 
significantly. The procedure is applied on Algorithm \ref{alg:1}. 
}

\vspace{-0.3cm}
\subsection{Achieving Constant Time and Space Complexity}
With the pruning power introduced by the deviation bounds, Algorithm
\ref{alg:1} achieves excellent performance in terms of time
complexity.
Its expected time complexity is $ \alpha \times
n \times c_1 +
(1-\alpha) \times n \times m \times c_2$, where $\alpha$ is the
pruning power, $m$ is the maximum buffer size, and $c_1,c_2$ are two
constants denoting the cost of the processing of each point by the BQS structure and by the full deviation calculation respectively.
\eat{
\begin{algorithm}[thp]
\caption{Faster BQS Algorithm}
\textbf{Input:} Start point $s$, incoming new point $e$, previous
point $e_p$, deviation tolerance $\epsilon^d$\\
\textbf{Algorithm:}

{\begin{algorithmic}[1]

  \STATE Identical with Algorithm \ref{alg:1} before Line 12 except
  line\\ 5 is no longer required
  \IF{$d^{lb} \leq \epsilon^d \leq d^{ub}$}
  \STATE $Decision$ = Current segment stops and new \\segment starts at the previous point $e_p$
  \ENDIF
  \RETURN $Decision$
\end{algorithmic}
}
\label{alg:2}
\end{algorithm}}
Empirical study in Section \ref{sec:exp} shows that $\alpha$ is
generally greater than $0.9$, meaning the time complexity is
approaching $\mathcal{O}(n)$ for the whole data stream. However, the
theoretical worst case time complexity is still $\mathcal{O}(n^2)$. 
Moreover, because we still keep a buffer for potential
deviation calculation, the worst-case space complexity
is $\mathcal{O}(n)$. To further reduce the complexity,  
we propose a more efficient version that still utilizes
the bounds but completely avoids any full deviation
calculation and any use of buffer, making the time and space
 costs constant for processing a point.

The algorithm is nearly identical to Algorithm \ref{alg:1}.
The only difference
is that whenever the case $d^{lb} \leq \epsilon^d < d^{ub}$ occurs
(Line 10), a conservative approach is taken. No deviation calculation is
performed, instead we take the point and start a new trajectory
segment to avoid any computation and to eliminate the necessity of
maintaining a buffer for the points in the current segment. So 
Lines 11-12 in Algorithm \ref{alg:1} are changed into making the ``stop and
restart'' decision (as in Line 9) directly without any full calculation in Line
11. The maintenance of the buffer is not needed any more.

The \emph{Fast BQS} (FBQS) algorithm takes slightly more points than Algorithm
\ref{alg:1} in the compression, reducing the compression rate by a small margin. However,
the simplification on the time and space complexity is
significant. The fast BQS algorithm achieves constant complexity in both time and
space for processing a point. Equivalently, time and space
complexity are $\mathcal{O}(n)$ and $\mathcal{O}(1)$ for
the whole data stream. 

The time complexity is only introduced
by assessing and updating a few key variables, i.e. bounding lines, intersection points and
corner points. We can now arrive at the compression decision by keeping only the
significant BQS points of the number $c\leq 32$ ( 4 corner points and
4 intersection points at most for each quadrant) for the
entire algorithm.

When the buffer size is unconstrained, the three algorithms 
Buffered Douglas-Peucker (BDP),
Buffered Greedy Deviation (BGD) and Fast BQS (FBQS)
have the following worst-case time and space complexity:
\eat{
\begin{table}[htp]
\centering
\caption{Worst-case Complexity}
\label{tbl:cmp}
\begin{tabular}{ | c |c|c| }
\hline  & \textbf{Time} & \textbf{Space} \\ 
\hline \textbf{FBQS} & $\mathcal{O}(n)$ & $\mathcal{O}(1)$\\
\hline \textbf{BDP} & $\mathcal{O}(n^2)$& $\mathcal{O}(n)$\\
\hline \textbf{BGD} & $\mathcal{O}(n^2)$& $\mathcal{O}(n)$\\
\hline
\end{tabular}
\end{table}
}

\begin{table}[htp]
\scriptsize 
\centering
\caption{Worst-case Complexity}
\label{tbl:cmp}
\begin{tabular}{ | c |c|c| c|}
\hline  & \textbf{FBQS} & \textbf{BDP} & \textbf{BGD}\\ 
\hline \textbf{Time} & $\mathcal{O}(n)$ & $\mathcal{O}(n^2)$& $\mathcal{O}(n^2)$\\
\hline \textbf{Space} & $\mathcal{O}(1)$ & $\mathcal{O}(n)$& $\mathcal{O}(n)$\\
\hline
\end{tabular}
\end{table}

\vspace{-0.5cm}
\eat{
Note that when the compressed trajectories are inserted to the spatial
index, their bounding boxes are loosened to cover the error-bound, so
that whenever a search is performed on the index, results will not be
missed because of the deviation introduced in the compression.
}

\eat{
\subsection{Generalization}
\label{sec:gen}
\eat{We briefly introduce a 3-D variant of the BQS algorithm to
demonstrate its extensibility for more complex application
requirements such as 3-D location tracking and time-sensitive tracking.
}
In scenarios such as indoor tracking or aero vehicle tracking, 
the process takes places in a
3-D geographic space defined by $<latitude, longitude, altitude>$ \cite{citeulike:3939699}.
In other tracking applications it is desirable to know where
the moving object is at a certain time, which makes the form $<latitude, longitude, timestamp>$. 
Therefore, time-synchronized error (as
in \cite{Cao:2006:SDR:1147679.1147681,DBLP:conf/edbt/MeratniaB04}) becomes a useful metric. 

Both applications require a trajectory compression algorithm that can handle
3-D data. Now we show that it is straightforward to extend the BQS algorithm to the 3-D
space. Let us revisit the example we
discussed in Figure \ref{fig:bqs}, where we demonstrated
the concept of bounding box and bounding line in the 2-D case. For
the 3-D case, instead of using bounding boxes and bounding lines, 
we use bounding right rectangular prisms and bounding
  planes to bound the location points. An illustration
is given in Figure \ref{fig:bqs3d}.

\begin{figure}[htp]
\vspace{-0.3cm}
\hspace{1.5cm}
\includegraphics[height=4.5cm]{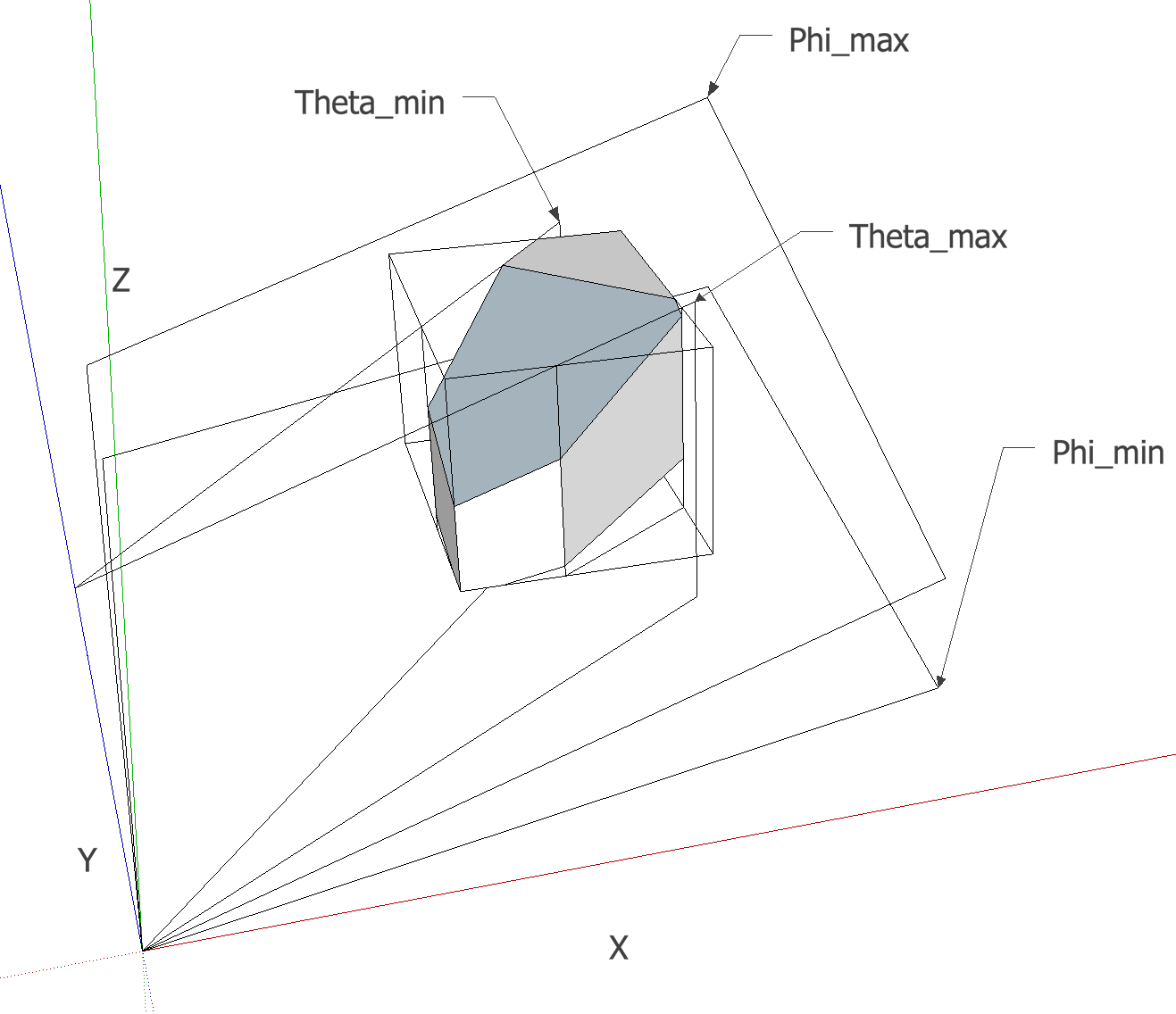}
\caption{The 3-D version of BQS illustrated}
\label{fig:bqs3d}
\end{figure}

In Figure \ref{fig:bqs3d}, the $z$ axis can either represent
the timestamp or the altitude. The deviation metric is  extended
to measure the maximum distance from original points to a
line in 3-D. There are eight quadrants in total. We use the quadrant 
($x>0 \land y>0 \land z>0 $) as an example. 
A bounding prism is used to bound the location points in each
quadrant, as a direct extension to the 2-D case. Then, in each
quadrant, we also establish two pairs of bounding planes to maintain
the minimum and maximum angles to their respective reference planes. 
We have a pair of ``vertical'' bounding planes $\Theta_{min},\Theta_{max}$, which both
are orthogonal to the $XY$ plane and contain the $z$ axis. They
represent the minimum and maximum angles formed by any plane that contains a
location points to the $YZ$ plane. Similarly, we have the ``inclined''
bounding planes $\Phi_{min}$ and $\Phi_{max}$. They represent the minimum and
maximum angles between the planes in $\{\Phi_i\}$ to the $XY$
plane. Any $\Phi_i$ is determined by three points, namely two anchor
points and one data point, and all $\Phi_i$ in a quadrant share the
same anchor points. The anchor points are determined by the quadrant,
as $(sign(x)\times 1,-sign(y) \times 1,0),(-sign(x) \times 1,sign(y) \times 1,0)$. For this quadrant, the anchor points
are $(1,-1,0)$ and $(-1,1,0)$.

\eat{
which form two inverted cones
starting from the origin. Every point on a cone yields the same angle
between the line from the origin to it and the line from the origin to
its projected point on the $XY$ plane. However, surfaces for inverted
cones are more expensive for the intersection calculations than
planes are. In practice, we hence use two planes two approximate the two
surfaces, namely the $\Phi_{min}$ and $\Phi_{max}$ bounding planes in
Figure \ref{fig:bqs3d}. They are derived as shown in Figure
\ref{fig:bqs3d}. First we calculate the intersecting curves between
the two cones and the bounding prism, then we choose the approximation
planes separately for the minimum and maximum angles. Figure
\ref{fig:approx} shows the approximation process, from an orthogonal
perspective from a point on the z-axis to the bounding prism. The two dotted
arcs are the intersections of the cones to the prism, and the two
dashed lines are the approximation plane for the two cones. For the maximum
case ($\Phi_{max}$), we use the two intersection points between the
cone and the edges of the first plane}

The bounding prism is hence ``cut'' into a 
polyhedron that is also a convex hull (shadowed part in Figure \ref{fig:bqs3d}), and the
vertices that form the hull are the significant points that derive error bounds. \eat{There are
mature software libraries to enable the efficient calculation of the bounding
polyhedron, such as GEOS\footnote{\url{http://trac.osgeo.org/geos/}}
or CGAL\footnote{\url{http://www.cgal.org/}}. 
In practice, to further improve
the efficiency of the computation of the convex hull, we only consider
the intersection points between the bounding planes and the
bounding prism, while intersections between bounding planes
are not considered. This will slightly increase the volume of the
bounding polyhedron but the computational cost will decrease
considerably and become constant (independent of the
spatial relations among bounding planes). Finally,} 
We obtain a 3-D convex hull formed by at
most 17 points (at most 4 intersection points for each bounding plane, 
plus the farthest vertex to the origin of the prism). 
Using similar techniques to Theorems
\ref{thm:1}, \ref{thm:2} and \ref{thm:4}, new pruning 
rules for the 3-D BQS can then be derived based on the significant
points. 

\eat{We also note that the
two bounding planes may cut the same part of the prism, resulting
complex intersections at in the figure. In practice this effect is
ignored for computational efficiency. Only the intersections between the
edges of the prism and the four bounding planes are considered as
significant points, leaving at most 16 significant 
points to assess. }
It is worth noting that besides the extended version in the 3-D
space, the BQS algorithm can also be used with different distance
metrics, such as the point-to-line distance. 
In cases where poin-to-line distance is used, Theorems
\ref{thm:1} and \ref{thm:2} can be used with minor
modification by changing Equation \ref{eqn:ub} to:

{\small
\begin{equation}
d^{max}(p, \overline{s,e}) \leq d^{ub} = max\{  d^{intersection}, d^{corner-nf} \}
\end{equation}
}
while Theorem \ref{thm:4} still holds. The definition
of the ``in quadrant'' property is slightly changed accordingly.
}
\vspace{-0.3cm}
\section{Trajectory Uncertainty and Amnesic BQS}
\label{sec:abqs}
Next we propose an online framework that compresses the historical data in an
``amnesic'' way to maximize the utilization of storage capacity. This scenario
poses two challenges:

\begin{enumerate}
\item how to compress the already compressed trajectories by extending and adhering to an increased error tolerance?
\item how to design the ageing procedure so that trajectories compressed with different error tolerances could be updated efficiently?
\end{enumerate}

\begin{figure}[htp]
\centering
\vspace{-0.7cm}
\includegraphics[width=4cm]{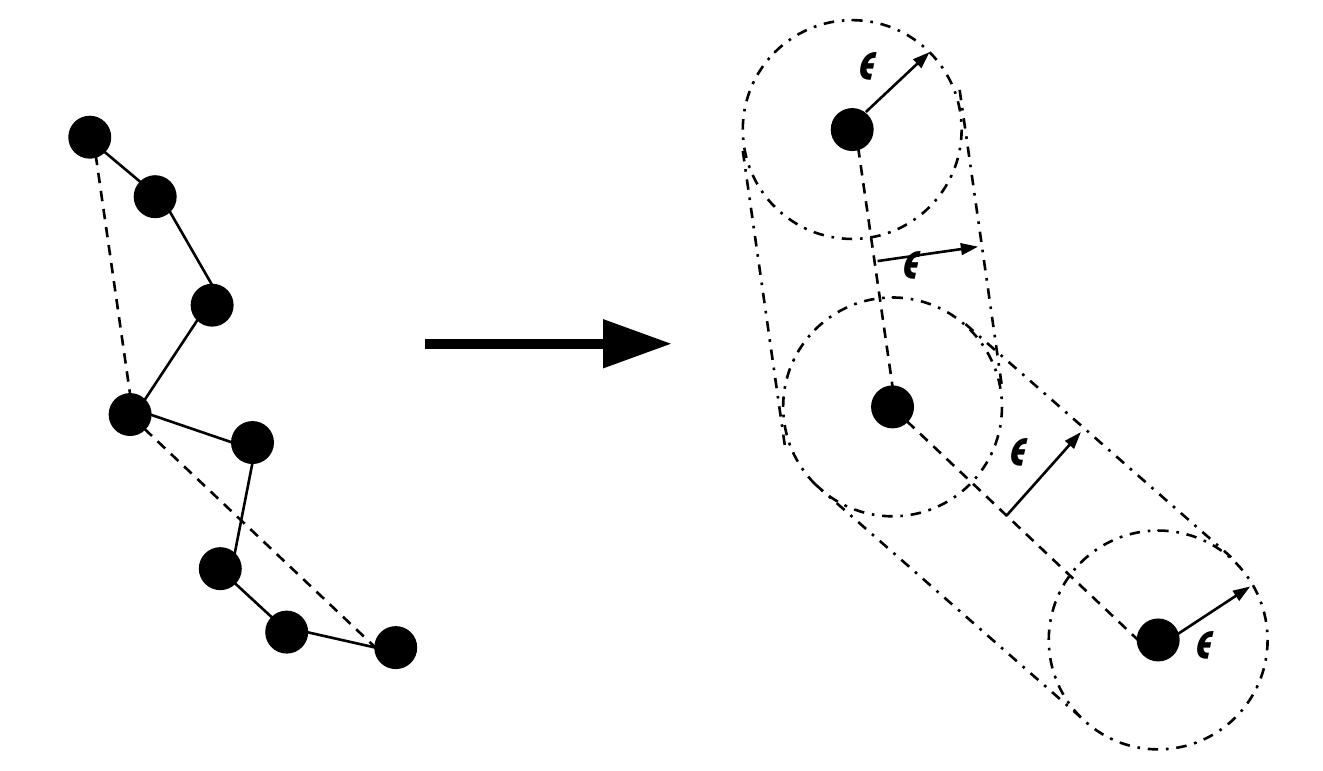}
\caption{A compressed trajectory has a bounded uncertainty (error tolerance $\epsilon$)}
\label{fig:unc}
\end{figure}

\begin{figure}[htp]
\centering
\vspace{-0.7cm}
\hspace{-0.5cm}
\includegraphics[width=8cm]{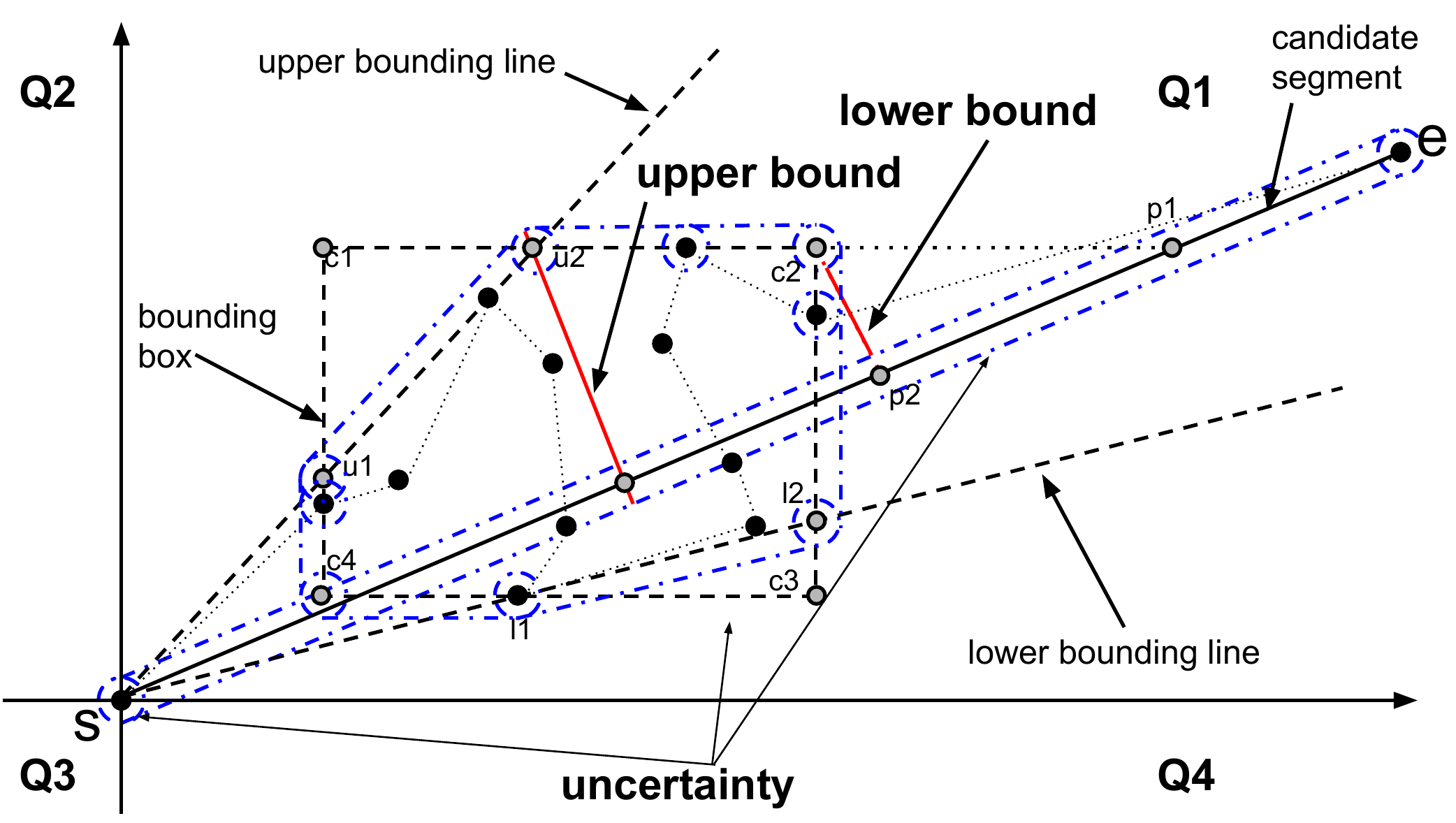}
\caption{BQS on a compressed trajectory that has a bounded uncertainty}
\label{fig:uc-bqs}
\end{figure}

To answer the first question, first we observe that after the first-pass compression with the error tolerance $\epsilon$, a trajectory is transformed into a simplified representation with a bounded uncertainty, as illustrated in Figure \ref{fig:unc}. Putting this representation into the BQS formulation (e.g. Figure \ref{fig:bqs}), we can derive an uncertain form of the BQS in Figure \ref{fig:uc-bqs}.

In this uncertain form, we use black dots to present trajectory points (key points with uncertainty from previously compressed trajectories in this case). The significant points on the bounding box or the bounding lines are still determined by the key points themselves, e.g. $l_1$ and $u_1$. However in addition to the key points, we also need to consider the uncertainty from the previous compression. The information for the discarded points from the original trajectory has been completely lost, and the previous error tolerance is the only information we could reconstruct such uncertainty. In this figure we use dash-dotted lines to illustrate such uncertainty. Note that the uncertainty here is twofold. Firstly, because the uncertainty around the start point $s$ and the end point $e$ (for points around both ends may have been discarded from the previous compression), the current candidate segment $\overline{se}$ is uncertain. Secondly, each compressed segment in $\overline{c_1c_2c_3c_4c_1}$ has its own bounded uncertainty, thus in the new form every significant point has an equivalent uncertainty and the bounding polygon $\overline{u_1u_2l_2l_1c_4u_1}$ is extended too and its corners are turned into round corners with radius equal to the previous tolerance.

\begin{algorithm}[thp]

\caption{The Amnesic BQS Procedure}
\label{alg:3}



  \textbf{Shared Variables:} $S, I, k, \epsilon,  m, N$ \\
\hfill
\hfill\\
\textbf{Procedure: $\mathbf{abqs}$()} \COMMENT{\textcolor{blue}{main entry}}\\
\textbf{Input:}  $p$
{\begin{algorithmic}[1]
\IF{$\|I\|>0$} \COMMENT{\textcolor{blue}{index not empty}}
	\STATE $ i \leftarrow I[0].e+1$ \COMMENT{\textcolor{blue}{find location for new point}}
		\IF{$I[0].a==0$} \COMMENT{\textcolor{blue}{youngest gen has age 0}}
			\STATE $\mathbf{update\_index(0,-1,i)}$ \COMMENT{\textcolor{blue}{update end index only}}
		\ELSE
			\STATE $\mathbf{update\_index(0,i,i)}$ \COMMENT{\textcolor{blue}{otherwise insert age 0 segment}}
		\ENDIF
\ELSE 
	\STATE $i\leftarrow 0$ \COMMENT{\textcolor{blue}{dealing with completely new storage}}
	\STATE $\mathbf{update\_index(0,i,i)}$
\ENDIF

\STATE $S[i] \leftarrow p$ \COMMENT{\textcolor{blue}{store new point}}
\STATE $\mathbf{amnesic\_sinking()}$ \COMMENT{\textcolor{blue}{check if ageing is needed}}

\end{algorithmic}
}
\textbf{Procedure: $\mathbf{amnesic\_sinking()}$} 
{\begin{algorithmic}[1]
\STATE $f \leftarrow \mathbf{trigger()} $ \COMMENT{\textcolor{blue}{check trigger}}
\WHILE{ f }  \COMMENT{\textcolor{blue}{if ageing is triggered}}
	\IF{$\|I\|==1$} \COMMENT{\textcolor{blue}{find dest loc for results from ageing}}
		\STATE $s^{buf} \leftarrow 0$   \COMMENT{\textcolor{blue}{dest  is 0 when buffer has age 0 only}}
	\ELSE		
		\STATE $s^{buf} \leftarrow I[1].e+1$ \COMMENT{\textcolor{blue}{otherwise append to next age}}
	\ENDIF

	\COMMENT{\textcolor{blue}{compress youngest generation}}
	\STATE $ \mathbf{compress}(I[0].s, I[0].e, I[0].a, s^{buf}) $ 
	\STATE $f \leftarrow \mathbf{trigger()} $	\COMMENT{\textcolor{blue}{check trigger again}}
\ENDWHILE
\algstore{testcont} 
\end{algorithmic}
}

\end{algorithm}


Recall that in the first-pass compression in Figure \ref{fig:bqs}, we use $d(u_2,\overline{se})$ as the upper bound and $d(c_2,\overline{se})$ as the lower bound for the maximum deviation. As we examine the uncertain form, we find that the lower bound and upper bound can be re-used, with simple modifications. In Figure \ref{fig:uc-bqs}, where $d(u_2,\overline{se})$ represents the upper bound without considering uncertainty, it is easy to prove that $d^{max}(p,\overline{se})\leq d(u_2,\overline{se})+2\epsilon$, and $d^{max}(p,\overline{se})\geq d(u_2,\overline{se})-2\epsilon$. The uncertainty of $\overline{se}$ and the bounding box each introduces uncertainty of $\epsilon$ to both bounds.

As a generalization of the specific case in Figure \ref{fig:uc-bqs}, we derive Theorem \ref{thm:ucbqs} to support the application of BQS on compressed trajectories with \textit{any} existing error tolerance (uncertainty).

{\small 
\begin{theorem}
\label{thm:ucbqs}
Given a set of compressed trajectories with a previous error tolerance $\epsilon$, if we obtain the lower and upper bounds $\widetilde{d^{lb}},\widetilde{d^{ub}}$ for the maximum deviation from the BQS constructed from the key points only and without considering the uncertainty, we then have the lower and upper bounds on the new maximum error $d^{lb},d^{ub}$ as
\begin{eqnarray}
d^{lb} = \widetilde{d^{lb}} - 2\epsilon, d^{ub} = \widetilde{d^{ub}} + 2\epsilon
\end{eqnarray}
\end{theorem}
\begin{proof}
The furthest distance possible from the lines $\overline{u_1u_2}$ and $\overline{c_2u_2}$ is at the round corner at $u_2$ which will always satisfy $d(p,\overline{se}) \leq d(u_2,\overline{se}) + \epsilon$. Then since $\overline{se}$ has an uncertainty of $\epsilon$ around itself, we have Theorem \ref{thm:ucbqs}.
\end{proof}
}

\begin{algorithm}[thp]
\ContinuedFloat
\caption{The Amnesic sinking Algorithm (Continued)}

\textbf{Procedure: $\mathbf{compress}$()}\\
\textbf{Input:}  $s^{src}, e^{src}, a, s^{dest}$ \COMMENT{\textcolor{blue}{src start\&end, age, dest start}}
{\begin{algorithmic}[1]
\algrestore{testcont} 
\makeatletter
\setcounter{ALG@line}{0}
\makeatother
\STATE $X \leftarrow S[s^{src}:e^{src}]$
\STATE $\epsilon^{prev} \leftarrow m^{a+1}\epsilon$ \COMMENT{\textcolor{blue}{compute new tolerance}}
\STATE $X' \leftarrow \mathbf{pbqs}(X, \epsilon^{prev}, \epsilon)$  \COMMENT{\textcolor{blue}{perform compression}}

\COMMENT{\textcolor{blue}{store results to dest location}}
\STATE $ S[s^{dest}:s^{dest}+\|X'\|-1] \leftarrow X'$
\STATE $\mathbf{update\_index}(a,-1,-1)$	\COMMENT{\textcolor{blue}{remove youngest gen from index}}

\COMMENT{\textcolor{blue}{update index for ageing results }}
\STATE $\mathbf{update\_index}(a+1,s^{dest},s^{dest}+\|X'\|-1)$
\end{algorithmic}
}

\textbf{Procedure: $\mathbf{trigger}$()} \COMMENT{\textcolor{blue}{ageing triggering}}
{\begin{algorithmic}[1]
\IF{$\|I\|>0$}

	\COMMENT{\textcolor{blue}{full with  age 0 points}}
	\IF{$I[0].a == 0 ~\&\&~ I[0].e==N-1$} 
			\STATE \textbf{return true}

	\COMMENT{\textcolor{blue}{trigger threshold reached}}			
	\ELSIF{ $I[0].a > 0 ~\&\&~ I[0].e>N-k-I[0].a$ }
		\STATE \textbf{return true}
	\ENDIF
\ENDIF
\STATE  $\textbf{return false}$
\end{algorithmic}
}



\textbf{Procedure: $\mathbf{update\_index}$()}\\
\textbf{Input:}  $a, s, e$ \COMMENT{\textcolor{blue}{age, start\&end locations}}
{\begin{algorithmic}[1]
\makeatletter
\setcounter{ALG@line}{0}
\makeatother

\IF{index for age $a$ exists} 
		\STATE update index of age $a$
	\ELSE  
		\STATE insert index for age $a$
	\ENDIF
\eat{
\IF{$\|I\|>0$}
	\STATE $ i \leftarrow \underset{i}{\operatorname{argmin}} { I[i].a == a}$ \COMMENT{\textcolor{blue}{find index for age $a$}}
	\IF{$i\neq -1$} \COMMENT{\textcolor{blue}{age $a$ index exists}}
		\STATE update index of age $a$
	\ELSE  \COMMENT{\textcolor{blue}{age $a$ index does not exists}}
		\STATE $i \leftarrow \underset{i}{\operatorname{argmin}}{ I[i].a < a}$    \COMMENT{\textcolor{blue}{find insert position for age $a$}}
		\STATE insert index for age $a$
	\ENDIF
\ELSE
		\STATE insert index for age $a$
\ENDIF	}
\end{algorithmic}
}
\end{algorithm}

We devise a new \emph{Progressive BQS} (PBQS) algorithm that supports an existing tolerance, with the main contents in Algorithm \ref{alg:1}. With Theorem \ref{thm:ucbqs}, the bounds from Line 5 is now updated by $d^{lb} \leftarrow d^{lb}-2\epsilon^p, d^{ub}\leftarrow d^{ub}+2\epsilon^p$, where $\epsilon^p$ is the previous error tolerance on the compressed trajectories, passed as an additional argument for the PBQS algorithm besides the new tolerance $\epsilon^d$ ($\epsilon^d>\epsilon^p$). Interestingly, the algorithm now obtains a progressive property. That is, once the existing error tolerance is known, the next compression is independent of any previously applied compression on the trajectory data. The compressed segments resulted from PBQS can again be the input of the next compression pass that has a greater error tolerance than the current. The implication here is that theoretically the tracking node's storage space will support operation of an indefinite period without hard data loss (data overwritten due to space limit). Instead, it will introduce graceful degradation to the aged data over time. 
\setlength{\textfloatsep}{1pt}

\begin{figure*}[htp]
\centering
\vspace{-0.8cm}
\includegraphics[width=14cm]{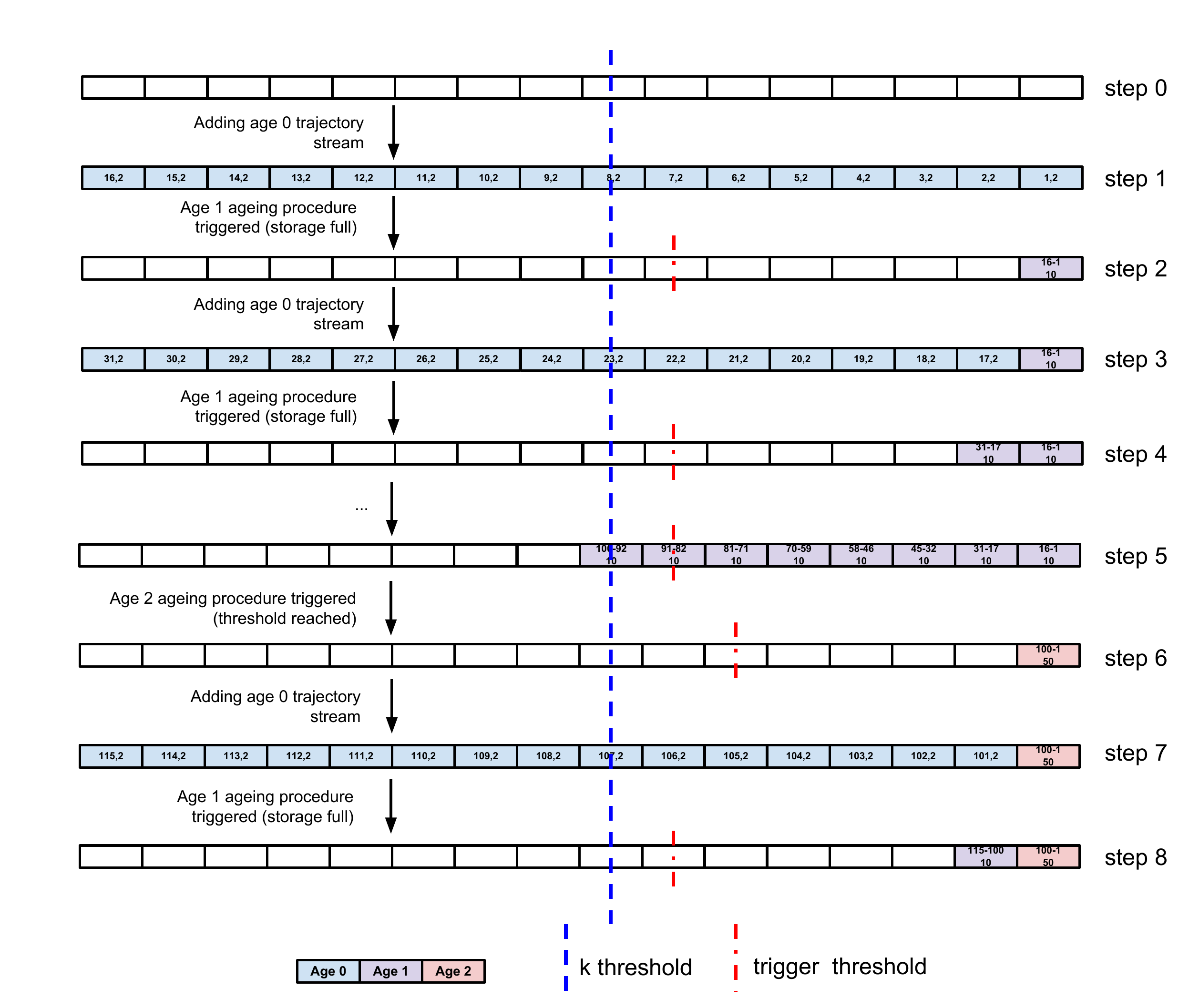}
\caption{Amnesic sinking: The Inline Ageing Procedure}
\label{fig:ab}
\vspace{-0.4cm}
\end{figure*}

Now to address the second challenge, the PBQS, which aims at making compression decisions for individual segments, evolves to a sophisticated framework called Amnesic BQS (ABQS) that manages a certain amount of storage space and seeks to maximize the trajectory information when the number of points far exceeds the space limit and the exact number of points is unknown at the begining. The framework is designed to determine when and how the Progressive BQS is applied on the aged data, and is designed to optimize the storage arrangements for the co-existence and handling of trajectories with difference ages and error tolerances.

The framework utilizes a few ideas:

\begin{enumerate}
\vspace{-0.2cm}
\item we should store as many ``younger generations'' as possible to improve the precision of more recent data.
\item as the storage space runs out, we always try to compress the youngest generation as that results in the least precision loss and storage access.
\item mechanisms are needed so that a ``younger generation'' can always be turned to an ``older generation'' with meaningful compressions (to a reduced number of points).
\item in the process of adding new points and making further compressions,  access to the storage should be minimized when possible for energy saving purposes \cite{DBLP:conf/ipsn/Nath09}.
\vspace{-0.2cm}
\end{enumerate}
In addition, for this framework we still face the limits in memory and computation power of the platform, so the framework still needs to be computationally efficient. Hence we propose ABQS in Algorithm \ref{alg:3}.

There are a few shared variables: $S$ is the entire storage space under ABQS' management, $S[i]$ means the $i^{th}$ slot of the storage. $I$ is a list that keeps track of storage segmentations for different ages, with the entries in ascending order by age (ages without any point of that age will not appear in the index). Each entry in I has three properties, start location $s$, end location $e$ and age $a$. For example, $I[0] = <0,12, 1>$ means currently the youngest generation in the storage is of age 1 (this is the case when all age 0 points in the storage are turned to age 1), and is occupying the $0^{th}$ to $12^{th}$ (inclusive) slots in the storage. $k$ is a fixed number to set the minimum buffer size for the age 0 trajectory points, which also is used to determine the triggering decision for the ageing procedure. $\epsilon$ is the starting error, $m$ is the error multiplier when data age increases, and $N$ is the storage space limit. The whole algorithm is described by five procedures in Algorithm \ref{alg:3}. Note that in Algorithm \ref{alg:3} the indices start from 0 and are inclusive.

\textbf{abqs} is the entry point of the procedure, i.e. whenever an initial data point with the lowest error tolerance (could be 0) is passed through to the framework, \textbf{abqs} will process it and at the same time maintain the structure of the entire storage. Thus \textbf{abqs}' main responsibilities are to find a proper storage space for the new point, to update the index structure, and to invoke the amnesic sinking procedure for further compression when necessary. 

\textbf{amnesic\_sinking} is the main ageing procedure for further compression and freeing up storage spaces. We call it ``amnesic sinking'' because it guarantees that older data that has greater tolerances will sink to the ``bottom'' of the storage, while younger data will remain closer to the top. From bottom to top, the tolerances decrease monotonically. The procedure iteratively checks a trigger flag $f$ which indicates whether it is necessary to perform a further compression on aged data, and then invokes \textbf{compress} when needed.

\textbf{compress} performs progressive compression with an extended error tolerance from an existing one. The inputs $s^{src}, e^{src}$ specify the indices of the data segment to be compressed in the storage, $a$ specifies the current age of the data to be further compressed, and $s^{dest}$ specifies the location in the storage where the further compressed data will be stored. Line 2 shows how the tolerance is updated for each generation. Lines 3 and 4 compress and age the existing trajectories to an older generation. Lines 5 and 6 remove the index for the generation being processed, and add the index for the resulted generation after the processing. Note that here Lines 1-4 are illustrated as a batch/block update but in practice it is implemented in an online/inline style so that the space complexity is still constant. That means the procedure will process a point at a time, similar to the way standalone BQS operates, and store the results in the space originally taken by the trajectory before the compression to avoid unnecessary copy and move in the storage.

\textbf{trigger} checks whether \textbf{compress} is needed for the youngest generation in the storage. The decision is made to fulfil the requirements on the space arrangement that the data points of the youngest generation with age $a$ cannot go beyond the limit of $I[0].e>N-k-I[0].a$ (Line 4) to maintain meaningful compression results for each generation. $k$ is used to reserve enough space for incoming age 0 points. For example, if we assume that the points passed to the ABQS framework have an initial error of 2m for the age 0 points, we can set this $k$ to be the average number of points included in an age 1 segment with 5m tolerance from empirical results. The effect of this parameter is to reduce occurrences in which the points are kept in the aged results not because the tolerance is reached but because the end of the youngest generation's segment is reached. 

$I[0].e>N-k-I[0].a$ follows the same rationale. The reserved space for younger generations is increased when the age of the currently youngest generation increases, so that there will be at least three points for any younger generation when they trigger an ageing procedure and receive a further compression (performing compression on two or less points will not result in any compression). If the youngest generation is of age $a$, then ABQS needs to reserve one slot for the immediately younger generation $a-1$, because for age $a-1$ there will be at least 2 points from $a-2$, so this single served slot guarantees the meaningful compression for this generation. Similarly, we need 1 slot reserved for $a-2$. Recursively, we can get the total number of reserved slots for age $a$ to be $a-1$ and consequently the maximum index for age $a$ as  $N-k-1-(a-1)=N-k-a$, because the $k$ slots are also reserved for the incoming age 0 points.

\textbf{update\_index} maintains the index structure which records the segmentation of storage for data points with different ages. The procedure does a few operations: in Lines 2-4 it attempts to find an existing index entry for the specified age $a$, and update its start and end locations to $s$ and $e$ on success. When a new entry is needed for an unrecorded age, Lines 6,7, and 10 try to insert the entry and maintain the ascending order by age of the index entries.

There are several advantages for the ABQS framework. First there is no theoretical limit on the maximum operational time (without hard data loss) of the tracking device, as the ABQS adaptively manages the storage space and trades precision off storage space. Second the entire procedure only uses little memory space (memory for a standard BQS plus three integers for each index entry). Third the compression error for each generation is still known and guaranteed. 

We illustrate this algorithm in Figure \ref{fig:ab}. Here the initial error tolerance $\epsilon$ is set to $2$ and the multiplier $m$ to $5$. The empty cells represent available storage spaces, and the cells with color and numbers mean occupied storage space. The storage is used from right (bottom) to left (top). Each occupied cell is labeled with two numbers $<data~point~id~range, error~tolerance>$ (ranges start with 1 and are inclusive), and the background colors of the cell indicate the age of the data. For example, in the third step, the right most cell has $<16-1,10>$ in it as well as a light blue background, indicating that the storage block now stores age $1$ compressed segments covering the original points $1-16$, with an error tolerance of $10$. The blue and red dashed lines are the $k$ threshold and the  trigger threshold defined by $I[0].e>N-k-I[0].a$ (Line 4 in \textbf{trigger}). Step 1 means when the space is filled with age 0 data then an ageing process is invoked which compresses data points $1-16$ into a segment with error tolerance $10$ and stores the results to the bottom of the storage. In step 3 the storage is full again and age 0 data points are compressed again to the adjacent storage to $<16-1,10>$. Step 5 illustrates that when the age 1 compressed segments take over a great part of the storage and leave insufficient space for younger data, they are further compressed. In step 6 we see that because the youngest generation becomes age 2, the trigger threshold moves rightwards, reserving more spaces for age 1 and age 0, whereas in step 8 the trigger threshold moves back to the $k$ threshold because age 1 is now the youngest generation.

\eat{
\subsection{Data Mining from Trajectories of Mixed Ages: A Case Study}

\subsection{Maintenance Procedures}
In addition to the compression itself, the BQS framework employs two
techniques to further reduce the storage space required for the
historical data, namely error-bounded merging and error-bounded
ageing. Due to space limitations, we only give a brief description of the
techniques. More details will be incorporated in an extended version
of this work. 

Merging is a procedure in which the newly compressed
segment is used as a query to search
similar historical segments in the trajectory database. If any
existing compressed segment could represent the same path with a minor error,
the new segment is considered duplicate information and is merged
into the existing one. 

Ageing is based on the intuition that newer and 
older trajectories should not
bear the same significance in the historical trajectory database.
More recent trips represent the moving object's recent travel
patterns better and should be regarded of greater interest.
Hence the ageing procedure re-runs the compression algorithm on the existing
trajectories that are already compressed, but with a greater error
tolerance, so that the compression rate will be further improved.

The procedures are aligned with our goal, namely
capturing the spatio-temporal characteristics such as the time windows,
destinations and routes for the major movements of the moving
object. By making a more aggressive trade-off between accuracy and
coverage, it further extends the frameworks capability to capture
mobility patterns over long periods.
}



\section{Experiments}
\label{sec:exp}
In this section we evaluate the performance of the
proposed BQS family and ABQS framework.

\vspace{-0.3cm}
\subsection{Dataset}
We use three types of data, namely the flying fox (bat) dataset, the vehicle dataset, and
the synthetic dataset. The two real-life datasets comprise of $138,798$ GPS samples, collected by 6 Camazotz
nodes (five on bats and one on a vehicle). The total travel distances for
the bat dataset and vehicle dataset are $7,206~km$ and $1,187~km$
respectively. The tracking periods spanned six months and two weeks for the bat and vehicle datasets respectively.

\eat{
\begin{figure}[htp]
\hspace{-0.6cm}
\begin{tabular}{c}
 \subfigure[Bat Tracking Data]{
\includegraphics[height=2cm]{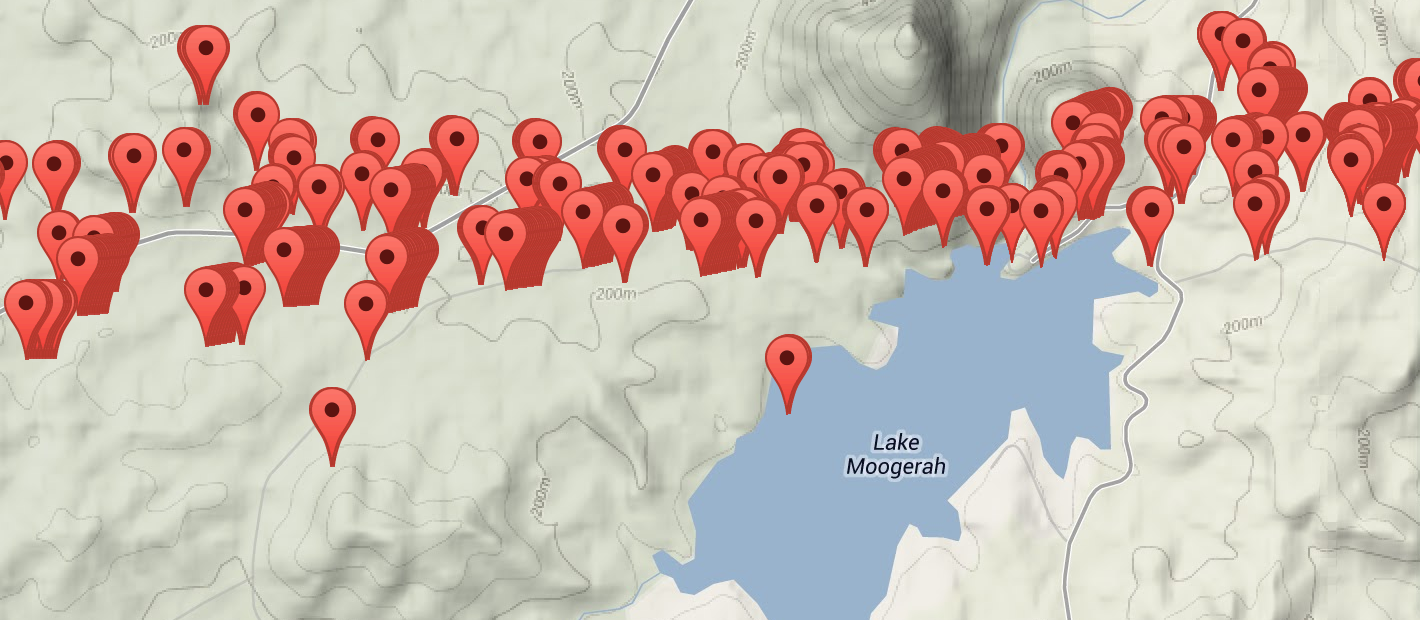}
}
\subfigure[Vehicle Tracking Data]{
\includegraphics[height=2cm]{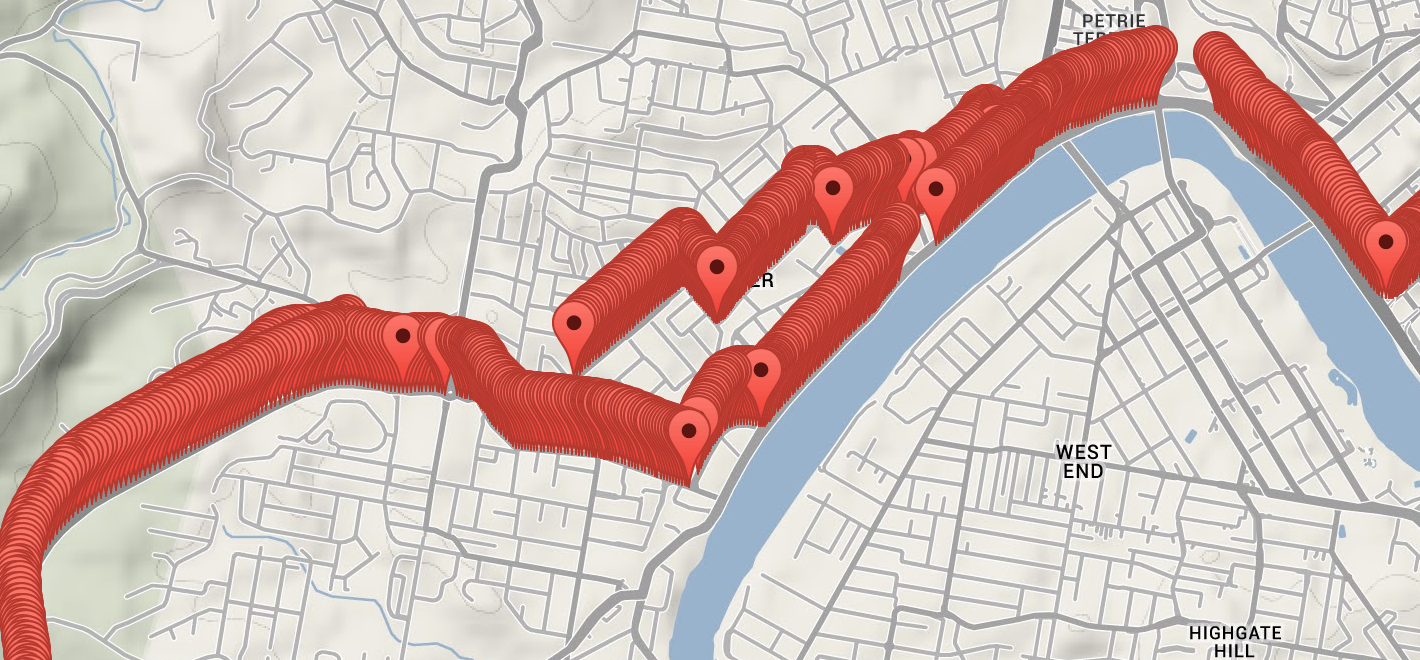}
}
\end{tabular}
\caption{Sample Traces}
\label{fig:sp}
\end{figure}
}

\begin{figure}[htp]
\hspace{-0.8cm}
\begin{tabular}{c}
\subfigure[Bat Tracking Data]{ \label{fig:cr1}
\includegraphics[height=3.5cm]{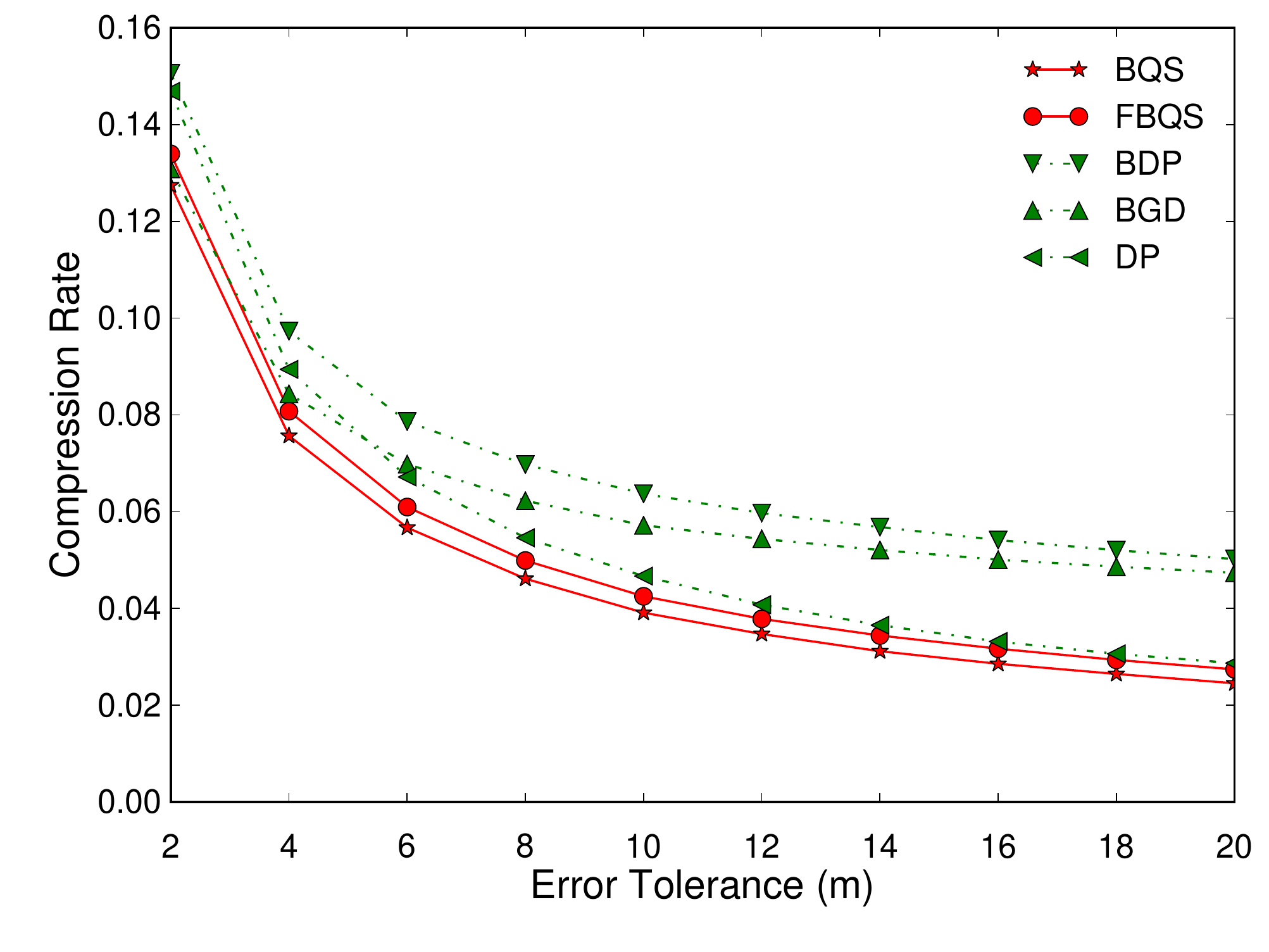}
}
\subfigure[Vehicle Tracking Data]{ \label{fig:cr2}
\includegraphics[height=3.5cm]{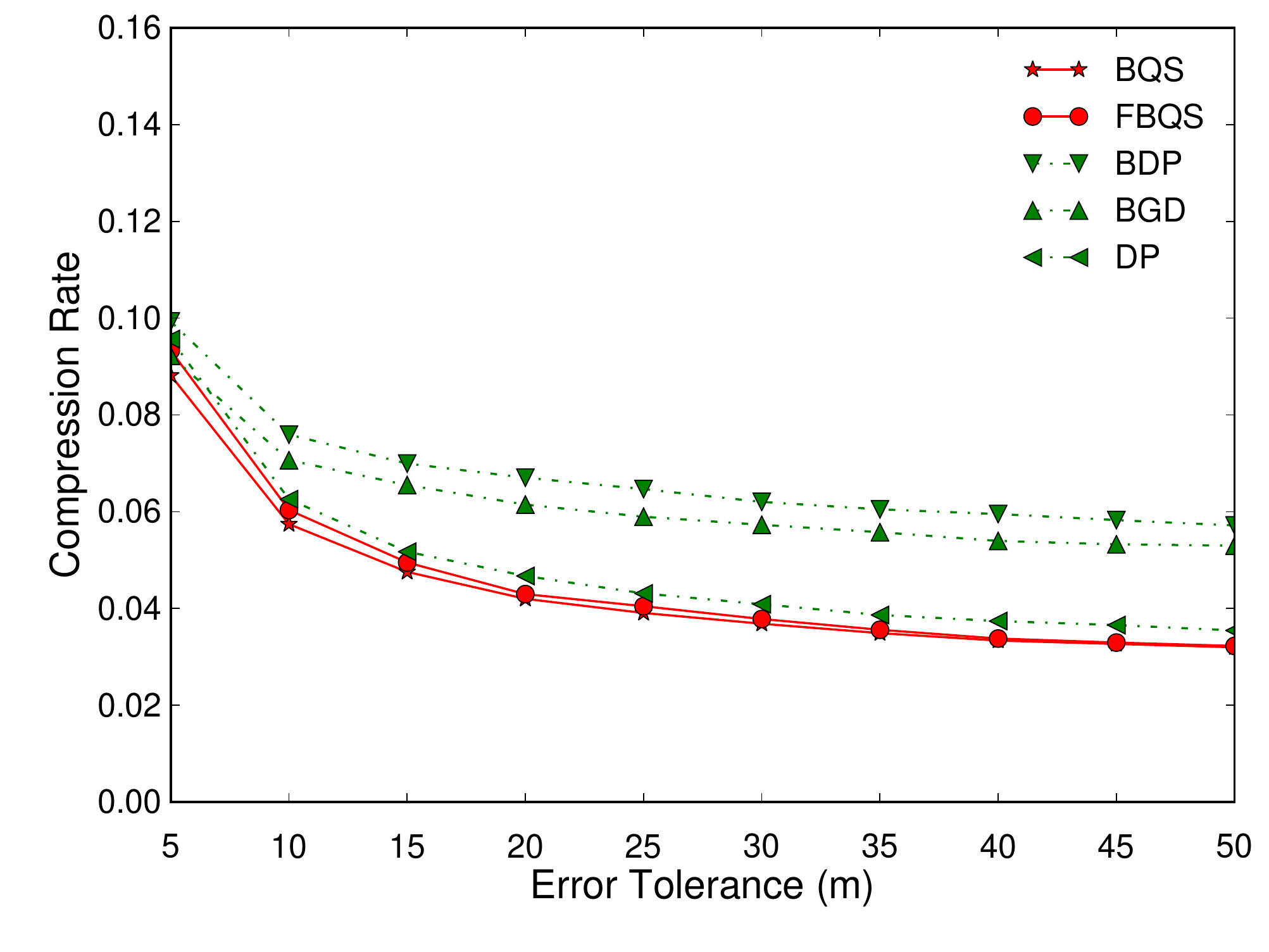}
}
\end{tabular}
\caption{Comparison of Compression Rate on Real-life Datasets (The lower the better)}
\end{figure}

Note that there are couple of differences between the two
datasets. The vehicle dataset shows larger scales in terms of travel distances well as moving speed. For instance, the length of a car
trip varies from a few kilometers to $1,000~km$ while a
trip for flying-foxes are usually around $10~km$. The car can travel
constantly at $100~km/h$ on a highway or $60~km/h$ on common roads,
while the common and maximum continuous flying speeds
for a flying-fox are approximately $35~km/h$ and $50~km/h$. In regard
to these differences, the two datasets are evaluated with different
ranges of error tolerance. With a much greater spatial scale of the
movements, the error tolerance used for the vehicle dataset is generally greater.

The vehicle dataset also shows more consistency in the heading
angles due to the
physical constraints of the road networks. On the contrary,
the bats's movements are unconstrained in the 3-D space, so their
turns tend to be more arbitrary. \eat{Figure \ref{fig:sp} shows 
segments of the actual trajectories from both datasets.}
We argue that by performing extensive experiments on both
datasets, the robustness of BQS and ABQS is demonstrated.

The synthetic dataset consists of data generated by a statistical model that anchors
patterns from real-life data, special trajectories that follow certain shapes.
The comparison between \emph{Dead Reckoning}~\cite{Trajcevski06on-linedata} and
\emph{Fast BQS} (FBQS) is conducted on this dataset. This is because continuous
high-frequency samples with speed readings are required to implement
DR in an error-bounded setting, while such data is lacking in the
real-life datasets. The model uses an event-based correlated 
random walk model to simulate the movement of the object. In the
simulation, waiting events and moving events are executed
alternately. The object stays at its previous location during a
waiting event, and it moves in a randomly selected speed and turning
angle for a randomly selected time. Note that the speed
follows the empirical distribution of speed, the turning angle 
is drawn from the von Mise distribution \cite{Risken}, while the move time is 
exponentially distributed, corresponding to the Poisson process. The
trajectories are bounded by a rectangular area of $10~km\times 10~km$,
and the speed and turning angle follow approximately the distributions
of the bat data. A total of $30,000$ points are generated by the
model. The synthetic dataset also contains trajectories of special shapes to
explicitly test the robustness of the proposed methods.

\subsection{Experimental Settings}
The evaluation is conducted on a desktop computer, however the extremely 
low space and time complexity of FBQS makes it 
plausible to implement the algorithms on the platform aforementioned
 in Section \ref{sec:mot} (32 KBytes ROM, 4 KBytes RAM). In particular, if we
 look into the FBQS algorithm, we only need tiny memory space to store
 at most 32 points besides the program image itself (4 corner points 
and 4 intersection points for each quadrant).

Two main performance indicators,
namely compression rate and pruning power are tested on the real-life datasets to evaluate the BQS family.
We define compression rate as $\frac{N^{compressed}}{N^{original}}$ where $N^{compressed}$ is
the number of points after compression, and $N^{original}$ is the number of
points in the original trajectory. Pruning power is defined as $1
-\frac{N^{computed}}{N^{total}}$, where $N^{computed}$ and $N^{total}$
are the number of full deviation calculations and the number of total
points respectively. 

For compression rate, we perform comprehensive comparative study to show
BQS's superiority over the other three methods, namely
\emph{Buffered-DouglasPeucker} (BDP), \emph{Buffered-Greedy} (BGD) and \emph{Douglas-Peucker} (DP). DR is compared against FBQS on the synthetic
dataset. For buffer-dependent algorithms, we set the buffer size to be
32 data points, the same as the memory space needed by the FBQS
algorithm to hold the significant points. 

To intuitively demonstrate the advantage of FBQS in compression rate, 
we provide comparison showing the number
of points taken by the FBQS algorithm and the DR algorithm 
on the synthetic dataset. We also show the estimated
operational time of tracking devices based on such compression rate. Finally,
we study comparatively the actual run time efficiency of FBQS.

The ABQS framework is evaluated with extensive experiments, for which the details are given in Section \ref{sssec:expabqs}.

\eat{
For all datasets, we combine all the data points into a single data
stream and use it to feed the algorithms. Then we calculate the
pruning power, compression rate and number of points used.
}
\begin{figure}[htp]
\centering
\begin{tabular}{c}
 \subfigure[Bat Tracking Data]{ \label{fig:pp1}
\includegraphics[height=3cm]{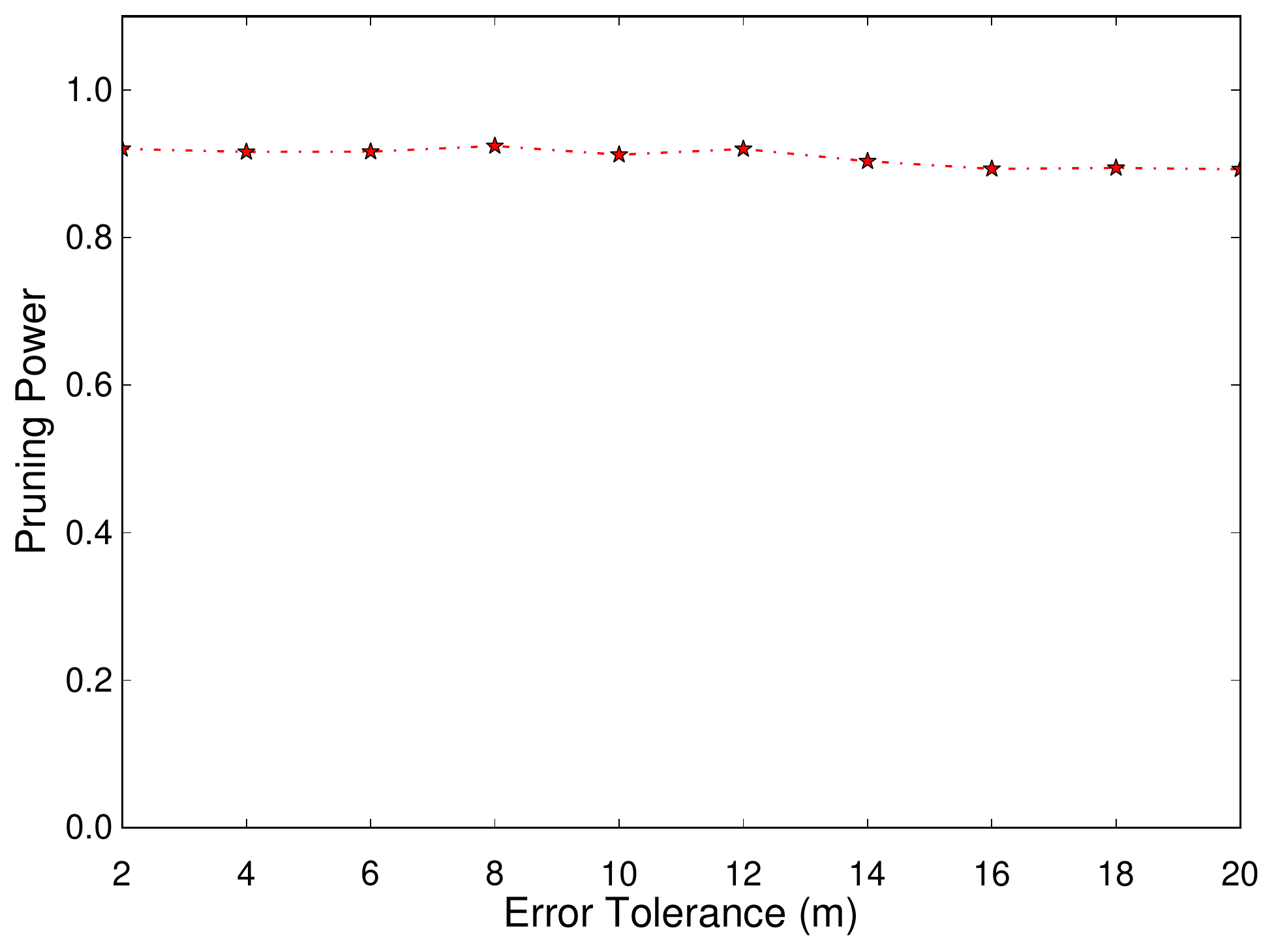}
}
\subfigure[Vehicle Tracking Data]{ \label{fig:pp2}
\includegraphics[height=3cm]{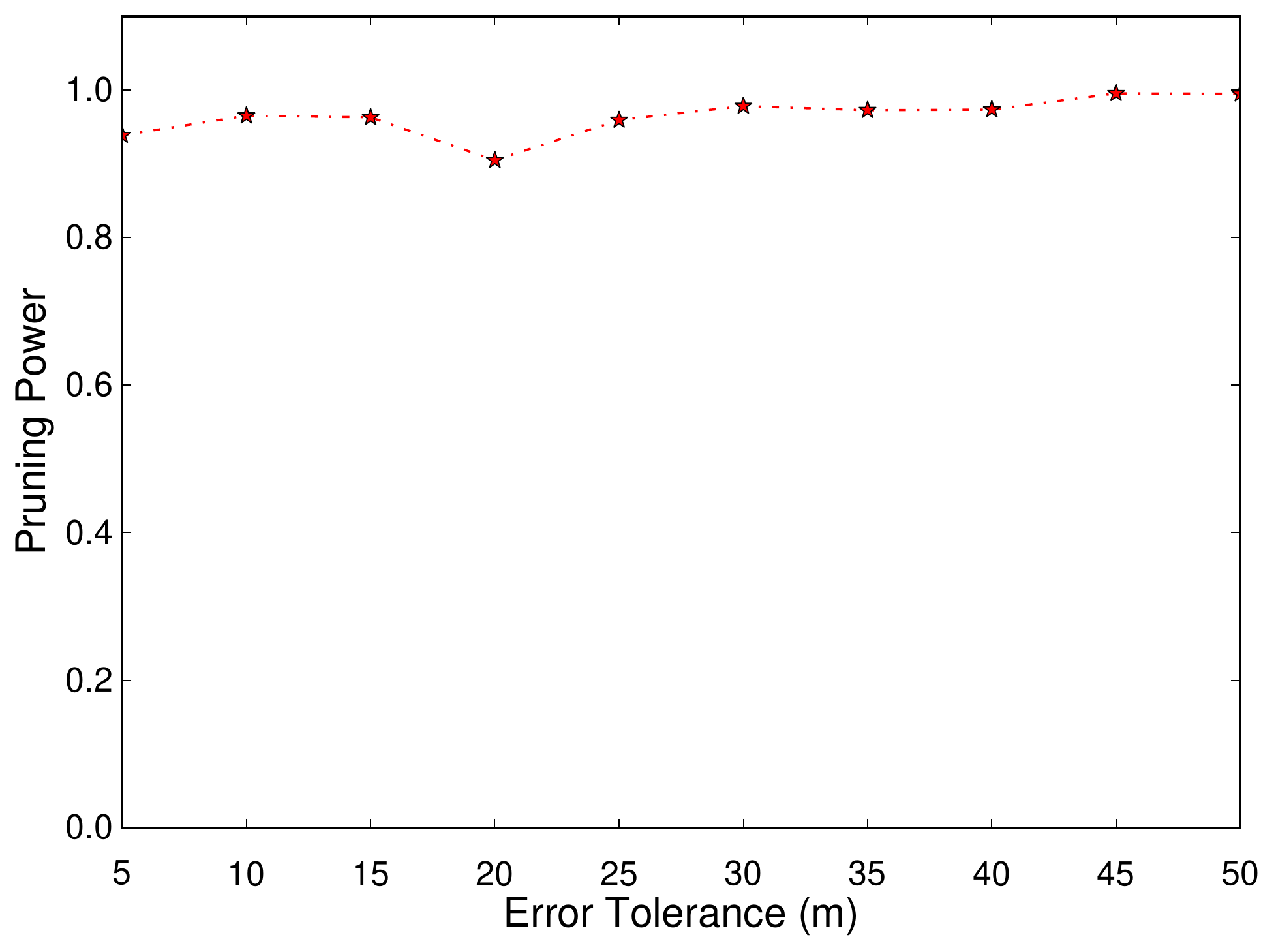}
}
\end{tabular}
\caption{Pruning Power of the BQS Algorithm (The higher the better)}
\end{figure}

\subsection{Compression Algorithm}

\subsubsection{Compression Rate on Real-life Data}

Compression rate is a key performance indicator for trajectory compression
algorithms. Here we conduct tests on the two real-life datasets. We
compare the performance of five algorithms, namely BQS, FBQS,
BDP, BGD and DP. All of the algorithms give error-bounded results.
The former four are online algorithms, and the last one is offline. 
The compression rates are illustrated in Figures \ref{fig:cr1} and \ref{fig:cr2}.

Evidently, BQS achieves the highest compression rate among the five
algorithms, while BDP and BGD constantly use approximately $30\%$ to $50\%$ more
points than BQS does. FBQS's compression rates swing between BQS's
and DP's, showing the second best overal performance. 
BDP has the worst performance overall as it inherits the 
weaknesses from both DP and window-based approaches. BGD's
performance is generally in between DP and BDP, but it still suffers
from the excessive points taken when the buffer is full.

Comparing the two figures, it is worth noting that all the algorithms
perform generally better on the bat data. Take the results at
$10~m$ error tolerance from both figures for example, on the bat data
the best and worst compression rates reach $3.9\%$ and $6.3\%$
respectively, while on the vehicle data the corresponding figures are
$5.4\%$ and $7.7\%$. This may seem to contradict the results of the
pruning power. However, it is in fact reasonable because bats perform 
stays as well as small movement around certain locations, making those
points easily
discardable. Hence the room for compression is larger for the bat
tracking data given the same error tolerance.

On the bat data, with $10~m$ error tolerance, BQS and FBQS achieve
compression rates of $3.9\%$ and $4.1\%$ respectively. DP, as an
offline algorithm that runs in $\mathcal{O}(nlogn)$ time, yields a
 worse compression rate than FBQS at $4.6\%$. Despite having poorer
 worst-case complexities, BDP and BGD also obtain worse compression rates 
than BQS and FBQS do at $6.3\%$ and $5.8\%$ respectively.
At this tolerance, the offline DP
algorithm uses approximately $20\%$ and $10\%$ more points than the
online BQS
and FBQS do, respectively. Furthermore, for online algorithms with $20~m$ tolerance,
FBQS (2.7\%) improves BDP (5.1\%) and BGD (4.9\%) by 47\% and 45\% respectively.

The results on the vehicle data show very similar trends
of the algorithms' compression rate curves. Interestingly, with this 
dataset, because the pruning power is in most of the cases
around and above $0.95$ as demonstrated in \ref{fig:pp2}, the
compression rate of FBQS is remarkably close to BQS'. For instance, at
$20~m$, $30~m$ and upwards, the difference between the two is smaller than $1\%$.
This observation supports our aforementioned argument
that the bounds of the original BQS are so effective that the
 number of extra points taken by FBQS is insignificant.

\eat{
As discussed in Section \MakeUppercase{Trajectory
  Compression with Bounded Quadrant System}, it is expected that BQS has
better compression rate than the two buffered online algorithms. Here
we explain why it also outperforms the offline
DP algorithm. DP's compression rate is heavily
affected by the dividing step. When it decides to split the trajectory
into two parts at the farthest point to the line defined by the start and
end points, there is a great likelihood that the split point is
in the middle of a smooth segment which could be represented by a
single compressed segment. However by making the split, DP takes more
points than BQS as BQS always tracks continuous smooth segments to
minimize the number of points taken. 
}
\vspace{-0.2cm}
\begin{figure}[htp]
\centering
\begin{tabular}{c}
\subfigure[Synthetic Dataset (unit:m)]{ \label{fig:syn}
\includegraphics[height=3.5cm]{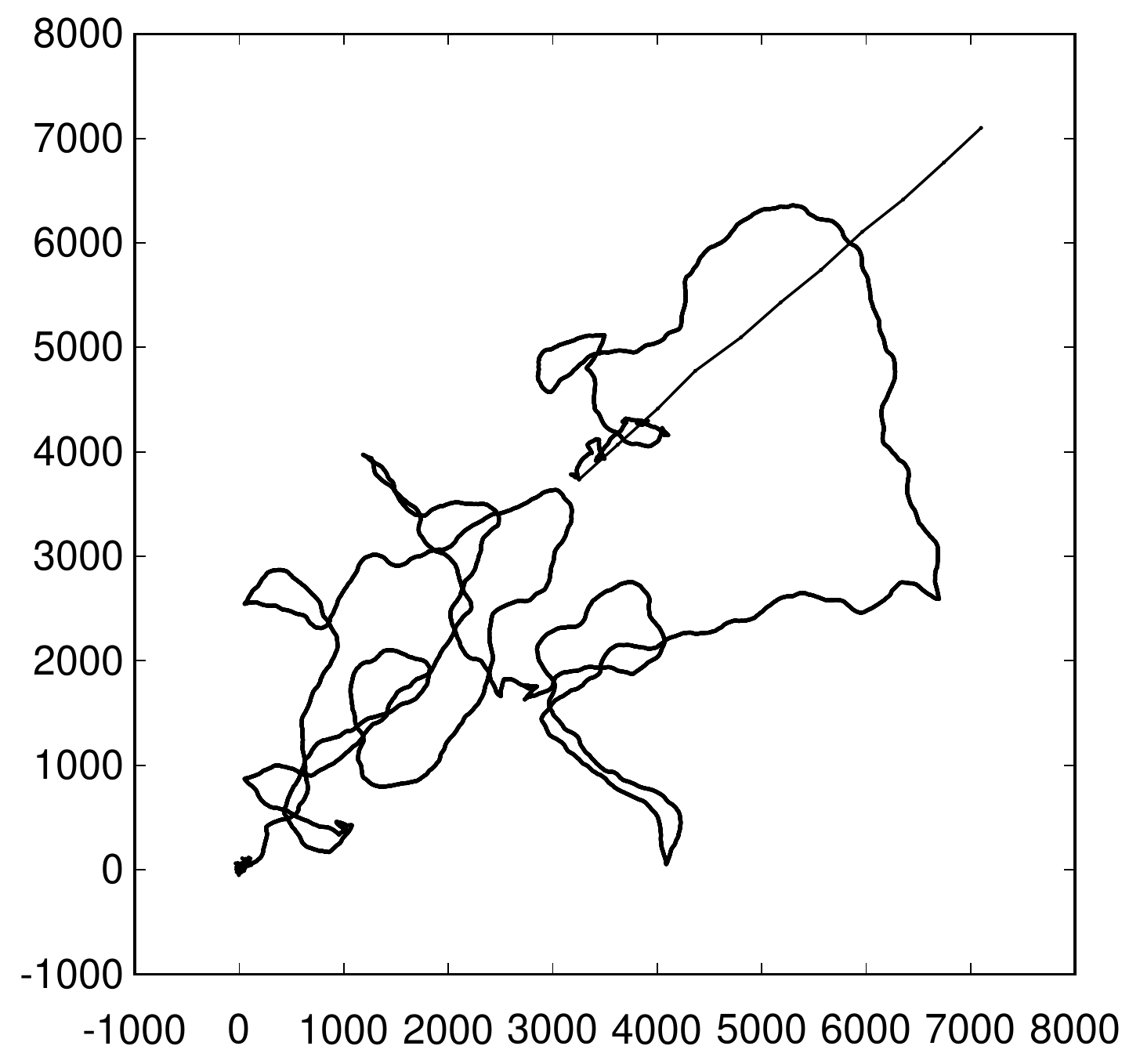}
}
\subfigure[No.Points Used on Synthetic Data]{ \label{fig:nps}
\includegraphics[height=3.5cm]{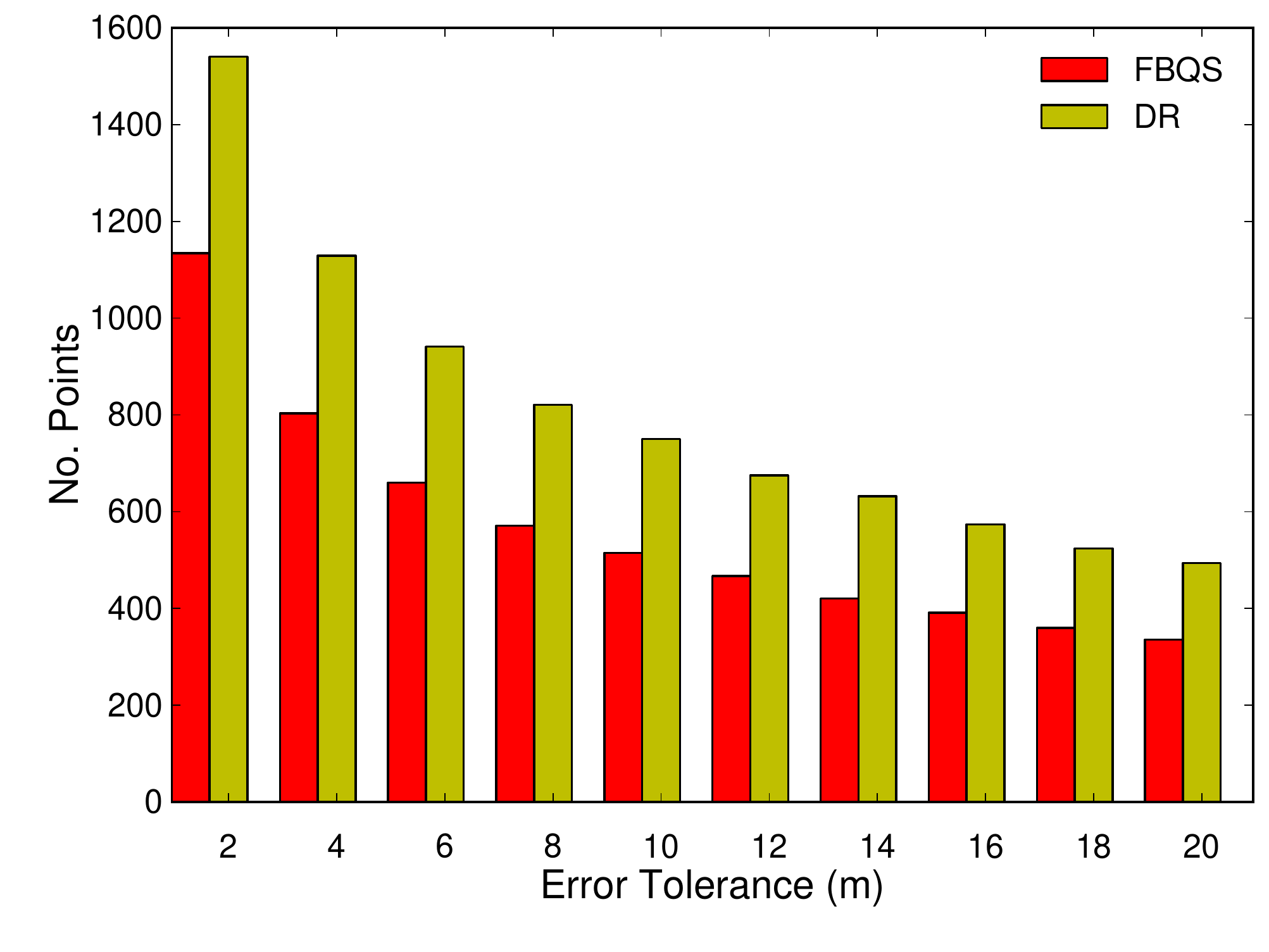}
}
\end{tabular}
\caption{Shape of Synthetic Data and Comparison of Number of Points
  Used (The lower the better)}
\end{figure}

\subsubsection{Pruning Power}

Pruning power determines how efficient BQS and FBQS are. In FBQS, if
the relation between the bounds and the deviation tolerance is
deterministic, FBQS generates a lossless result as
BQS. If it is uncertain, then FBQS will
take a point regardless of the actual
deviation. The pruning power reflects how often the relation is
deterministic, and it indicates how many extra points could be taken by
the approximate algorithm. With high pruning power, 
the overhead of FBQS will be small. Here we investigate
the pruning power of BQS in this subsection. 

Figures \ref{fig:pp1} and \ref{fig:pp2} show the pruning power
achieved by BQS on both datasets. The sensitivity of the algorithm to the error tolerance or to the
shape of the trajectories appears low, 
as the pruning power generally stays above $90\%$ for most of the
tolerance values on both datasets. This means approximately only $10\%$ more
points will be taken in the Fast BQS algorithm compared to the
original BQS algorithm. The running values in Figure \ref{fig:bds} also support
 this observation.

BQS shows higher pruning power
on the car dataset than on the bat dataset. The higher pruning power on the vehicle data is a
result of the physical constraints of the road networks,
preventing abrupt turning and deviations, and making the trajectories
smoother. Naturally the pruning power will be higher as a result of the higher
regularity in the data's spatio-temporal characteristics.

\subsubsection{Comparison with Dead Reckoning on Synthetic Data}
In Figure \ref{fig:syn} we show the simulated trajectories from our
statistical model. Visibly the trajectories show little physical
constraint and considerable variety in heading and turning angles. On
this dataset we study the performance comparison because on this
dataset we are able to simulate the tracking node
closely with high frequency sampling. Hence FBQS is used as 
a light-weight setup to fit such online environment.

We show in Figure \ref{fig:nps} the numbers of points taken after the
compression of $30,000$ points under different error tolerances. With
smaller $\epsilon$ such as $2~m$, DR uses $1,547$ points compared to
$1,122$ for FBQS, indicating that DR needs $37\%$ more points. As $\epsilon$ grows, DR's
performance tends to slowly approach FBQS in absolute numbers yet the
difference ratio becomes more significant. 
At $20~m$ error tolerance, FBQS
only takes $341$ points while DR uses $493$ points, the
difference ratio to FBQS is around $45\%$. 

Evidently, FBQS has achieved excellent 
compression rate compared to other existing online algorithms.

\vspace{-0.3cm}
\subsubsection{Robustness to Trajectory Shapes}

\begin{figure}[htp]
\vspace{-0.5cm}
\hspace{-0.3cm}
\begin{tabular}{c}
 \subfigure[Extreme Cases]{ \label{fig:ec1}
\includegraphics[height=3.3cm]{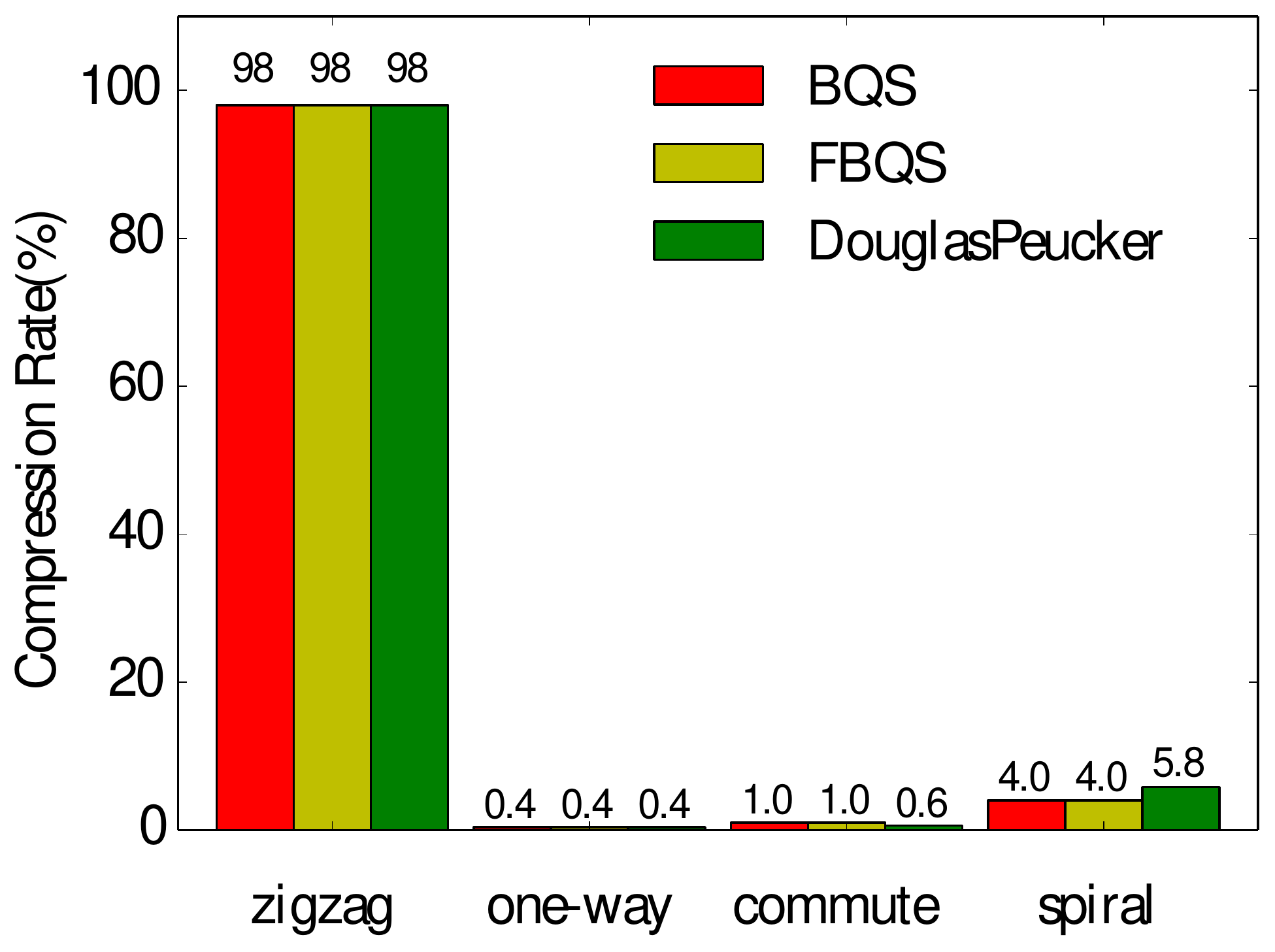}
}
\subfigure[Spiral Trajectory]{ \label{fig:ec2}
\includegraphics[height=3.3cm]{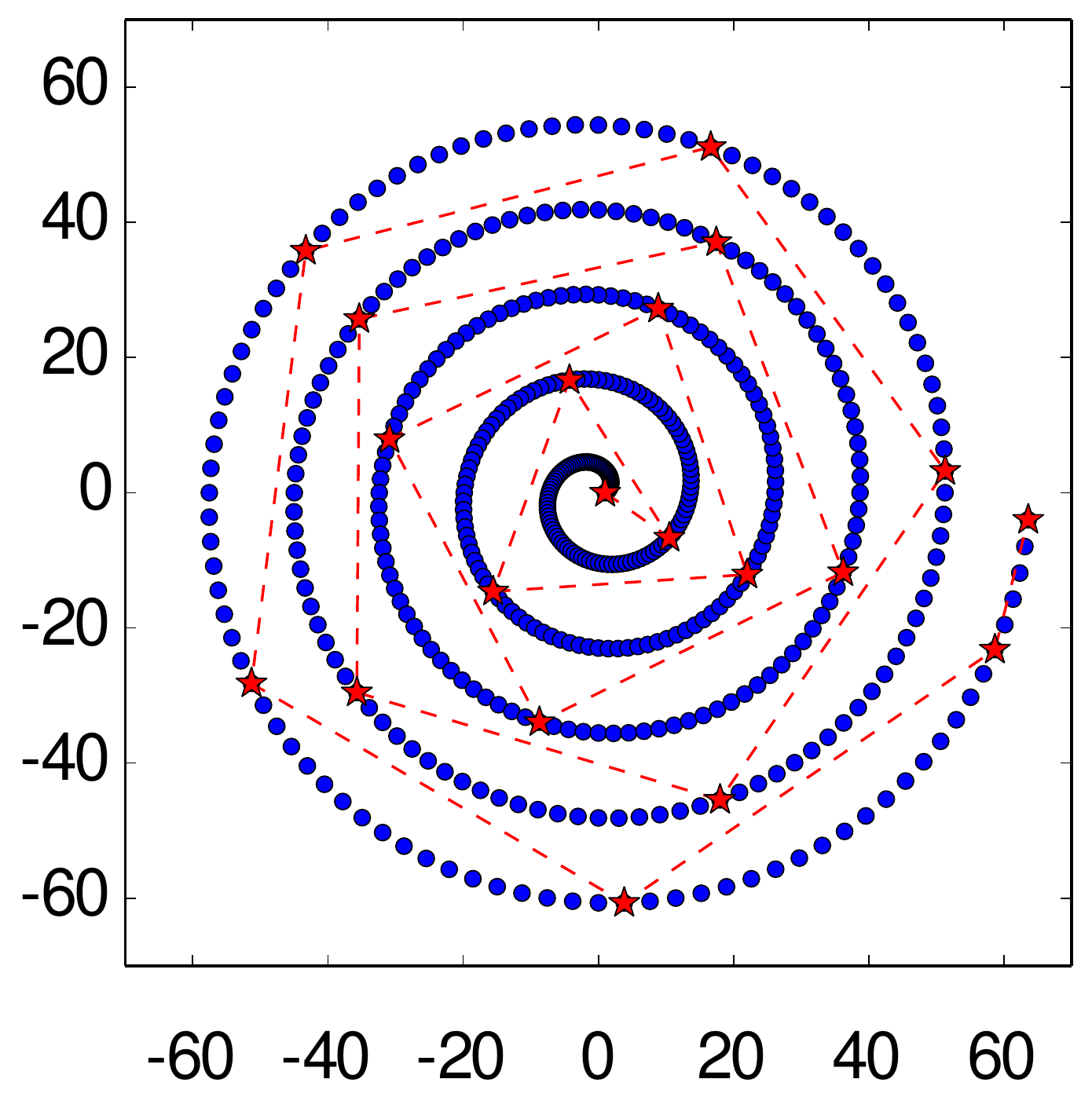}
}
\end{tabular}
\caption{Robustness of BQS}
\vspace{-0.4cm}
\end{figure}

In this experiment we examine BQS' robustness to trajectory shapes. Though in previous experiments 
we have used trajectories with a variety of characteristics such as
speed, spatial range and turning angles, and have verified that BQS achieves competitive performances on
those datasets, here we synthesize a few more with extreme shapes to further examine the robustness of BQS.

We use four synthetic datasets:

\begin{itemize}
\small
\vspace{-0.2cm}
\item \textit{zigzag}: the trajectory travels in a zigzag fashion. $x = <1, 0.5, 3, 0.5....2N-2, 0.5, 2N, 0.5>$, $y = <0.5, 2, 0.5, 4, ....0.5, 2N-2, 0.5, 2N>$.
\item \textit{one-way}: a straight line $x = <1, 2,...,N>, y = <0.5,0.5...,0.5>$.
\item \textit{commute}: moving back and forth on a straight line  $x = <1, 2,...,n, n-1,...1,2,...,n,...>, y = <0.5,0.5...,0.5>$.
\item \textit{spiral}: the Achimedean spiral defined as $\rho =1+2\theta$.
\vspace{-0.2cm}
\end{itemize}
For each synthetic dataset we sample $500$ points with the above strategies and display the compression rates of BQS, FBQS and \emph{Douglas-Peucker} (DP) in Figure \ref{fig:ec1}, with the tolerance $\epsilon=10$. The results of BQS and FBQS are almost identical, suggesting that the conservative decision making in the FBQS introduces little overhead for even the most extreme shapes. The BQS family achieves comparable results with DP on all four datasets. On the first two they are all 98\% and 0.4\%. On the third DP achieves slight better result 0.6\% versus 1.0\% of BQS, but on the fourth BQS has a much better compression rate of 4.0\%, in comparison with 5.8\% from DP. Figure \ref{fig:ec2} demonstrates the Achimedean Spiral with 500 samples in blue dots. The red stars and the dashed line illustrate the key points selected by BQS.

\vspace{-0.2cm}
\subsubsection{Effect on Operational Time of Tracking Device}
Next we investigate how different online algorithms affect the maximum
operational time of the targeted device. This operational time
indicates how long the device can keep records of the locations 
before offloading to a server, without data loss.

\begin{figure*}[htp]
\vspace{-0.6cm}
\centering
\begin{tabular}{c}
 \subfigure[]{ \label{fig:re}
\includegraphics[height=3cm]{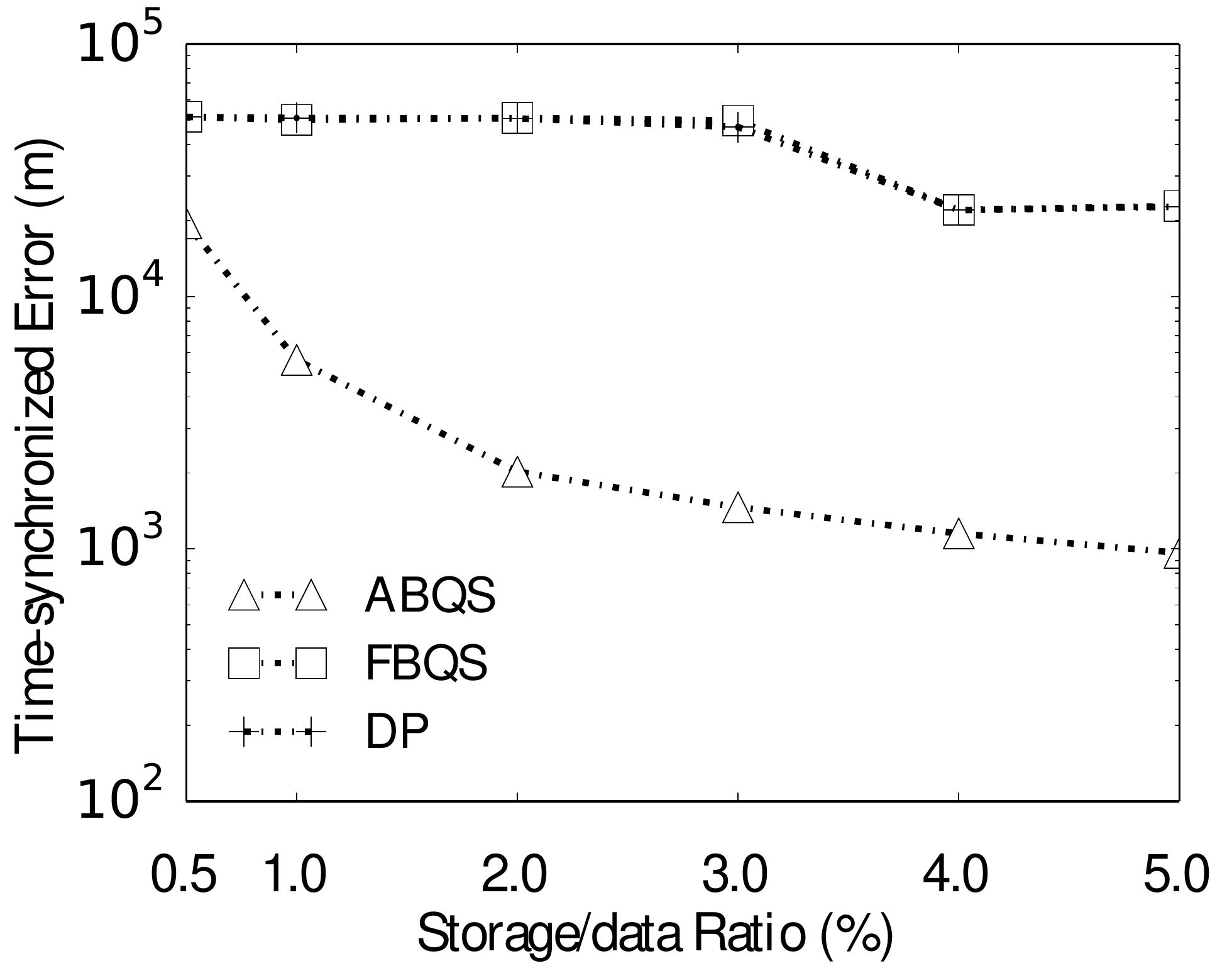}
}
\subfigure[]{ \label{fig:rem}
\includegraphics[height=3cm]{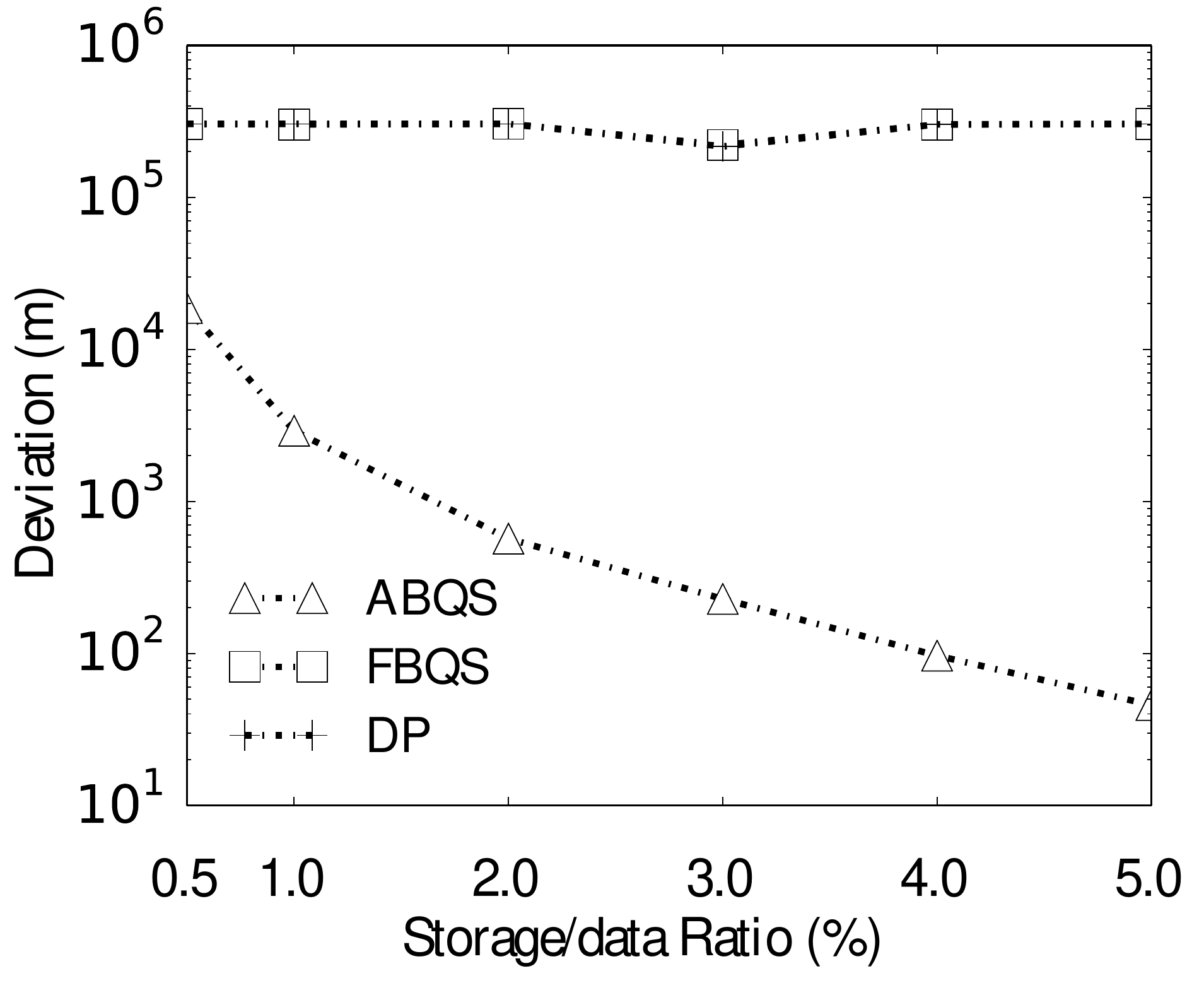}
}
 \subfigure[]{ \label{fig:ee}
\includegraphics[height=3cm]{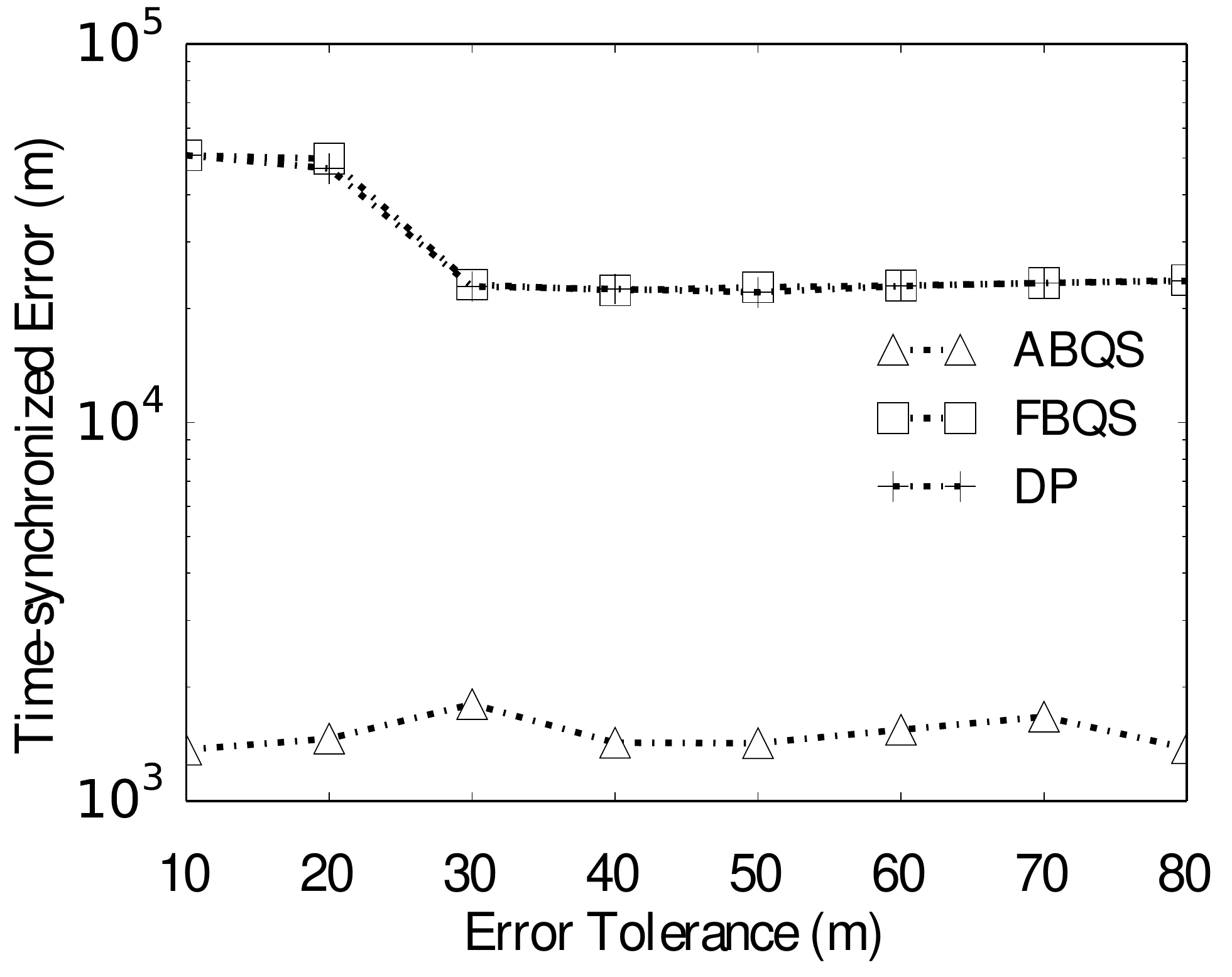}
}
\subfigure[]{ \label{fig:eem}
\includegraphics[height=3cm]{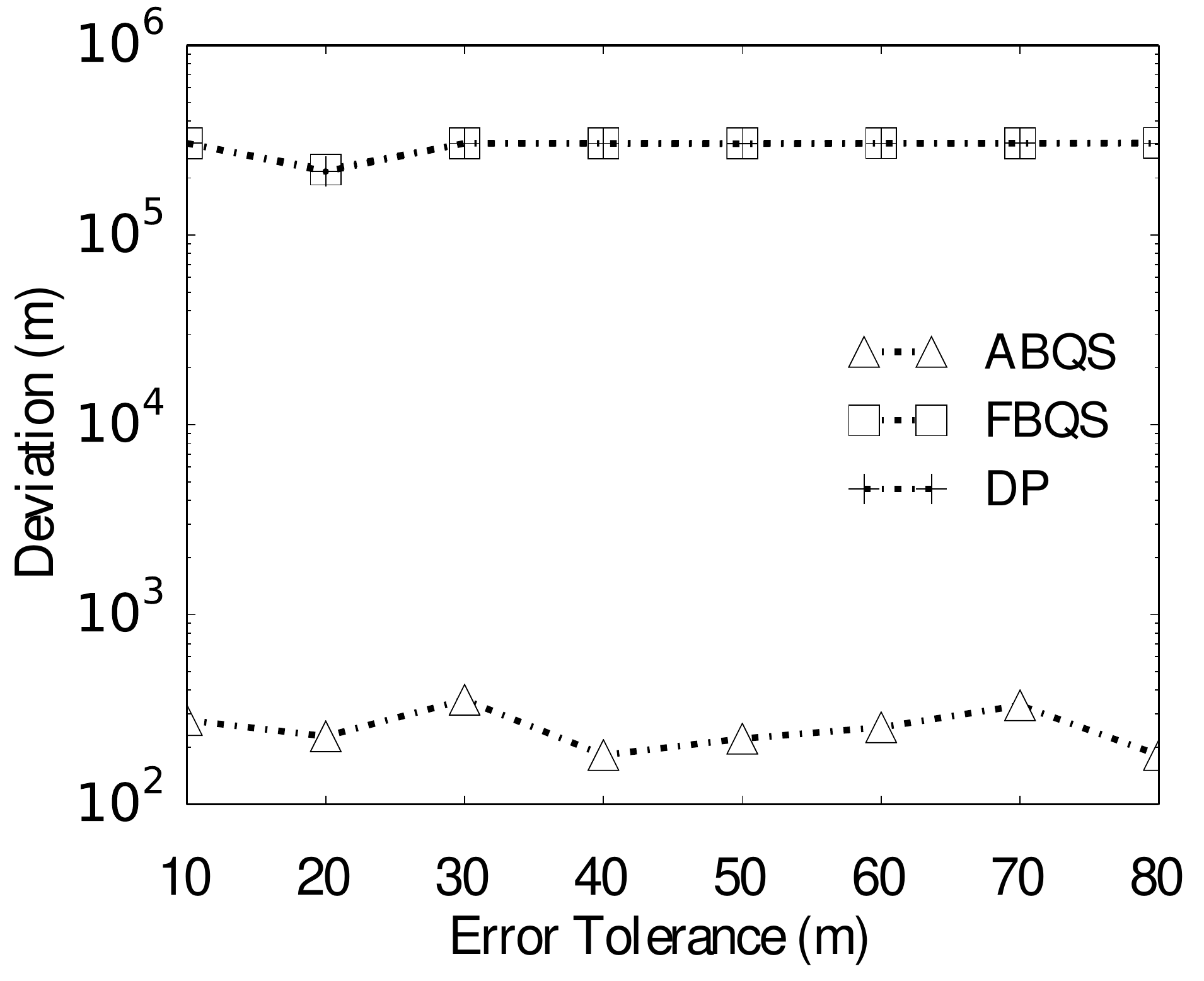}
} 
\vspace{-0.3cm}
\\
 \subfigure[]{ \label{fig:ne}
\includegraphics[height=3cm]{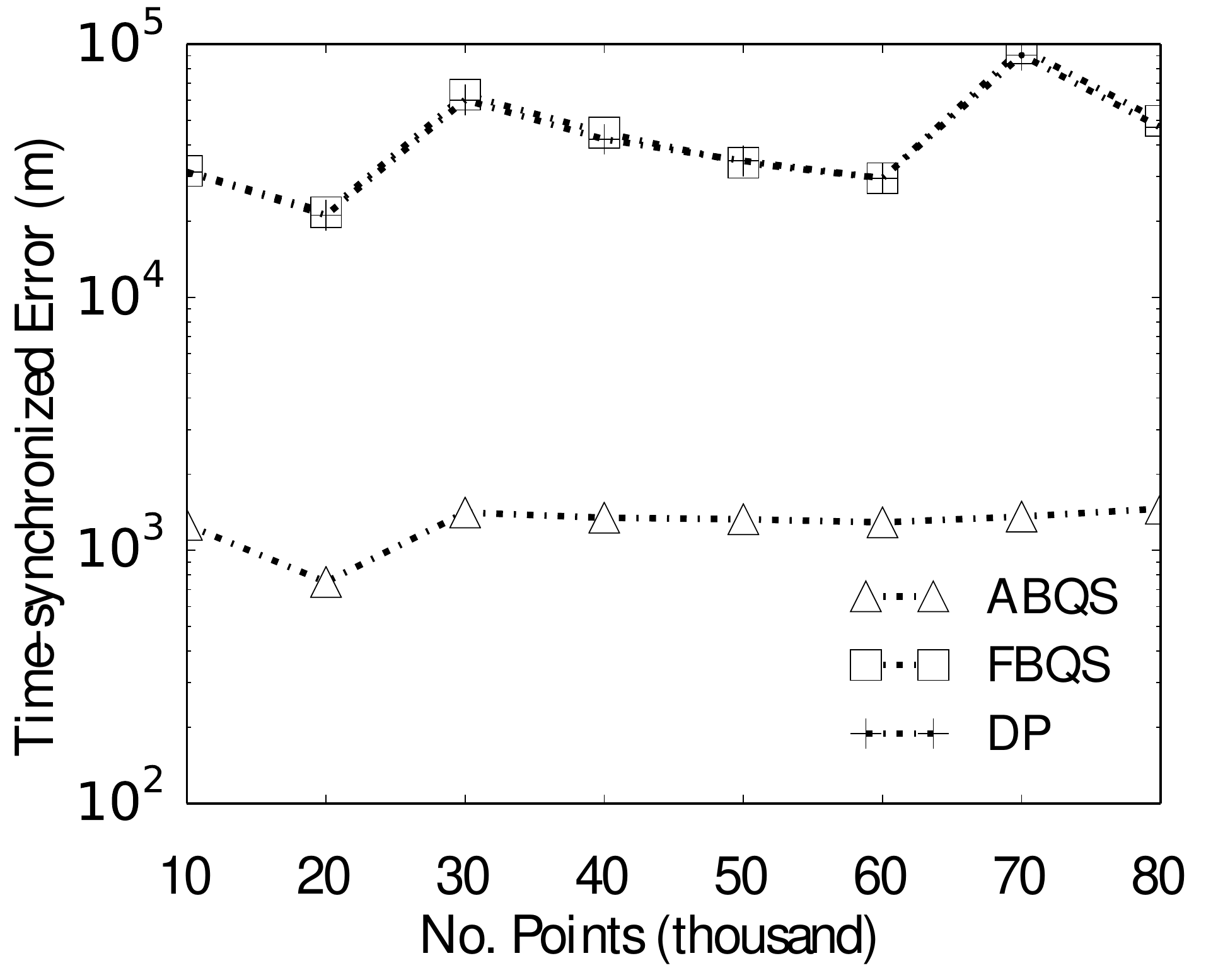}
}
\subfigure[]{ \label{fig:nem}
\includegraphics[height=3cm]{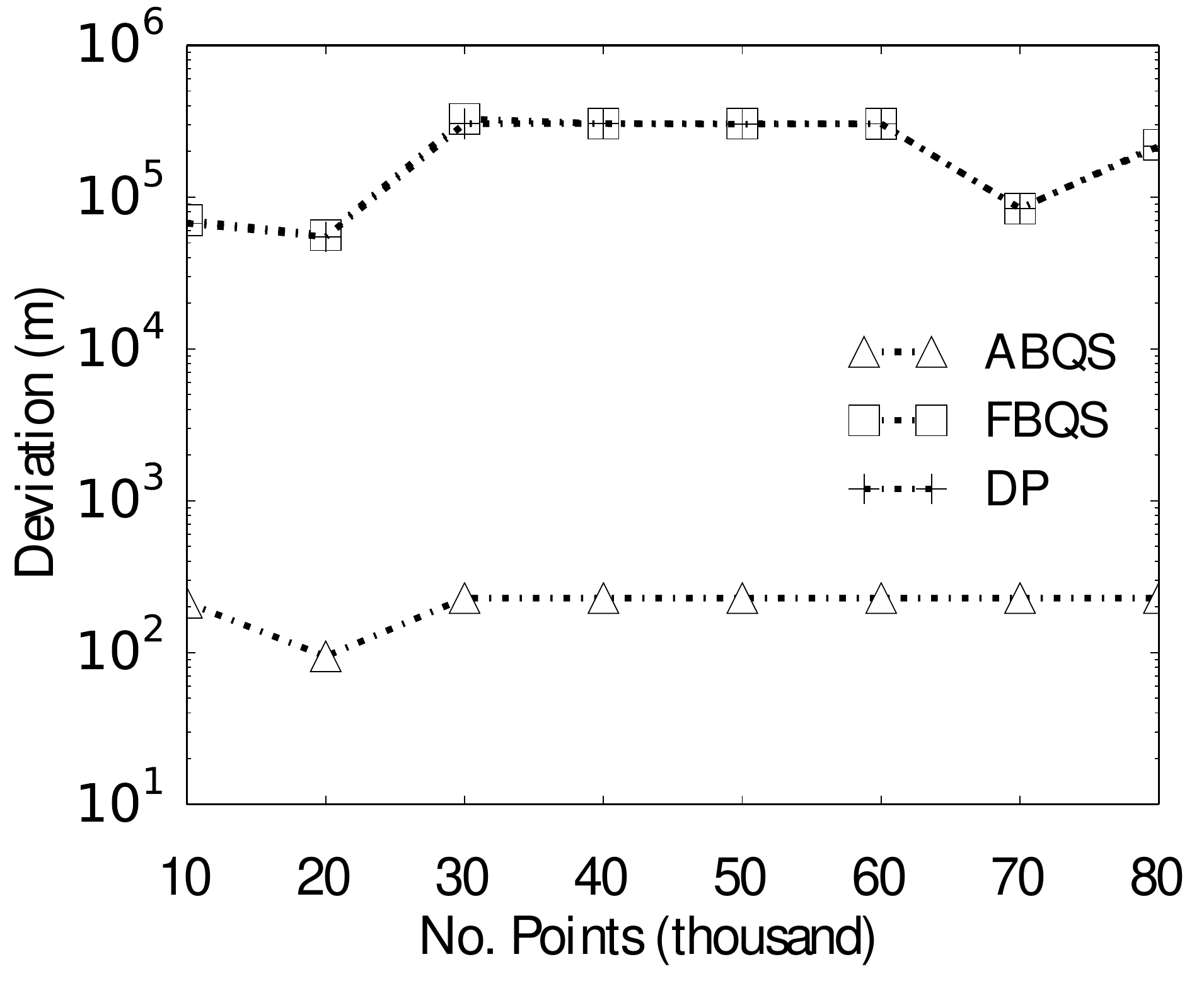}
}
\subfigure[]{ \label{fig:de}
\includegraphics[height=3cm]{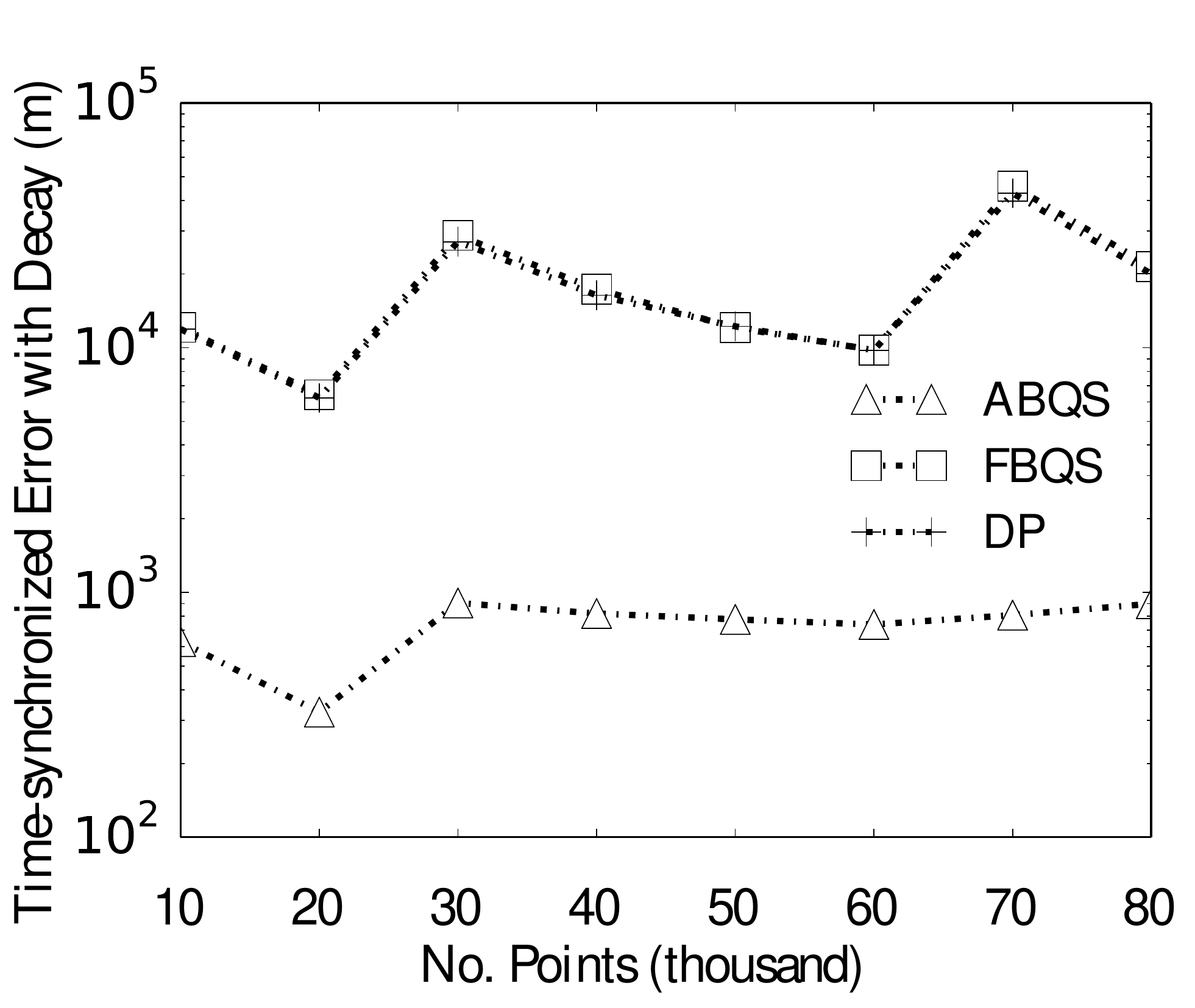}
}
\subfigure[]{ \label{fig:ds}
\includegraphics[height=3cm]{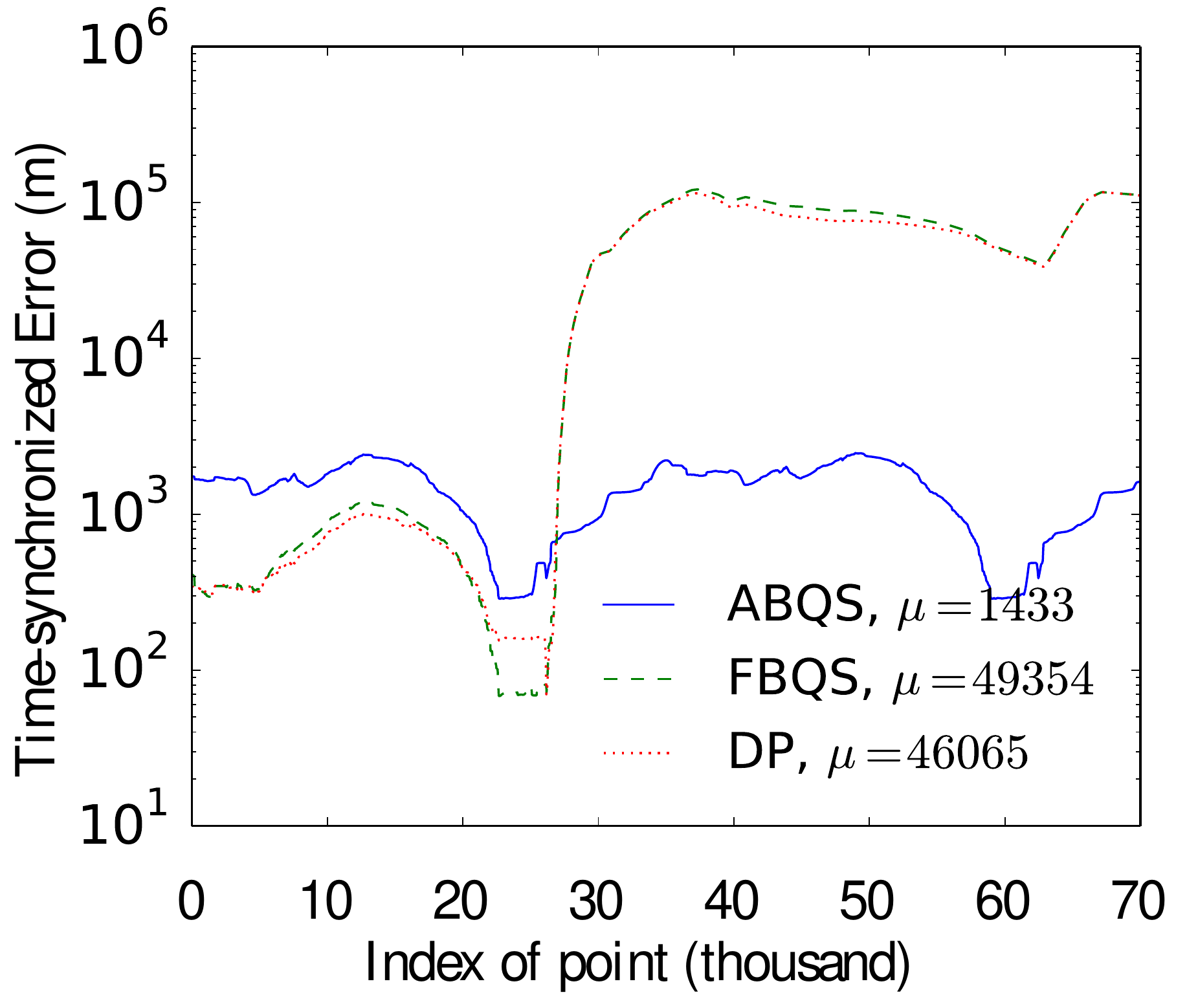}
}\\
\end{tabular}
\caption{Performance of ABQS for the Combined Datasets}
\vspace{-0.6cm}
\label{fig:perfa}
\end{figure*}

In the real-life application, the nodes also store other sensor 
information such as acceleration, heading, temperature,
humidity, energy profiling, sampled at much higher frequencies
due to their relatively low energy cost. We assume that of the 1MBytes
storage, GPS traces can use up to 50KBytes, and that the sampling
rate of GPS is 1 sample per minute. Each GPS sample requires
at least 12 bytes storage (latitude, longitude, timestamp). For the
error tolerance, we use 10 meters as it is reasonable for both
animal tracking and vehicle tracking. The average compression rate at 10
meters for both datasets is used for the algorithms. For the DR
algorithm, we assume it uses 39\% more points than FBQS as shown in
Figure \ref{fig:nps} at the same tolerance.

Given the set up, the compression rate and
estimated operational time without data loss for each algorithm are
listed in Table \ref{tbl:eot}. We can see a maximum 36\% improvement
from FBQS over the existing methods (60 v.s. 44), and a maximum 41\% improvement
from BQS (62 v.s. 44).

\begin{table}[htp]
\scriptsize 
\vspace{-0.2cm}
\centering
\caption{Estimated Operational Time}
\label{tbl:eot}
\begin{tabular}{ | c |c|c|c|c|c| }
\hline  & \textbf{BQS} & \textbf{FBQS} & \textbf{BDP}& \textbf{BGD}&
\textbf{DR}\\ 
\hline \textbf{Compression rate} & 4.8\% & 5.0\% & 6.65\%& 6.75\%& 6.65\%\\
\hline \textbf{Time (days)} & 62&60&45&44&45\\
\hline
\end{tabular}
\end{table}

\vspace{-0.2cm}
\subsubsection{Run Time Efficiency}
We compared the run time efficiency of FBQS, BDP
and BGD. 87,704 points from the empirical traces are used as the test
data. The error tolerance is set to 10 meters. For BDP and BGD, to
minimize the effect of the buffer size, we report their performances
with different buffer sizes, as in Table \ref{tbl:ct}:

\begin{table}[htp]
\scriptsize 
\centering
\caption{Performance Comparison with different buffer sizes}
\label{tbl:ct}
\begin{tabular}{ | c |c|c|c|c|c| }
\hline \textbf{Buffer size (points)} && 32 & 64& 128& 256 \\ \hline
\multirow{3}{*}{\textbf{Compression rate}} 
& \textbf{FBQS} & 3.6\%&--- &--- &--- \\
& \textbf{BDP} & 6.8\%& 6.7\%& 5.4\%& 4.9\% \\
& \textbf{BGD} &6\%& 4.8\%& 4.6\%&4.4\%\\ \hline
\multirow{3}{*}{\textbf{Run time (ms)}} 
& \textbf{FBQS} &99&--- &--- &--- \\
& \textbf{BDP} & 76& 101& 163& 292 \\
& \textbf{BGD} &182&285&446&628\\
\hline
\end{tabular}
\end{table}

There are two notice-able advantages of FBQS from the
comparison. Firstly, both the compression rate and the run time efficiency of 
FBQS algorithm are stable, independent of the buffer size setting. 
Secondly, it offers competitive run time efficiency while providing
leading compression rate. The only case when BDP is able to
outperform FBQS in run time efficiency is when the buffer is set to
32, where BDP has a far worse compression rate (89\% more points).

\vspace{-0.1cm}
\subsubsection{Evaluation of the ABQS Framework}
\label{sssec:expabqs}

We evaluate ABQS with the results presented in Figure \ref{fig:perfa}, in a series of experiments in comparison with FBQS and DP. Here instead of using the bat and vehicle datasets separately, we concatenate the data points from both datasets and put them into a unified timeframe. This introduces greater variations of the dynamics in the trajectory dataset, while the characteristics of both bats and vehicle are still inherited. The error multiplier $m$ is set to $2.5$ (as PBQS introduces $2\epsilon^{prev}$ uncertainty in both the lower and the upper bounds, this multiplier leaves $0.5\epsilon^{prev}$ margin for pruning). We mainly use two error metrics: deviation (maximum perpendicular distance as in Definition \ref{def:dev}) and  time-synchronized error (as in Figure 5(b) in \cite{DBLP:conf/edbt/MeratniaB04}).

In Figures \ref{fig:re} and \ref{fig:rem}, we fix the error tolerance $\epsilon=20$ (m) and the number of points $N=80,000$, 
and study how the errors change with the storage/data ratio $r$. 
This ratio represent the situation where the location data far exceeds the storage limit in volume. 
The smaller $r$ is, the greater challenge it is to control the compression error. For ABQS, it will continuously tradeoff precision for space, while for 
FBQS and DP, hard data loss (overwriting) is inevitable after the storage is full. We observe that in both error metrics, ABQS's performances
are significantly better than FBQS and DP. For example, as $r$ changes from 0.5\% to 5\%, the synchronized error and deviation of ABQS decreased from 15km to 1km, and from 15km to  30m.  On the other hand, FBQS and DP both have much worse results which have not shown clear reduction of compression error even when storage space is increased. For synchronized error their numbers drop from 40km merely to 20km, and for deviation the error lingers around 250km with fluctuation. Note that the standard deviations of the dataset are 30km and 85km on the $x$ and $y$ axis, hence a synchronized error of 40km and a deviation of 150km basically render the compression results useless for FBQS and DP. Such huge error is mainly introduced by the hard data loss over time. However, with ABQS, it does not only guarantee a decreasing deviation on each segment, but also obtains a reasonable and decreasing synchronized error as storage space grows.

Figures \ref{fig:ee} and \ref{fig:eem} demonstrate how the error tolerance $\epsilon$ influences the compression results. Here we set $N=80,000$ and the storage/data ratio $r=3\%$, and change  $\epsilon$ from 10m to 80m. We find that under both metrics, none of the algorithms is particularly sensitive to this parameter. Take Figure \ref{fig:ee} for example, FBQS and DP again have almost identical results (45km to 25km), while ABQS' number is varying between 1.5km to 2km. Again in both results we see that ABQS has far better performance. This implies that ABQS' performance is rather robust to the starting error tolerance, as it automatically increases the tolerance on demand. For FBQS and DP, the sychronized error decreased significantly when the tolerance changes from 10 to 30, this is because with a smaller tolerance the number of points after compression is greater and hence the ratio of the hard data loss is greater. Hence more information is lost even when the tolerance is smaller. After 30, the sychronized error stablizes around 25km. It is possible that when $\epsilon$ continues to grow, the gap between ABQS, FBQS, and DP will close (when trajectories size is smaller than storage after first compression pass), however in practice it is unrealistic to set an appropriate value for $\epsilon$ when the operational time is unknown.

Figures \ref{fig:ne} and \ref{fig:nem} study the effect of data volume. When $\epsilon=20$ and $r=3\%$, we change the number of points to be compressed from 10,000 to 80,000, and examine the errors. It is evident that ABQS has very consistent performance regardless of data size (around 1.5km synchronized  error and 200m deviation). However, with hard data losses, both FBQS and DP perform poorly, with consistently over 15 times greater synchronized errors and over 400 times greater deviation. These results show that ABQS is robust to the data size.

Figure \ref{fig:de} shows the changes of the synchronized error with a linear significant decay function (similar to the ``amnesic function'' in \cite{DBLP:conf/icde/PalpanasVKGT04}) that reduces the importance of points from 1.0 to 0.2 from younger to older data ($\widetilde{e^i}_{i=1:N}= e^i(\frac{0.8i}{N} + 0.2)$ ), as a 
simplistic demonstration of the significance function's effect in the relative performance for each method. The results verifiy that though the errors induced by older points (which are likely overwritten in FBQS and DP) are regarded less important, ABQS still wins by a large margin.

Finally, Figure \ref{fig:ds} presents how the synchronized errors are distributed over time for the $80,000$-point long empirical traces. Here the error is calculated for each individual point in the original trajectory, and is smoothed with a sliding window of 1,000 to improve visual clarity. The mean errors for ABQS, FBQS and DP are 1,433m, 49km and 46km, showing a clear advantage of ABQS. We also notice a clear cutoff point for FBQS and DP around 27,000 points. Before this point FBQS and DP have better precision than ABQS because the trajectories they store have gone through one compression. However, after the cutoff point where hard data loss begins to occur, the error jumped dramatically from 80m to 100km. Contrarily, ABQS' error remains much more stable during for the whole course. Indeed changing the initial error tolerance for FBQS and DP may as well change the location of the cutoff point slightly for FBQS and ABQS, but when the total number of points is unknown when deployed, ABQS has a clear advantage as it does not require a precise estimation of the final data size, or a carefully chosen error tolerance.

\vspace{-0.4cm}
\section{Conclusion}
\label{sec:conc}
We present a family of online trajectory
compression algorithms called BQS and an amnesic compression framework 
called ABQS for resourced-constrained environments. 
We first propose a convex-hull bounding structure and then show tight bounds 
can be derived from it so that compression decisions will be 
efficiently determined without actual deviation calculations. A light version
of the algorithm is hence proposed for the most constrained computation environments. 
\eat{A discussion is provided for the generalization and extensibility of
BQS for the 3-D space as well as for a different error
metric.} The standalone algorithms are then incorporated in an online framework that
manages a given storage and performs amnesic compression called ABQS.

\eat{
To further reduce the time and space complexity of the BQS algorithm,
a fast version of the BQS compression algorithm is also proposed. In
this version, error calculations are completely eliminated. Instead,
when uncertain of the error, the fast algorithm aggressively takes a
point. However due to the tight error bounded provided by the BQS,
the overhead in the compression rate is minimum, making it a
light-weight, efficient and effective algorithm, which is ideal for
 constrained computation environments.

As we establish the BQS algorithm, a discussion is also provided 
for the generalization and extensibility of
the BQS algorithm for the 3-D space as well as for a different error
metric. We have showed that such extensions are natural and
straightforward. BQS' flexibility to generalize to other applications
and settings is demonstrated.
}
To evaluate the proposed methods, we have collected empirical data using a low-energy 
tracking platform called Camazotz on both animals and vehicles. 
We also used synthetic dataset that 
is statistically representative of flying foxes' movement dynamics to improve data diversity. 
In the experiments we evaluate the framework in various aspects 
from pruning power, compression rate, run time efficiency, operational time,
compression error, etc. The proposed methods demonstrate significant advantages
in general. In some experiments it even achieves 15 to 400 times smaller error than its competitors.

\eat{
We examine the pruning power of the original BQS algorithm,
demonstrating that the great pruning power of BQS leaves an ideal
opportunity for FBQS to exploit so that further improvement on
the time and space efficiencies is achieved without sacrificing much
compression rate. We present the actual compression rates of
BQS and fast BQS, and compare them to the results of competitive
methods. Comparison of the estimated operational time with different
algorithms is also presented. We also study the robustness of BQS to 
trajectory shapes. Extensive evaluation The actual run time is also reported.

There are a few immediate extensions to this work. The excellent
performance of the BQS algorithms provides a unique opportunity to
develop online and individualized smart systems for long-term tracking. For
instance, merging and ageing can be used on the historical trajectory
data to further reduce storage space. Individualized trajectory
and waypoint discovery can also be used to facilitate advanced applications
like real-time trip prediction or trip-duration estimation. Exploring
the potential of a 4-D BQS could be another interesting extension to
this work.
}

\vspace{-0.3cm}
\small{
\bibliographystyle{IEEEtran}
\bibliography{main}
}
\eat{
\section{Appendix: Discussion and Proof of Theorems}
\label{sec:app}
First we explain the purpose of splitting the space into four
quadrants and the properties obtained by this setup.

\begin{itemize}
\item No two bounding lines will intersect with the same edge of a
  bounding box, and every edge will have exactly
 one intersection with the bounding lines, except at the corner points
 or on the axes.
\item The angle between $\overline{se}$ and either of the two bounding lines will be smaller than
  $90^\circ$, 
\item Hence, any point will be bounded by the convex hull formed by the
  points $\{l_1, l_2, u_1, u_2, c_n, c_f\}$. No bounding convex-hulls from two adjacent BQS will overlap. 
\end{itemize}

Splitting the space into four quadrants ensures that the
aforementioned properties hold. Otherwise, the only case that the convex hull
  $\overline{l_1 l_2 c_f u_2 u_1 c_n}$ may not cover all the
  points is when there is an edge with zero intersection from the
  bounding lines, as depicted in Figure \ref{fig:ebqs}. In this case,
  $c_f$ is  $c_3$, however $c_2$ has a greater deviation to $\overline{se}$
  than any point in $\{l_1, l_2, u_1, u_2, c_n, c_f\}$. Notice
  that when the system is in a single quadrant, the layout would not
  be a possibility.

\begin{figure}[htp]
\vspace{-0.3cm}
\hspace{1.8cm}
\includegraphics[height=2.5cm]{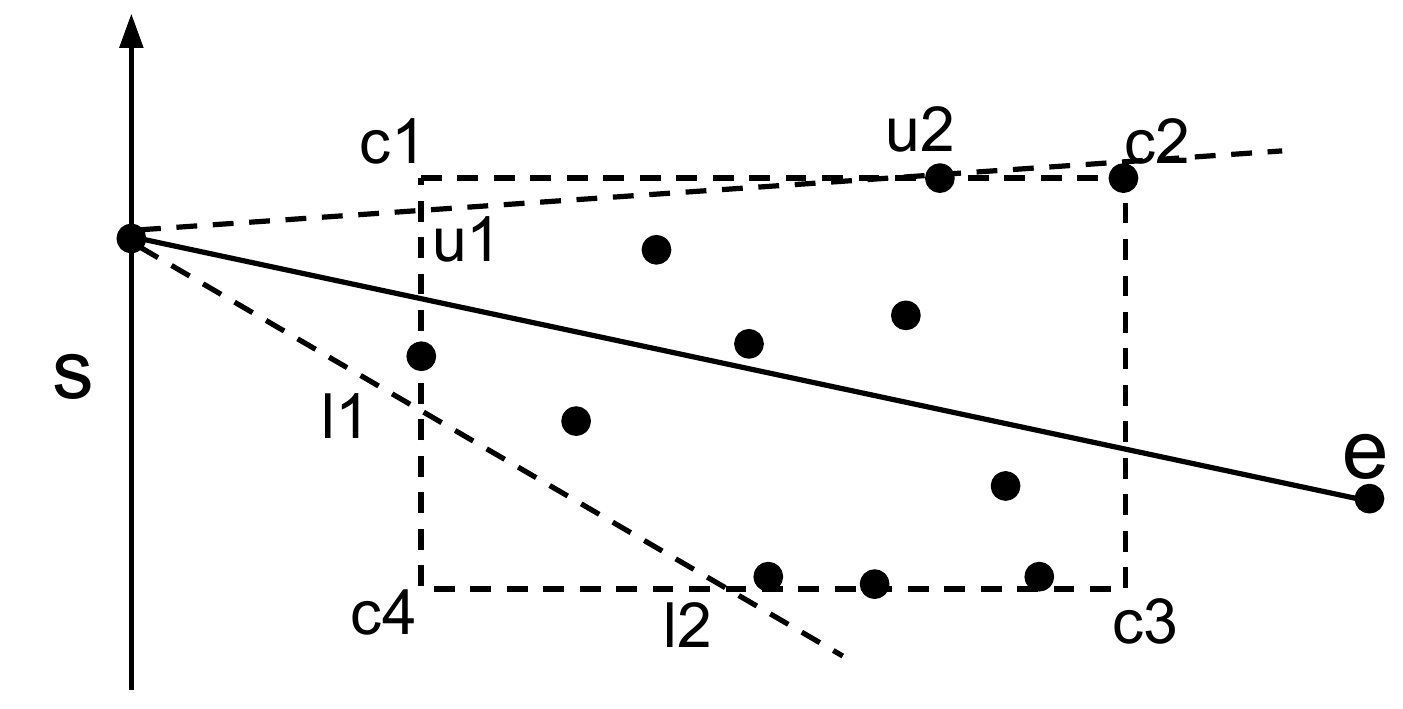}
\caption{Bounding system across quadrants}
\label{fig:ebqs}
\vspace{-0.5cm}
\end{figure}

Next we present the proof for Theorems
\ref{thm:1}, \ref{thm:2}, and \ref{thm:4}:
\paragraph{Proof of Theorem \ref{thm:1}} Using the first quadrant as
example as shown in Figure \ref{fig:bqs}, we
can state there are no points in areas $\overline{c_1u_1u_2}$ or
$\overline{c_3l_1l_2}$. Now we have the following properties:
\begin{itemize}
\item $d(u_1, \overline{se}) \leq d^{max}(p, \overline{se}) \leq d( u_2, \overline{se} )$ :
  Because the angle between $\overline{se}$ and $l_{u_1,u_2}$ is less than
  or equal to 90$^\circ$, if we extend $\overline{u_2c_2}$ to
  intersect with $\overline{se}$ at $p_1$, the three vertices form a
  triangle $\overline{u_2p_1s}$, so the greatest distance from
  any point in the triangle to the edge $\overline{se}$ would be $d(u_2,
  \overline{se})$. Then because the bounding line dictates that there must
  be at least one point (denoted as $p^{u}$) in the line segment $\overline{u_1, u_2}$, which does
  not intersect with line $\overline{se}$, we have $d(p^{u}, \overline{se}) \leq
  min\{ d(u_1, \overline{se}), d(u_2, \overline{se}) \}$. In this case we have
  $min\{ d(u_1, \overline{se}), d(u_2, \overline{se}) \}=d(u_1, \overline{se})$, hence we
  have $d(u_1, \overline{se}) \leq d^{max}(p_i, \overline{se}) \leq d( u_2, \overline{se} )$.
\item $d(l_1, \overline{se}) \leq d^{max}(p_i, \overline{se}) \leq d( l_2, \overline{se}
  )$ : proof is similar as above using the triangle $\overline{u_2p_1s}$. We could also see
  that by using $\overline{l_2p_1s}$ we have two triangles that
  together contain the convex hull
  $\overline{c_4u_1u_2c_2l_2l_1}$ which bounds the points, so we have $d^{max}(p_i, \overline{se}) \leq  max\{
  d^{intersection} \}$ as the upper bound. Similarly, we can have $max\{ d(u_1, \overline{se}),
  d(l_1, \overline{se}) \}  \leq d^{max}(p_i, \overline{se}) $ as the lower
  bound.

\item $d^{max}(p, \overline{se}) \geq max(d^{corner-nf})$: This property is based on the fact that
  the line $\overline{se}$ will not intersect with both edges that a corner
  point is on except at the corner points, while every edge must have at least one point on it.
  In this quadrant, $c_n$ is $c_4$ and $c_f$ is $c_2$. Because there is at least one point on
  the edge $\overline{c_1c_2}$, and $c_2$ is the closest point to
  $\overline{se}$ on $\overline{c_1c_2}$, we have $d^{max}(p, \overline{se}) \geq d(c_2,
  \overline{se})$. With an identical case of edge $\overline{c_1c_4}$, we
  have $d^{max}(p, \overline{se}) \geq max(d^{corner-nf})$ when we combine the
  lower bounds.
\end{itemize}

The line $\overline{se}$ may intersect with the bounding box in different
angle and at different locations but with the properties guaranteed
by the BQS, we can use the same proof for all cases.

Theorems \ref{thm:2} is proven with similar techniques.
Theorem \ref{thm:4}, i.e. cases in which the line $\overline{se}$ is in a different quadrant from
the bounding box, can
be proven using the same proof as for Theorem \ref{thm:0}.
}


\section*{Authors' Biographies}
{\footnotesize
\vspace{-1.5cm}
\begin{IEEEbiography}[{\includegraphics[width=1in,height=1.25in,clip,keepaspectratio]{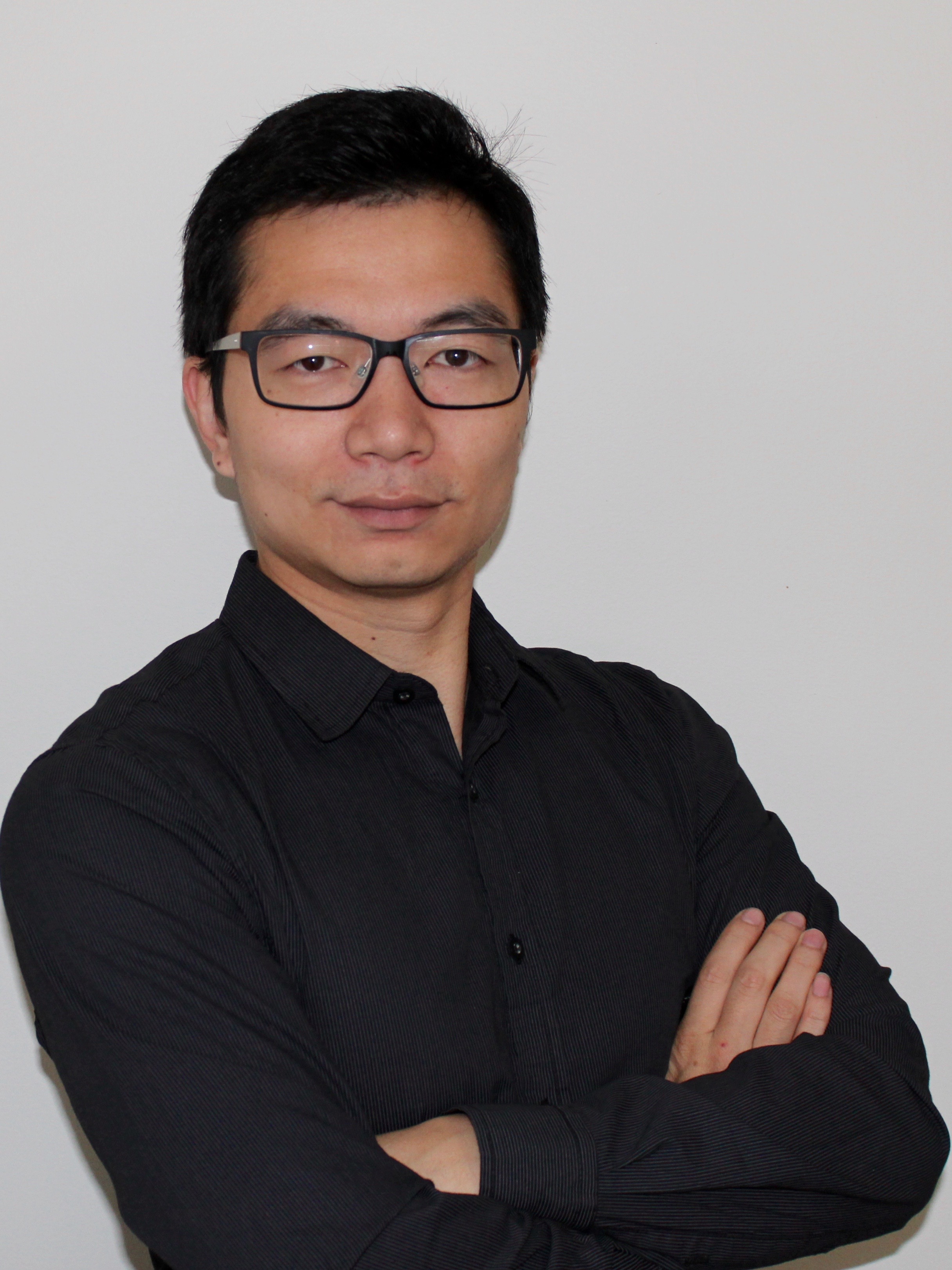}}]%
{Jiajun Liu} is an Associate Professor at Renmin University of China. He received his PhD and BEng from The University of Queensland, Australia and from Nanjing University, China in 2012 and 2006 respectively. Before joining Renmin University he has been a Postdoctoral Fellow at the CSIRO of Australia from 2012 to 2015. From 2006 to 2008 he also worked as
a Researcher/Software Engineer for IBM China
Research/Development Labs. His main research interests are in multimedia and spatio-temporal data management and mining. 
He serves as a reviewer for multiple journals such as VLDBJ, TKDE,
TMM, and as a PC member for ACM MM and CCF Big Data.
\end{IEEEbiography}
\vspace{-1cm}
\begin{IEEEbiography}[{\includegraphics[width=1in,height=1.25in,clip,keepaspectratio]{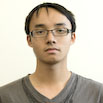}}]%
{Kun Zhao} is a Postdoctoral Fellow at CSIRO Australia. He obtained his PhD from Northeastern University, USA and his research interests include network science, mobility  data analysis in sensor networks.
\end{IEEEbiography}
\vspace{-1cm}
\begin{IEEEbiography}[{\includegraphics[width=1in,height=1.25in,clip,keepaspectratio]{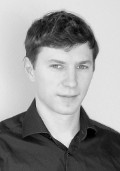}}]%
{Philipp Sommer} is a Research Scientist at ABB Corporate Research in Switzerland. His research interests include a broad range of topics in the field of wireless sensor networks, distributed computing and embedded systems. He received MSc and PhD degrees in Electrical Engineering from the Swiss Federal Institute of Technology (ETH) in 2007 and 2011 respectively. He has been a postdoctoral fellow at the Autonomous System lab of the CSIRO, Australia between Sept. 2011 and Oct. 2014.\end{IEEEbiography}
\vspace{-1cm}
\begin{IEEEbiography}[{\includegraphics[width=1in,height=1.25in,clip,keepaspectratio]{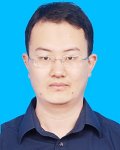}}]%
{Shuo Shang} is currently a faculty member of Computer Science at China University of Petroleum-Beijing, P.R.China. He was a Research Assistant Professor level Research Fellow with Department of Computer Science, Aalborg University, and was a faculty member of the Center for Data-intensive Systems (Daisy), which conducts research and offers education with a focus on data management for various data-intensive systems.\end{IEEEbiography}
\vspace{-1cm}
\begin{IEEEbiography}[{\includegraphics[width=1in,height=1.25in,clip,keepaspectratio]{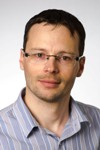}}]%
{Brano Kusy} is a Principal Research Scientist in CSIRO's Autonomous Systems Program in Brisbane, Australia. CSIRO is Australia's leading scientific research organization that operates across more than 50 sites in Australia. Autonomous systems program has a diverse group of researchers and engineers that investigate new technology in wireless sensor networks and field robotics, including autonomous land, sea, and air vehicles. He is also an adjunct senior lecturer at the University of Queensland.\end{IEEEbiography}
\vspace{-1cm}
\begin{IEEEbiography}[{\includegraphics[width=1in,height=1.25in,clip,keepaspectratio]{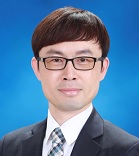}}]%
{Jae-Gil Lee} is an Associate Professor at KAIST and he is leading Data Mining Lab. Before that, he was a Postdoctoral Researcher at IBM Almaden Research Center and a Postdoc Research Associate at Illinois at Urbana-Champaign. His research interests encompass spatio-temporal data mining, social-network and graph data mining, and big data analysis with Hadoop/MapReduce.
\end{IEEEbiography}
\vspace{-1cm}
\begin{IEEEbiography}[{\includegraphics[width=1in,height=1.25in,clip,keepaspectratio]{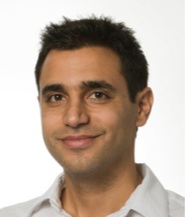}}]%
{Raja Jurdak} is a Principal Research Scientist and leads the Distributed Sensing Systems Group at CSIRO. He received the B.E. degree from the American University of Beirut in 2000, and the MSc and Ph.D. degrees from the University of California Irvine in 2001 and 2005 respectively. His current research interests are around energy and mobility in sensor networks. 
\end{IEEEbiography}

}
\end{document}